\documentclass[11pt, a4paper]{article}
\usepackage{amssymb}
\usepackage[english]{babel}
\usepackage{amsmath,amsthm,amssymb,amsfonts}
\usepackage[colorlinks = true]{hyperref}
\usepackage{makeidx}
\usepackage[pdftex]{graphicx, xcolor}
\usepackage[all]{xy}
\usepackage{extarrows}
\usepackage[new]{old-arrows}
\usepackage{cancel}
\usepackage[top=3.80cm, bottom=2cm, left=2cm, right=2cm]{geometry}
\usepackage[utf8]{inputenc}
\usepackage[T1]{fontenc}
\usepackage{tikz}
\usetikzlibrary{shapes.geometric, arrows.meta, calc, positioning}

\usepackage{needspace}
\allowdisplaybreaks

\theoremstyle{plain} 
\newtheorem{thm}{Theorem}[section] 
 
\newtheorem{lem}[thm]{Lemma} 
\newtheorem{prop}[thm]{Proposition} 
\newtheorem{conj}[thm]{Conjecture}
\newtheorem*{theorem*}{Theorem} 
\newtheorem*{prop*}{Proposition} 

\renewenvironment{proof}{{\noindent\bfseries Proof.

}}{\qed \medskip}

\theoremstyle{definition} 
\newtheorem{defn}[thm]{Definition}

\theoremstyle{definition} 
\newtheorem{oss}[thm]{Remark} 
\newtheorem{ex}[thm]{Example}

\title{\Large{\bfseries A constructive approach to generalized principal connections}}
\author{\normalsize Lorenzo Fatibene
\\ 
{\footnotesize Department of Mathematics, University of Torino}
\\
{\footnotesize via Carlo Alberto 10, 10123 Torino, Italy}
\\
{\footnotesize Istituto Nazionale di Fisica Nucleare (INFN), Sezione di Torino}
\\
{\footnotesize via Pietro Giuria 1, 10125 Torino, Italy}
\\
{\footnotesize e--mail: {\sc lorenzo.fatibene@unito.it}} 
\\
\\
\normalsize Hartwig Winterroth\thanks{Corresponding author.}
\\
{\footnotesize Department of Mathematics, University of Torino}
\\
{\footnotesize via Carlo Alberto 10, 10123 Torino, Italy}
\\
{\footnotesize e--mail: {\sc hartwig.winterroth@unito.it}}
}
\date{}
\begin{document}

\maketitle

\begin{abstract}

\noindent We address the recently introduced notions of {\em generalized principal bundle} and {\em generalized principal connection} by keeping track of global geometric properties through local coordinate transformation laws. This approach leads us to introduce generalized principal bundle coordinates and to find their transformation laws. Besides, we show that any {\em Lie group fiber bundle} (and hence, in particular, any vector bundle) is a generalized principal bundle and we give a proof of the fact that any Lie group fiber bundle with connected typical fiber is an associated bundle to a suitable principal bundle. Moreover, we present a direct way to characterize Lie group fiber bundle connections and generalized principal connections in terms of horizontal lifts and of local conditions. Finally, we recover in our setting some already known results, including that generalized principal connections are associated only to Lie group fiber bundle connections and that they reduce to usual principal connections on standard principal bundles. Our results are needed in order to understand how generalized principal connections might fit in the fiber bundle treatment of classical field theories, aiming towards a notion of generalized gauge theory.
\end{abstract}


\medskip


\section{Introduction}

Principal bundle theory is the cornerstone of the modern approach to {\em gauge theories}, one of the main breakthroughs of the last century regarding deep links between mathematics and physics. The path which led to understand that the known fundamental interactions in physics are governed by gauge symmetries is usually set to start with the much celebrated article \cite{YANG (1954)} published by Yang and Mills in 1954, although the following words of Yang underline the often neglected fact that fiber bundle theory was not the initially background in which gauge theories were constructed:
\begin{quotation}
\begin{footnotesize} 
\noindent That non-Abelian gauge fields are conceptually identical to ideas in the beautiful theory of fiber bundles, developed by mathematicians {\em without reference to the physical world}, was a great marvel to me. In 1975, I discussed my feelings with Chern [Shiing Shen Chern (1911--2004)], and said, \lq\lq This is both thrilling and puzzling, since you mathematicians dreamed up these concepts out of nowhere.\rq\rq\,\,He immediately protested, \lq\lq No, no, these concepts were not dreamed up. They were natural and real.\rq\rq

\noindent (Chen Ning Yang in {\em Magnetic monopoles, fiber bundles, and gauge fields}, 1977, \cite{YANG (1977)})
\end{footnotesize} 
\end{quotation}
This is even more surprising from the perspective that we have today, since the rigorous approach (shared for example by \cite{BAEZ (1994)}, \cite{BLEECKER (1981)}, \cite{FATIBENE (2003)}, \cite{GIACHETTA (2009)}, \cite{KOLAR (1993)}, \cite{KRUPKA (1973)} and \cite{MARATHE (1992)}) to classical field theories (which comprise gauge theories without quantizations) is the one constituted by fiber bundle theory together with the jet prolongation formalism, i.e.\,\,the relevant field equations are formulated using {\em variational calculus} on fiber bundles and their jet prolongations.

\medskip

Following the spirit behind the answer of Chern to Yang, it would be then quite reasonable to explore the possibility of developing an instance of generalized gauge theories by investigating the notions of {\em generalized principal bundle} and {\em generalized principal connection} as recently introduced by Castrillón López and Rodríguez Abella in \cite{CASTRILLON LOPEZ (2023)}. This idea should be carried out obviously in the present-day mathematical framework for classical field theories and its reasonableness must be retraced in the prospect of finding a potentially unifying language for different kinds of covariant field theories, new geometrical structures or even hints of new physics.

At this stage we are not concerned with applications, but more with whether these concepts are \lq\lq natural and real\rq\rq\,\,and we think that the new contributions appeared in \cite{CASTRILLON LOPEZ (2023)} have this connotation, as this work should also prove.

\medskip

For now, in this paper we shall review the theory of generalized principal bundles and connections through a complementary {\em constructive approach}. By this we mean an approach arising from the {\em fiber bundle construction theorem} (see for example \cite[Chapter 16, Section 13]{DIEUDONNE' (1972)}), which tells us what are the minimal data that help us to {\em construct} a fiber bundle. At length, this approach can be fruitfully reduced to an ample usage of local coordinate techniques and this is how we will proceed in practice. Hence, this will be a rather technical (but nevertheless crucial) result filled preliminary paper, while the forthcoming paper \cite{FATIBENE (2026)} will actually start the development of generalized gauge theories, in the sense presented above.

\medskip

With respect to the last paragraph, we should undertake a short digression. The general structure of our paper will be the following: to give definitions in global language regarding global objects, to rewrite these definitions in local language and then to prove theorems using local coordinates, keeping track of the relevant transformation laws or, in other words, of the globality present in the background.

We have chosen to use this kind of procedure because working in local coordinates gives in general, on the one hand, the opportunity to easily see the relations between the mathematical objects at the center of the discussion and, on the other hand, the possibility to compare the obtained results with the past literature in a more straightforward way. Nonetheless, this situation can cause a bit of suspicion since it is common sense that in mathematics one should use global languages as long as it is possible, especially in differential geometry. 

\medskip

As a matter of fact, these two approaches are equally right. For example in any classical field theory the usage of the fiber bundle framework (global side) is a mean to trace covariance of the transformation laws of the corresponding fields (local side). By not declaring transformation laws for fields we would obtain just a local field theory which means, at a physical level, that we would get the description of some events from one specific observer without knowing how to relate observations from different laboratories. Hence {\em absolute knowledge}, or {\em globality}, must be found through the transformation laws. In other words, {\em the geometry is in the transformation laws}. The last sentence is the core concept of what we denote by constructive approach.

\medskip

Moreover, since in applications to physics coordinates (and trivializations) are to be identified with observers and their readings, when one is required to theoretically model experiments and observations, coordinates are needed eventually.

\medskip

Returning to the contents of the paper, after recollecting some useful results from the theory of principal bundles in Section \ref{sec 1.1}, we subsequently summarize in Section \ref{sec 1.2} the main results about Lie group fiber bundles and fibered actions appeared in the article by Castrillón López and Rodríguez Abella, adding some new remarks and details. In particular, we carefully review the Generalized Quotient Manifold Theorem (see Theorem \ref{thm 1.12}) because of its significance for the remainder of our paper. For the historical background of Lie group fiber bundles we refer to \cite{CASTRILLON LOPEZ (2023)} and the references therein, noting also that a different approach that addresses these fiber bundles as particular cases of Lie groupoids has appeared in \cite{BLAZQUEZ-SANZ (2022)}.

In Section \ref{sec 1.3} we study generalized principal bundles. Through Example \ref{ex 1.3.2}, we notably prove that any Lie group fiber bundle is a generalized principal bundle and deduce from this result that vector bundles and adjoint bundles are examples of generalized principal bundles. Following our constructive approach to the topic, we introduce generalized principal bundle coordinates, i.e.\,\,fibered coordinates adapted to the peculiar structure of generalized principal bundles, with their transformation laws (see \eqref{11}). Besides, we also prove that any Lie group fiber bundle, with connected typical fiber, is an associated bundle to a certain standard principal bundle (see Proposition \ref{prop 1.3.4}). We can note here that the recent paper \cite{CATTAFI (2026)} provides a notion of principal bundle groupoid, or {\em PB-groupoid}, centered on Lie 2-groupoid actions, which recovers as a particular case the notion of generalized principal bundle. However, a discussion of this framework lies beyond the scope of the present work and will hence not be pursued here.

In Section \ref{sec 2.1} we investigate Lie group fiber bundle connections presenting a straightforward way to characterize Lie group fiber bundle connections in terms of horizontal lifts (see Proposition \ref{prop 6.3}), alternative to the one given in \cite{CASTRILLON LOPEZ (2023)}. We subsequently use this result in order to find the local coordinate characterization of Lie group fiber bundle connections (see Theorem \ref{thm 6.4}). In Appendix \ref{app A} we explicitly check the invariance of these local conditions under change of coordinates. After that, in Example \ref{ex 2.2.2} we give a proof of the fact that the Lie group fiber bundle connections on a fixed vector bundle are exactly the linear connections on it.

In Section \ref{sec 2.3} we examine generalized principal connections (while in Appendix \ref{app B} a heuristic comparison with usual principal connections is given) presenting again a straightforward way to characterize generalized principal connections in terms of horizontal lifts (see Proposition \ref{prop 7.4}). We use this result in order to find as well the local coordinate characterization of generalized principal connections (see Theorem \ref{thm 7.6}). In Appendix \ref{app C} we explicitly check the invariance of these local conditions under change of coordinates. After that, in Example \ref{ex 2.4.2} we give a proof of the fact that the generalized principal connections on a fixed vector bundle are exactly the affine connections on it. In the meantime, we also exhibit a new proof (derived from the aforementioned local conditions) of the fact that generalized principal connections are associated only to Lie group fiber bundle connections (the original proof of this result can be found still in \cite{CASTRILLON LOPEZ (2023)}).

In Section \ref{sec 2.5} we provide a new, direct and short proof of the fact that usual principal bundles and connections are indeed particular cases of the generalized objects (see Theorem \ref{thm 2.5.1}), another key result present in \cite{CASTRILLON LOPEZ (2023)}.

Finally, in Section \ref{sec 2.6} we draw our conclusions, describing with particular attention what are the next steps in the planned path towards generalized gauge theories. Specifically, we offer insights into our prospective results and goals, mentioning a preliminary example for illustrative purposes.

\medskip

Although we will not go along the following lines of study, it is worth noting that the original interests of Castrillón López and Rodríguez Abella in generalized principal connections lay in their application to the formulation of a Lagrangian reduction theory for covariant field theories with gauge symmetries (as they eventually carried out in \cite{CASTRILLON LOPEZ (2024)}).

\medskip

In the following, we suppose that all manifolds are real, of finite dimension and smooth.  We denote as $(P, M, \pi_{P, M}, F)$ a fiber bundle (which will be always locally trivial for us) such that $P$ is the total space, $\mathrm{dim}(P)=m+n$, $M$ is the base space, $\mathrm{dim}(M)=m$, $F$ is the standard fiber, $\mathrm{dim}(F)=n$, and $\pi_{P, M} \colon P \rightarrow M$ is the projection. Moreover, the notation $(x^\mu, y^i)$ will stand for a system of fibered coordinates on $P$, where $\mu \in\{1, \dots, m\}$ and $i \in\{1, \dots, n\}$. For more details on the general theory of fiber bundles we refer for example to \cite[Chapter 1]{SAUNDERS (1989)}. The Einstein summation convention for repeated indices regarding local coordinates is assumed.

We point out as well that, given two manifolds $M$ and $N$, we denote the tangent map of a smooth map $f \colon M \rightarrow N$ at a point $x \in M$ by $T_x(f) \colon T_{x}(M) \rightarrow T_{f(x)}(N)$.

When we need to compute the partial derivatives of a function of several variables, for example $f=f(x, y)$, $\partial_{x}f$ or $\frac{\partial f}{\partial x}$ stands for the partial derivative function with respect to the $x$ argument and $\partial_{x}f(\xi, \eta)$ or $\frac{\partial f}{\partial x}(\xi, \eta)$ stands for its value at the point $(\xi, \eta)$. As a clarification, this means for example that if $f(x, y)=x^2+y$, then $\partial_{x}f(\xi, \eta)=\frac{\partial f}{\partial x}(\xi, \eta)=2\xi$, while $\frac{\partial}{\partial x}f(\xi, \eta)=\frac{\partial}{\partial x}[\xi^2+\eta]=0$, i.e.\,\,the preceding notation is introduced in order to take into account the non-commutativity between the operation of partial derivation and the one of evaluation.

\section{Standard principal bundles} \label{sec 1.1}

In the first place we recall some basic results from the theory of principal bundles that will be useful below. Recall in particular that one of the equivalent ways to define principal bundles, crucially capturing the bare minimum data in order to do so, is the following:
\begin{defn}\label{defn 1.4}
A {\em (standard) principal bundle with structure Lie group} $G$ is a manifold $P$ equipped with a smooth right action of a Lie group $G$:
\begin{align*}
m \colon P \times G \longrightarrow P \colon (p, g) \longmapsto m(p, g)
\end{align*}
such that:
\begin{enumerate}
\item $m$ is {\em free}, which means that if $m(p, g)=p$, for some $p \in P$, then $g=e$, where $e$ is the identity element of $G$.
\item The orbit space $S=P/G$ has a smooth manifold structure with the property that the canonical projection $\pi_{P, S} \colon P \longrightarrow S$ is a smooth submersion {\em or}, equivalently, $m$ is {\em proper}, that is the (smooth) map:
\begin{align*}
\Theta \colon P \times G \longrightarrow P \times P \colon (p, g) \longmapsto \Big(p, m(p, g)\Big)
\end{align*}
is a proper map, meaning that for every compact set $K \subseteq P \times P$, the preimage $\Theta^{-1}(K) \subseteq P \times G$ is compact.
\end{enumerate}
\end{defn}

Using this definition we get that $(P, S, \pi_{P, S}, G)$ is a fiber bundle (see Theorem \ref{thm 1.2} below) and that the action $m$ is vertical (sends fibers into themselves) and transitive on the fibers (since the fibers are the orbits of $m$).

\medskip

This way to define principal bundles is the more suited one in order to understand the generalized notion of principal bundle that we will review in Section \ref{sec 1.3} and is based on the well known {\em Quotient Manifold Theorem} and also on a partial converse to it:
\begin{thm}[Quotient Manifold Theorem]\label{thm 1.2}
Let $m \colon P \times G \longrightarrow P$ be a smooth, free and proper right action of a Lie group $G$ on a manifold $P$. 
Then we have that:
\begin{enumerate}
\item $S=P/G$ is a topological manifold.
\item $\mathrm{dim}(S)=\mathrm{dim}(P)-\mathrm{dim}(G)$.
\item $S$ has a unique smooth structure with the property that the canonical projection $\pi_{P, S} \colon P \longrightarrow S$ is a smooth submersion.
\item $(P, S, \pi_{P, S}, G)$ is a fiber bundle.
\end{enumerate} 

\end{thm}

\begin{prop}\label{prop 1.3}
Let $m \colon P \times G \longrightarrow P$ be a smooth and free right action of a Lie group $G$ on a manifold $P$ such that $S=P/G$ has a smooth structure with the property that the canonical projection $\pi_{P, S} \colon P \longrightarrow S$ is a smooth submersion. Then the action is proper (observe that in the present hypotheses then Theorem \ref{thm 1.2} holds, therefore the smooth structure on $S$ must be the unique one from the Quotient Manifold Theorem).

\end{prop}
For a proof of the Quotient Manifold Theorem and to see more details on this particular topic compare with \cite[Chapter 21]{LEE (2013)}. These results also lead to the more common way to define principal bundles as fiber bundles $(P, M, \pi_{P, M}, G)$ equipped with a free smooth right action of $G$ on $P$, with the requests that the orbit space $S=P/G$ has a smooth manifold structure coinciding with $M$ and that the canonical projection $\pi_{P, S} \colon P \rightarrow S$ coincides with the projection $\pi_{P, M} \colon P \rightarrow M$. 

\medskip

Recall also that on a principal bundle $P$ it is possible to consider {\em (standard) principal bundle coordinates} $(x^\mu, g^I)$, which are characterized by the transformation laws:
\begin{align}\label{100}
                 \left\{
		\begin{array}{l}
			x'^\mu=x'^\mu(x)
			\\
			g'^I=\pi^I\Big(\varphi(x), g\Big)
		\end{array}
		\right.	
\end{align}
where $\pi^I$ are the local expressions of the multiplication in the Lie group $G$ and $\varphi$ corresponds to local functions characterizing the principal bundle. For more details on principal bundle coordinates see for example \cite[Chapter 18, Section 1]{FATIBENE (2024)}.

\section{Lie group fiber bundles and fibered actions} \label{sec 1.2}

Starting with this section and going until the end of this work, we will review some aspects of and some theorems about the new geometric tools (that is the theory of generalized principal bundles and connections) introduced by Castrillón López and Rodríguez Abella (with complete exposition in \cite{CASTRILLON LOPEZ (2023)} and quickly resumed in \cite{CASTRILLON LOPEZ (2024)}). For the reader’s benefit, given that the cited articles are relatively new, our account on the already known matter will be exhaustive.  At the same time we will point out some new remarks and examples, add certain details and prove some new results.

\medskip 

Before discussing generalized principal bundles, there is the need to define the notion of {\em Lie group fiber bundle}:
\begin{defn}\label{defn 1.5}
A {\em Lie group fiber bundle with typical fiber} $G$ is a fiber bundle $(\mathcal{G}, M, \pi_{\mathcal{G}, M}, G)$, where $G$ is a Lie group, such that:
\begin{enumerate}
\item $\forall \, x \in M$, the fiber $\mathcal{G}_x=(\pi_{\mathcal{G}, M})^{-1}(x)$ is equipped with a (canonical) Lie group structure.
\item $\forall \, x \in M$, there is a neighbourhood $x \in U \subseteq M$ and a local trivialization:
\begin{align*}
\psi \colon (\pi_{\mathcal{G}, M})^{-1}(U) \longrightarrow U \times G
\end{align*}
(a diffeomorphism with $p_1 \circ \psi=\pi_{\mathcal G, M}$, where $p_1 \colon U \times G \longrightarrow U$ is the projection on the first factor) preserving the Lie group structure fiberwise, that is:
\begin{align*}
\psi|_{\mathcal{G}_x} \colon \mathcal{G}_x \longrightarrow \{x\} \times G \equiv G
\end{align*}
is a Lie group isomorphism (in particular given $e_x$, the identity element of the Lie group $\mathcal{G}_x$, we must have $\psi|_{\mathcal{G}_x}(e_x)=(x, e)$, where $e$ is the identity element of $G$).
\end{enumerate}
\end{defn}
\begin{oss}\label{oss 1.6}
A natural consequence of these assumptions is that the transition functions of a Lie group fiber bundle $(\mathcal{G}, M, \pi_{\mathcal{G}, M}, G)$ take values in $\mathrm{Aut}(G)$, the set of Lie group automorphisms of $G$. Then we can choose fibered coordinates $(x^\mu, g^I)$ on $(\mathcal{G}, M, \pi_{\mathcal{G}, M}, G)$ such that a general change of these fibered coordinates must satisfy:
\begin{align}\label{1}	
		\left\{
		\begin{array}{l}
			x'^\mu=x'^\mu(x)
			\\
			g'^I=G^I(x, g)
		\end{array}
		\right.	
\end{align}
with:
\begin{align}
G^I(x, e)&=e^I \label{2}	
\\
G^I\Big(x, \pi(g, h)\Big)&=\pi^I\Big(G(x, g), G(x, h)\Big) \label{3}	
\end{align}
where $e$ still indicates the identity element of $G$ and $\pi \colon G \times G \rightarrow G$ is the multiplication in the Lie group $G$. The symbol $\cdot$ will also indicate this multiplication. We also set from now on $\mathrm{dim}(G)=l$.

Note that a Lie group fiber bundle is {\em not} an example of principal bundle (just compare conditions \eqref{2} and \eqref{3} with the transformation laws characterizing principal bundles \eqref{100}).
\end{oss}

\begin{ex}\label{ex 1.2.3}
{\em Vector bundles} are examples of Lie group fiber bundles with typical fiber $(\mathbb{R}^l, +)$. Namely if $(E, M, \pi_{E, M}, \mathbb{R}^l)$ is a vector bundle, then by definition:
\begin{enumerate}
\item $\forall \, x \in M$, $E_x=(\pi_{E, M})^{-1}(x)$ is a vector space, so that it is equipped with a Lie group structure.
\item $\forall \, x \in M$, there is a neighbourhood $x \in U \subseteq M$ and a local trivialization:
\begin{align*}
\psi \colon (\pi_{E, M})^{-1}(U) \longrightarrow U \times \mathbb{R}^l
\end{align*}
such that:
\begin{align*}
\psi|_{E_x} \colon E_x \longrightarrow \{x\} \times \mathbb{R}^l \equiv \mathbb{R}^l
\end{align*}
is a vector space isomorphism and, as a consequence, a diffeomorphism. It is then in particular a Lie group isomorphism between $(E_x, +)$ and $(\mathbb{R}^l, +)$.
\end{enumerate}
\end{ex}

\begin{ex}\label{ex 1.2.4}
The {\em adjoint bundle} $(Ad(P), M, \pi_{Ad(P), M}, G)$ of a given standard principal bundle $(P, M, \pi_{P, M}, G)$ is a main instance of Lie group fiber bundle, as already observed in \cite{CASTRILLON LOPEZ (2023)}. We stress that here by \lq\lq adjoint bundle\rq\rq\,\,we mean the associated bundle (to the fixed standard principal bundle) which has sections in a one to one correspondence to the group of {\em gauge transformations} of $(P, M, \pi_{P, M}, G)$, that is principal bundle automorphisms of $P$ projecting on $id_M$, as mentioned for example in \cite[Chapter 6]{MARATHE (1992)}.
\end{ex}

\begin{oss}\label{oss 1.2.5}
It is nevertheless possible to note that {\em not} all Lie group fiber bundles are the adjoint bundle of certain principal bundle. A necessary and sufficient condition for the existence of such a standard principal bundle was found in \cite{MACKENZIE (1989)} by Mackenzie. In the same article it is also simply observed that if we consider an abelian Lie group $G$, then $(Ad(P), M, \pi_{Ad(P), M}, G)$ is just the trivial bundle $(M \times G, M, \pi_{M \times G, M}, G)$. This implies that if we fix a non-trivial Lie group fiber bundle $(\mathcal{G}, M, \pi_{\mathcal{G}, M}, G)$ with abelian typical fiber $G$ (e.g.\,\,the tangent bundle of the $2$-dimensional sphere $S^2$, on the base of Example \ref{ex 1.2.3}), then $(\mathcal{G}, M, \pi_{\mathcal{G}, M}, G)$ cannot be the adjoint bundle of any principal bundle.
\end{oss}

\medskip

The {\em multiplication map}:
\begin{align*}
\mathcal{M} \colon \mathcal{G} \times_M \mathcal{G} \longrightarrow \mathcal{G} \colon (\gamma_1, \gamma_2) \longmapsto \mathcal{M}(\gamma_1, \gamma_2)=\gamma_1 \cdot \gamma_2
\end{align*}
where $\mathcal{G} \times_M \mathcal{G}=\{ (\gamma_1, \gamma_2)  \in \mathcal{G} \times \mathcal{G} \,\,|\,\, \pi_{\mathcal G, M}(\gamma_1)=\pi_{\mathcal G, M}(\gamma_2) \}$ is the fibered product and where we denote this time with $\cdot$ the multiplication in the Lie group $\mathcal{G}_{x}$, $x=\pi_{\mathcal G, M}(\gamma_1)=\pi_{\mathcal G, M}(\gamma_2)$, and {\em the inversion map}:
\begin{align*}
\iota \colon \mathcal{G} \longrightarrow \mathcal{G} \colon \gamma \longmapsto \iota(\gamma)=\gamma^{-1}
\end{align*}
where ${\cdot}^{-1}$ is the inversion in the Lie group $\mathcal{G}_{x}$, $x=\pi_{\mathcal G, M}(\gamma)$, are bundle morphisms over the identity ${id}_M \colon M \rightarrow M$, since $\pi_{\mathcal G, M} \circ \mathcal{M}=\pi_{\mathcal{G} \times_{M} \mathcal{G}, M}$ and $\pi_{\mathcal G, M} \circ \mathcal{\iota}=\pi_{\mathcal{G}, M}$.

\begin{oss}\label{oss 1.7}
Similarly to the case of vector bundles, which always admit the global {\em zero section}, for a Lie group fiber bundle $(\mathcal G, M, \pi_{\mathcal G, M}, G)$ we have the global {\em unit section} defined by the map:
\begin{align*}
1 \colon M \longrightarrow \mathcal{G} \colon x \longmapsto e_x \in \mathcal{G}_x
\end{align*}
which is smooth since in fibered coordinates it is written as $1 \colon x^\mu \mapsto (x^\mu, e^I)$, because $\psi|_{\mathcal G_x}(e_x)=(x, e)$ holds true for the local trivializations of Definition \ref{defn 1.5}.
\end{oss}

Observe also that {\em Lie algebra fiber bundles} $(\mathfrak{a}, M, \pi_{\mathfrak{a}, M}, \mathfrak{g})$ can be constructed exactly as in Definition \ref{defn 1.5} except that $\mathfrak{g}$ is a Lie algebra, the fibers $\mathfrak{a}_x=(\pi_{\mathfrak{a}, M})^{-1}(x)$ have a (canonical) Lie algebra structure, $\psi \colon (\pi_{\mathfrak{a}, M})^{-1}(U) \rightarrow U \times \mathfrak{g}$ and $\psi|_{\mathfrak{a}_x} \colon \mathfrak{a}_x \rightarrow \{x\} \times \mathfrak{g} \equiv \mathfrak{g}$ is a Lie algebra isomorphism. Note that Lie algebra fiber bundles are in particular vector bundles (and are hence a particular instance of Lie group fiber bundles, according to Example \ref{ex 1.2.3}).

To a Lie group fiber bundle $(\mathcal G, M, \pi_{\mathcal G, M}, G)$ is always attached a specific Lie algebra fiber bundle for which $\mathfrak{g}$ is the Lie algebra of $G$, for all $x \in M$ the fiber $\mathfrak{a}_x$ is the Lie algebra of $\mathcal{G}_x$, $\mathfrak{a}=\coprod_{x \in M}\mathfrak{a}_x$ and the local trivializations are of the form:
\begin{align*}
\Psi \colon (\pi_{\mathfrak{a}, M})^{-1}(U) \longrightarrow U \times \mathfrak{g} \colon \xi \longmapsto \Big(\pi_{\mathfrak{a}, M}(\xi)=x, T_{e_x}(\psi|_{\mathcal G_x})(\xi)\Big)
\end{align*}
where $\psi$ is a local trivialization of $(\mathcal{G}, M, \pi_{\mathcal{G}, M}, G)$ (note that $\Psi|_{\mathfrak{a}_x}=T_{e_x}(\psi|_{\mathcal G_x})$ is a Lie algebra isomorphism since $\psi|_{\mathcal G_x}$ is a Lie group isomorphism).
 
\medskip 

It is now possible to give some definitions regarding the crucial notion of {\em fibered action}, adapting to the present case of study similar notions from the standard Lie group theory:
\begin{defn}\label{defn 1.9}
A {\em (right) fibered action} of a Lie group fiber bundle $(\mathcal{G}, M, \pi_{\mathcal{G}, M}, G)$ on a fiber bundle $(P, M, \pi_{P, M}, F)$ is a (smooth) bundle morphism $\Phi$ over the identity ${id}_M \colon M \longrightarrow M$ (this means that $\Phi$ is a {\em vertical} bundle morphism) such that:
\begin{equation}\label{4}
\xymatrixrowsep{0.56in}
\xymatrixcolsep{0.8in}
	\xymatrix{
	   {P \times_M \mathcal{G}} \ar[r]^{\Phi} \ar[d]_{\pi_{P \times_M \mathcal{G}, M}} & {P} \ar[d]^{\pi_{P, M}} \\
	   {M}  \ar[r]_{{id}_M} & {M}  \\}
\end{equation}
is a commutative diagram, where $P \times_M \mathcal{G}=\{ (p, \gamma)  \in P \times \mathcal{G} \,\,|\,\, \pi_{P, M}(p)=\pi_{\mathcal G, M}(\gamma) \}$, with the conditions:
\begin{enumerate}
\item $\forall \, p \in P$, $\Phi(p, e_x)=p$, with $\pi_{P, M}(p)=x$.
\item $\forall \, (p, \gamma_1), (p, \gamma_2) \in P \times_M \mathcal{G}$, $\Phi\Big(p, \mathcal{M}(\gamma_1, \gamma_2)\Big)=\Phi(p, \gamma_1 \cdot \gamma_2)=\Phi\Big(\Phi(p, \gamma_1), \gamma_2\Big)$.
\end{enumerate}
\end{defn}

Note that in the last definition the second condition is well-posed since if $(p, \gamma_1)$ and $(p, \gamma_2)$ belong to $P \times_M \mathcal{G}$, then $\pi_{\mathcal G, M}(\gamma_1)=\pi_{P, M}(p)=\pi_{\mathcal G, M}(\gamma_2)$ and as a consequence $\gamma_1$ and $\gamma_2$ belong to the same fiber of $\mathcal{G}$, that is the product $\mathcal{M}(\gamma_1, \gamma_2)=\gamma_1 \cdot \gamma_2$ exists.

Notice also that a right fibered action $\Phi \colon P \times_M \mathcal{G} \rightarrow P$ induces, once restricted to the fibers of $P \times_M \mathcal{G}$ and since \eqref{4} holds, a smooth right action of the Lie groups $\mathcal{G}_x$ on the respective fibers of $P$, i.e.\,\,$P_x=(\pi_{P, M})^{-1}(x)$:
\begin{align}\label{5}
\Phi_x=\Phi|_{P_x \times \mathcal{G}_x} \colon P_x \times \mathcal{G}_x \longrightarrow P_x
\end{align}

\medskip

We would like to emphasize that $(P, M, \pi_{P, M}, F)$ is assumed to be a generic (not necessarily principal) fiber bundle. 

\begin{defn}\label{defn 1.10}
Let $\Phi \colon P \times_M \mathcal{G} \longrightarrow P$ be a fibered action of a Lie group fiber bundle $(\mathcal{G}, M, \pi_{\mathcal{G}, M}, G)$ on a fiber bundle $(P, M, \pi_{P, M}, F)$. The fibered action $\Phi$ is said to be:
\begin{enumerate}
\item {\em Free} if $\Phi(p, \gamma)=p$, for some $(p, \gamma) \in P \times_M \mathcal{G}$, then $\gamma=e_x$, $\pi_{P, M}(p)=x$.
\item {\em Proper} if the (smooth) map:
\begin{align*}
\Theta \colon P \times_M \mathcal{G} \longrightarrow P \times_M P \colon (p, \gamma) \longmapsto \Big(p, \Phi(p, \gamma)\Big)
\end{align*}
is a proper map ($\Theta$ is well-defined since $\pi_{P, M}\Big(\Phi(p, \gamma)\Big)=\pi_{P \times_M \mathcal{G}, M}(p, \gamma)=\pi_{P, M}(p)$).
\end{enumerate}
\end{defn}

\begin{lem}\label{lem 1.11}
Let $\Phi \colon P \times_M \mathcal{G} \longrightarrow P$ be a free and proper fibered action of a Lie group fiber bundle $(\mathcal{G}, M, \pi_{\mathcal{G}, M}, G)$ on a fiber bundle $(P, M, \pi_{P, M}, F)$. Then each $\Phi_x$, given by \eqref{5} with $x \in M$, is a free and proper smooth right action and hence each $(P_x, \mathcal{S}_x, \pi_{P_x, \mathcal{S}_x}, \mathcal{G}_x)$ is a principal bundle, by setting $\mathcal{S}_x=P_x/ \mathcal{G}_x$.
\end{lem}
\begin{proof}

\medskip

\noindent From the assumption that the fibered action $\Phi$ is free, we get that $\Phi_x$ is free since:
\begin{align*}
\Phi_x(p, \gamma)=p \, \Longrightarrow \, \Phi(p, \gamma)=p  \, \Longrightarrow \, \gamma=e_x \in \mathcal{G}_x
\end{align*}
given $(p, \gamma) \in P_x \times \mathcal{G}_x$. Moreover the following map:
\begin{align*}
\Theta_x=\Theta|_{P_x \times \mathcal{G}_x} \colon P_x \times \mathcal{G}_x \longrightarrow P_x \times P_x \colon (p, \gamma) \longmapsto \Big(p, \Phi_x(p, \gamma)\Big)
\end{align*}
is proper since the map $\Theta$ is proper by hypothesis and if $K$ is compact in $P_x \times P_x$, then it is compact also in $P \times_M P$ and hence $\Theta_x^{-1}(K)=\Theta^{-1}(K)$ is compact. Thereby $\Phi_x$  is also proper.

The second part of the statement is obvious in view of Definition \ref{defn 1.4}.

\end{proof}

Once fixed a fibered action of $(\mathcal{G}, M, \pi_{\mathcal{G}, M}, G)$ on $(P, M, \pi_{P, M}, F)$, it is possible to define an equivalence relation $\sim_{\mathcal{G}}$ on $P$ declaring:
\begin{align*}
p_1 \sim_{\mathcal{G}} p_2  \, \Longleftrightarrow \, \exists \, \gamma \in \mathcal{G} \, | \, \Phi(p_1, \gamma)=p_2
\end{align*}
for $p_1, p_2 \in P$. By Definition \ref{defn 1.9}, the relation $\Phi(p_1, \gamma)=p_2$ holds only if one has that $p_1, p_2 \in P_x$ and $\gamma \in \mathcal{G}_x$ for the same $x \in M$, hence we have that:
\begin{align*}
p_1 \sim_{\mathcal{G}} p_2  \, &\Longleftrightarrow \, \exists \, x \in M, \, \exists \, \gamma \in \mathcal{G}_x \, | \, p_1, p_2 \in P_x, \, \Phi(p_1, \gamma)=\Phi_x(p_1, \gamma)=p_2 
\\
&\Longleftrightarrow \, \exists \, x \in M \, | \, p_1, p_2 \in P_x, \, p_1 \sim_{\mathcal{G}_x} p_2
\end{align*}
where $\sim_{\mathcal{G}_x}$ is the equivalence relation on $P_x$ giving rise to the canonical projection $\pi_{P_x, \mathcal{S}_x}$. In this sense the quotient space $\mathcal{S}=P/\mathcal{G}=P/{\sim_{\mathcal{G}}}$ can be regarded as the disjoint union of the quotient spaces $\mathcal{S}_x=P_x/\mathcal{G}_x=P_x/{\sim_{\mathcal{G}_x}}$:
\begin{align*}
\mathcal{S}=\coprod_{x \in M} \mathcal{S}_x=\{ [p]_{\mathcal{G}} \equiv (x, [p]_{\mathcal G_x}) \,\,|\,\, x\in M,\, p\in P_x \} 
\end{align*}

In particular the following diagram is commutative:
\begin{equation}\label{7}
\xymatrixrowsep{0.56in}
\xymatrixcolsep{0.8in}
	\xymatrix{
	   {P} \ar[r]^{\pi_{P, \mathcal{S}}} \ar@/_2.0pc/_{\pi_{P, M}}[rr] & {\mathcal{S}} \ar[r]^{\pi_{\mathcal{S}, M}} & {M} \\}
\end{equation}
where $\pi_{P, \mathcal{S}}(p)=[p]_{\mathcal{G}}$ and $\pi_{\mathcal{S}, M}\Big([p]_{\mathcal{G}} \equiv (x, [p]_{\mathcal G_x})\Big)=x$.

\medskip

\needspace{3\baselineskip}
In \cite{CASTRILLON LOPEZ (2023)} a generalization of the aforementioned Quotient Manifold Theorem is proved, which eventually enables the definition of generalized principal bundles:
\begin{thm}[Generalized Quotient Manifold Theorem]\label{thm 1.12}
Given a free and proper fibered action $\Phi$ of a Lie group fiber bundle $(\mathcal{G}, M, \pi_{\mathcal{G}, M}, G)$ on a fiber bundle $(P, M, \pi_{P, M}, F)$, then $\mathcal{S}=P/\mathcal{G}$ admits a unique smooth structure such that $(P, \mathcal{S}, \pi_{P, \mathcal{S}}, G)$ is a fiber bundle.
\end{thm}
\noindent Even though, as mentioned above, it is possible to find the core idea of the proof in the work by Castrillón López and Rodríguez Abella, we prefer to go through some of its details since this will be crucial in order to introduce in the next section the new notion of {\em generalized principal bundle coordinates} (with their transformation laws).

\medskip

\begin{proof}

\medskip

\noindent We firstly fix local trivializations for $\mathcal{G}$ and $P$ on the same open set $U \subseteq M$:
\begin{align*}
\psi_\mathcal{G} \colon (\pi_{\mathcal{G}, M})^{-1}(U) \longrightarrow U \times G, \quad \psi_P \colon (\pi_{P, M})^{-1}(U) \longrightarrow U \times F 
\end{align*}
Recall that given local trivializations as above we get a local trivialization of the fibered product $P \times_M \mathcal{G}$ of the form:
\begin{align*}
\psi_{P \times_M \mathcal{G}} \colon (\pi_{P \times_M \mathcal{G}, M})^{-1}(U) \longrightarrow U \times F \times G \colon (p, \gamma) \longmapsto \psi_{P \times_M \mathcal{G}}(p, \gamma)=(x, f, g)
\end{align*}
where $\psi_P(p)=(x, f)$ and $\psi_\mathcal{G}(\gamma)=(x, g)$. In this way the fibered action $\Phi \colon P \times_M \mathcal{G} \rightarrow P$ takes the local expression:
\begin{align*}
\psi_P \circ \Phi \circ (\psi_{P \times_M \mathcal{G}})^{-1} \colon U \times F \times G \longrightarrow U \times F \colon (x, f, g) \longmapsto \psi_P\Big(\Phi(p, \gamma)\Big)=\Big(x, \hat{f}\Big)
\end{align*}
since $\Phi$ is vertical. Using the fact that $\Phi$ is a free, proper fibered action, it can be checked that this can be seen as a standard free and proper smooth action of the Lie group $G$ on $U \times F$.

From the Quotient Manifold Theorem we get that $(U \times F)/G$ has a unique smooth structure with the property that the canonical projection $\pi_{U \times F, (U \times F)/G} \colon U \times F \rightarrow (U \times F)/G$ is a smooth submersion and in particular that $(U \times F, (U \times F)/G, \pi_{U \times F, (U \times F)/G}, G)$ is a principal bundle.

\medskip

We can now consider the map:
\begin{align*}
\varepsilon \colon (\pi_{\mathcal{S}, M})^{-1}(U) \longrightarrow (U \times F)/G \colon [p]_{\mathcal{G}} \equiv (x, [p]_{\mathcal G_x}) \longmapsto [x, f]_G
\end{align*}
where $\psi_P(p)=(x, f)$. It is a bijection since:
\begin{enumerate}
\item It is well-posed and injective:
\begin{align*}
(x, [p_1]_{\mathcal{G}_x})=(x, [p_2]_{\mathcal{G}_x}) \, &\Longleftrightarrow \, [p_1]_{\mathcal{G}_{x}}=[p_2]_{\mathcal{G}_{x}} \, \Longleftrightarrow \, \Phi(p_1, \gamma)=p_2, \, \gamma \in \mathcal{G}_x 
\\
&\Longleftrightarrow \, \Big(\psi_P \circ \Phi \circ (\psi_{P \times_M \mathcal{G}})^{-1}\Big)(x, f_1, g)=(x, f_2) \, \Longleftrightarrow \, [x, f_1]_G=[x, f_2]_G
\end{align*}
where $\psi_P(p_1)=(x, f_1)$, $\psi_P(p_2)=(x, f_2)$ and $\psi_\mathcal{G}(\gamma)=(x, g)$.
\item It is surjective since, fixing $[x, f]_G \in (U \times F)/G$, it is enough to define $p=\psi_P^{-1}(x, f)$, then $\varepsilon\Big([p]_{\mathcal{G}}\Big)=[x, f]_G$ by definition.
\end{enumerate}
In this way it is possible to transport the smooth charts on $(U \times F)/G$ to $(\pi_{\mathcal{S}, M})^{-1}(U)$ through $\varepsilon$, obtaining a manifold structure on $\mathcal{S}=P/\mathcal{G}$ and making $\varepsilon$ a diffeomorphism.

\medskip

Thanks to \eqref{7} we clearly have the commutative diagram:
\begin{equation*}
\xymatrixrowsep{0.56in}
\xymatrixcolsep{0.8in}
	\xymatrix{
	   {(\pi_{P, M})^{-1}(U)} \ar[r]^{\psi_P} \ar[d]_{\pi_{P, \mathcal{S}}} & {U \times F} \ar[d]^{\pi_{U \times F, (U \times F)/G}} \\
	   {(\pi_{\mathcal{S}, M})^{-1}(U)}  \ar[r]_{\varepsilon} & {(U \times F)/G} \\}
\end{equation*}
We also know that $\pi_{U \times F, (U \times F)/G}$ is a smooth submersion and that $\psi_P$ and $\varepsilon$ are diffeomorphisms, thereby $\pi_{P, \mathcal{S}}$ is a smooth submersion on $(\pi_{P, M})^{-1}(U)$ and also globally, as $U$ is arbitrary.

\medskip

Finally it is possible to construct a trivializing atlas for $(P, \mathcal{S}, \pi_{P, \mathcal{S}}, G)$ starting from a fixed trivializing atlas $\{ (U_\alpha,\psi_{\mathcal{G}}^\alpha) \,\,|\,\, \alpha \in \mathcal{A} \}$ of $(\mathcal{G}, M, \pi_{\mathcal{G}, M}, G)$. For $\alpha \in \mathcal{A}$, let $V_\alpha=\pi^{-1}_{\mathcal{S}, M}(U_\alpha)$ and we choose a local section $\sigma \colon V_\alpha \rightarrow \pi^{-1}_{P,\mathcal{S}}(V_\alpha)$ (which exists if we have fixed a small enough $U_\alpha$). We set:
\begin{align*}
\psi^\alpha_P \colon \pi^{-1}_{P,\mathcal{S}}(V_\alpha) \longrightarrow V_\alpha \times G \colon p \longmapsto \Big([p]_{\mathcal{G}}, g\Big)
\end{align*}
where $g \in G$ is such that $p=\Phi\Big(\sigma\Big([p]_{\mathcal{G}}\Big), \gamma\Big)$, with $\gamma=(\psi^\alpha_{\mathcal G})^{-1}(x, g) \in \mathcal{G}_x$, where $\pi_{\mathcal{S}, M}\Big([p]_{\mathcal{G}}\Big)=x$. The element $g$ exists uniquely since the element $\gamma$ exists uniquely, because:
\begin{align*}
\pi_{P, \mathcal{S}}(p)=[p]_{\mathcal{G}}=\pi_{P, \mathcal{S}}\Big(\sigma\Big([p]_{\mathcal{G}}\Big)\Big)
\end{align*}
and hence:
\begin{align*}
\pi_{P, M}(p)&=\Big(\pi_{\mathcal{S}, M} \circ \pi_{P, \mathcal{S}}\Big)(p)=\pi_{\mathcal{S}, M}\Big([p]_{\mathcal{G}}\Big)=x=\Big(\pi_{\mathcal{S}, M} \circ \pi_{P, \mathcal{S}}\Big)\Big(\sigma\Big([p]_{\mathcal{G}}\Big)\Big)=\pi_{P, M}\Big(\sigma\Big([p]_{\mathcal{G}}\Big)\Big)
\end{align*}
which means that $p, \sigma\Big([p]_{\mathcal{G}}\Big) \in P_x$. We have that $(P_x, \mathcal{S}_x, \pi_{P_x, \mathcal{S}_x}, \mathcal{G}_x)$ is a principal bundle, by Lemma \ref{lem 1.11}, and from standard principal bundle theory there exists a unique $\gamma \in \mathcal{G}_x$ such that $p=\Phi_x\Big(\sigma\Big([p]_{\mathcal{G}}\Big), \gamma\Big)=\Phi\Big(\sigma\Big([p]_{\mathcal{G}}\Big), \gamma\Big)$.

An inverse of $\psi^\alpha_P$ is clearly given by:
\begin{align*}
(\psi^\alpha_P)^{-1} \colon V_\alpha \times G \longrightarrow \pi^{-1}_{P,\mathcal{S}}(V_\alpha) \colon \Big([p]_{\mathcal{G}}, g\Big) \longmapsto \Phi\Big(\sigma\Big([p]_{\mathcal{G}}\Big), (\psi^\alpha_{\mathcal{G}})^{-1}\Big(\pi_{\mathcal{S}, M}\Big([p]_{\mathcal{G}}\Big), g\Big)\Big)
\end{align*}
which is smooth by composition. Using that $\pi_{P, \mathcal{S}} \colon P \rightarrow \mathcal{S}$ is a smooth submersion along the lines of the standard proof of the fourth statement in Theorem $\ref{thm 1.2}$, $\psi^\alpha_P$ itself is smooth, so that $\psi^\alpha_P$ is a diffeomorphism. In conclusion, $\{ (V_\alpha, \psi_P^\alpha) \,\,|\,\, \alpha \in \mathcal{A} \}$ is a trivializing atlas for $(P, \mathcal{S}, \pi_{P, \mathcal{S}}, G)$, which means that $(P, \mathcal{S}, \pi_{P, \mathcal{S}}, G)$ is a fiber bundle. 

\medskip

The uniqueness part of the thesis is a consequence of the standard argument based on the use of local sections, which is part of the proof of the classical Quotient Manifold Theorem (see again \cite{LEE (2013)}).

\end{proof}

\begin{oss}\label{oss 1.13}
Note that using local sections of the fiber bundle $(P, \mathcal{S}, \pi_{P, \mathcal{S}}, G)$ and \eqref{7}, we can also deduce from the previous theorem that also $\pi_{\mathcal{S}, M} \colon \mathcal{S} \rightarrow M$ is a smooth map. Moreover, since $\pi_{P, M}$ and $\pi_{P, \mathcal{S}}$ are both surjective submersions, then also $\pi_{\mathcal{S}, M}$ is a surjective submersion, making $(\mathcal{S}, M, \pi_{\mathcal{S}, M})$ a fibered manifold.
\end{oss}

It holds as well that (see once more \cite{CASTRILLON LOPEZ (2023)}):
\begin{prop}\label{prop 1.14}
Considering the fiber bundle $(P, \mathcal{S}, \pi_{P, \mathcal{S}}, G)$ of Theorem \ref{thm 1.12} with the corresponding trivializing atlas $\{ (V_\alpha, \psi_P^\alpha) \,\,|\,\, \alpha \in \mathcal{A} \}$, let $p=(\psi_P^\alpha)^{-1}\Big([p]_{\mathcal G}, h\Big) \in P_x$ and $\gamma=(\psi_{\mathcal{G}}^\alpha)^{-1}(x, g) \in \mathcal{G}_x$, with $\pi_{\mathcal{S}, M}\Big([p]_{\mathcal{G}}\Big)=x \in U_\alpha$ and $g, h \in G$, then $\Phi(p, \gamma)=(\psi_P^\alpha)^{-1}\Big([p]_{\mathcal{G}}, h \cdot g\Big)$.
\end{prop}
This means that, using the trivializing atlas $\{ (V_\alpha, \psi_P^\alpha) \,\,|\,\, \alpha \in \mathcal{A} \}$, the fibered action is locally given by the right multiplication on the Lie group $G$, as it happens also in the usual principal bundle theory.

\section{Generalized principal bundle coordinates} \label{sec 1.3}

In the same manner as we have established principal bundles in Definition \ref{defn 1.4} through the Quotient Manifold Theorem, we can now define generalized principal bundles using the generalization of the Quotient Manifold Theorem given by Theorem \ref{thm 1.12}:
\begin{defn}\label{defn 1.15}
A {\em generalized principal bundle with structure Lie group fiber bundle} $\mathcal{G}$ is a fiber bundle $(P, M, \pi_{P, M}, F)$ equipped with a smooth right fibered action of a Lie group fiber bundle $(\mathcal{G}, M, \pi_{\mathcal{G}, M}, G)$:
\begin{align*}
\Phi \colon P \times_M \mathcal{G} \longrightarrow P \colon (p, \gamma) \longmapsto \Phi(p, \gamma)
\end{align*}
such that $\Phi$ is free and proper.
\end{defn}
We eventually also say that the {\em generalized principal bundle} is the fiber bundle $(P, \mathcal{S}, \pi_{P, \mathcal{S}}, G)$ obtained from $(P, M, \pi_{P, M}, F)$ through Theorem \ref{thm 1.12} (as it is stated in \cite{CASTRILLON LOPEZ (2023)}).

\medskip

The geometric setting described until now can be summarized by the following figure:
\begin{figure}[h]
    \centering
    \boxed{\begin{tikzpicture}[scale=0.85, every node/.style={scale=0.765}]

        \tikzset{
            surface/.style={draw, ellipse, minimum width=3cm, minimum height=1cm},
            cylinder_top/.style={draw, ellipse, minimum width=2cm, minimum height=0.6cm},
            line/.style={draw, -{Stealth[length=2mm]}},
            dashed_line/.style={draw, dashed}}

        \node at (-1.3, 0) {$M$};
        \draw (0,0) ellipse (1cm and 0.3cm);
        \filldraw (0,0) circle (1pt) node[left] {\footnotesize $x$};

        \draw (-1,4) arc (180:360:1cm and 0.3cm);
        \draw[dashed] (-1,4) arc (180:0:1cm and 0.3cm);
        \draw (-1,4) -- (-1,6);
        \draw (1,4) -- (1,6);
        \draw (0,6) ellipse (1cm and 0.3cm);
        \node at (-1.3, 6) {$P$};

        \draw (-0.5,4) arc (180:360:0.5cm and 0.08cm);
        \draw[dashed] (-0.5,4) arc (180:0:0.5cm and 0.08cm);
        \draw (-0.5,4) -- (-0.5,6);
        \draw (0.5,4) -- (0.5,6);
        \draw (0,6) ellipse (0.5cm and 0.08cm);
        \node[left] at (-0.45, 5) {\footnotesize $P_x$};

        \draw[line] (0.7, 3.75) -- (0.7, 0.25) node[midway, right] {$\pi_{P, M}$};
        \draw[dashed] (-0.5,4) -- (0,3) -- (0,0);
        \draw[dashed] (0.5,4) -- (0,3);

        \begin{scope}[xshift=4.5cm]
            \node at (-1.3, 0) {$M$};
            \draw (0,0) ellipse (1cm and 0.3cm);
            \filldraw (0,0) circle (1pt) node[left] {\footnotesize $x$};
            
            \draw (-1,4) arc (180:360:1cm and 0.3cm);
            \draw[dashed] (-1,4) arc (180:0:1cm and 0.3cm);
            \draw (-1,4) -- (-1,6);
            \draw (1,4) -- (1,6);
            \draw (0,6) ellipse (1cm and 0.3cm);
            \node at (-1.3, 6) {$\mathcal{G}$};

            \draw (0, 6.0) -- (0, 4.0);
            \filldraw (0, 4.0) circle (1.0pt);
            \filldraw (0, 6.0) circle (1.0pt);
            \node[left] at (0.05, 5) {\footnotesize $\mathcal{G}_x$};
            
            \draw[line] (0.7, 3.75) -- (0.7, 0.25) node[midway, right] {$\pi_{\mathcal{G}, M}$};
            \draw[dashed] (0,4) -- (0,0);
        \end{scope}
        
        \node at (2.3, 5) {\Large $\times_M$};

        \draw[line, dashed] (4.5, 5.1) to[bend right=30] node[above] {$\Phi_x$} (0.5, 5.1);

        \begin{scope}[xshift=10cm]
            \node at (-1.3, 0) {$M$};
            \draw (0,0) ellipse (1cm and 0.3cm);
            \filldraw (0,0) circle (1pt) node[left] {\footnotesize $x$};

            \draw (0,2) ellipse (1cm and 0.3cm);
            \fill[gray!20] (-0.5,2) arc (180:360:0.5cm and 0.08cm) -- (0.5,2) arc (0:180:0.5cm and 0.08cm) -- cycle;
            \draw (-0.5,2) arc (180:360:0.5cm and 0.08cm);
            \draw (-0.5,2) arc (180:0:0.5cm and 0.08cm);
            \node at (-1.3, 2) {$\mathcal{S}$};
            \filldraw (0,2) circle (1.0pt) node[left] {};
            \node[left] at (-0.4, 2) {\footnotesize $\mathcal{S}_x$};
        
             \fill[gray!20] (-0.5,4) arc (180:360:0.5cm and 0.08cm) -- (0.5,6) arc (0:180:0.5cm and 0.08cm) -- cycle;
             \draw (-0.5,4) arc (180:360:0.5cm and 0.08cm);
             \draw[dashed] (-0.5,4) arc (180:0:0.5cm and 0.08cm);
             \draw (-0.5,4) -- (-0.5,6);
             \draw (0.5,4) -- (0.5,6);
             \draw (0,6) ellipse (0.5cm and 0.08cm);
             \node[left] at (-0.45, 5) {\footnotesize $P_x$};
        
             \draw (0, 6.0) -- (0, 4.0);
             \filldraw (0, 4.0) circle (1.0pt);
             \filldraw (0, 6.0) circle (1.0pt);
             \node[left] at (0.05, 5) {\footnotesize $\mathcal{G}_x$};
        
             \draw (-1,4) arc (180:360:1cm and 0.3cm);
             \draw[dashed] (-1,4) arc (180:0:1cm and 0.3cm);
             \draw (-1,4) -- (-1,6);
             \draw (1,4) -- (1,6);
             \draw (0,6) ellipse (1cm and 0.3cm);
             \node at (-1.3, 6) {$P$};

            \draw[line] (0.7, 3.75) -- (0.7, 2.24) node[midway, right] {$\pi_{P, \mathcal{S}}$};
            \draw[line] (0.7, 1.75) -- (0.7, 0.25) node[midway, right] {$\pi_{\mathcal{S}, M}$};
            \draw[dashed] (0,4) -- (0,2);
            \draw[dashed] (-0.5,4) -- (-0.5,2) -- (0,1) -- (0,0);
            \draw[dashed] (0.5,4) -- (0.5,2) -- (0,1);
            
            \draw[line] (-0.7, 3.75) .. controls (-2, 3) and (-2, 1) .. (-0.7, 0.25) node[midway, left] {$\pi_{P, M}$};
        \end{scope}

        \draw[line] (5.7, 5) -- (8.8, 5) node[midway, above] {$\Phi$};

    \end{tikzpicture}}
    \caption{Generalized principal bundle construction}
\end{figure}
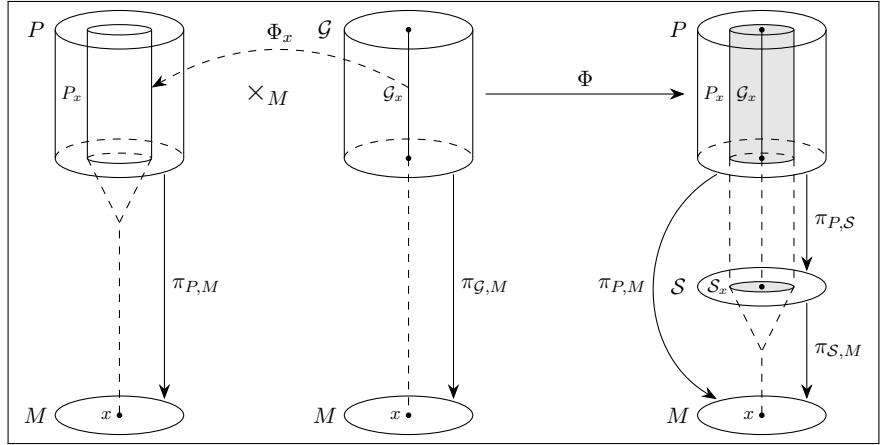

\noindent In view of Lemma \ref{lem 1.11}, we could in particular say (oversimplifying) that a generalized principal bundle can be seen as a \lq\lq bundle\rq\rq\,\,of standard principal bundles.

\begin{ex}\label{ex 1.3.2}
We will see in Section \ref{sec 2.5} that a usual principal bundle is an example of generalized principal bundle. Another example of generalized principal bundle is actually given by any Lie group fiber bundle $(\mathcal{G}, M, \pi_{\mathcal{G}, M}, G)$. 

In fact, the multiplication map $\mathcal{M} \colon \mathcal{G} \times_M \mathcal{G} \rightarrow \mathcal{G}$ defines a right fibered action of the Lie group fiber bundle $(\mathcal{G}, M, \pi_{\mathcal{G}, M}, G)$ on itself since it is a bundle morphism over the identity ${id}_M \colon M \rightarrow M$ and it is possible to check that the fibered action axioms of Definition \ref{defn 1.9} are satisfied. Moreover, $\mathcal{M}$ is also free and proper as a fibered action, since:
\begin{enumerate}
\item If $\mathcal{M}(\gamma, \delta)=\gamma \cdot \delta=\gamma$, for some $(\gamma, \delta) \in \mathcal{G} \times_M \mathcal{G}$, then $\delta=e_x$, $\pi_{\mathcal{G}, M}(\gamma)=x$.
\item The map:
\begin{align*}
\Theta_{\mathcal{M}} \colon  \mathcal{G} \times_M \mathcal{G} \longrightarrow \mathcal{G} \times_M \mathcal{G} \colon (\gamma, \delta) \longmapsto \Big(\gamma, \mathcal{M}(\gamma, \delta)\Big)=(\gamma, \gamma \cdot \delta)
\end{align*}
is proper. In fact $\mathcal{G} \times_M \mathcal{G}$ is an Hausdorff space (because it is a differentiable manifold) and $\Theta_{\mathcal{M}}$ has a smooth left inverse given by:
\begin{align*}
\Psi_{\mathcal{M}} \colon \mathcal{G} \times_M \mathcal{G} \longrightarrow \mathcal{G} \times_M \mathcal{G} \colon (\gamma, \delta) \longmapsto \Big(\gamma, \mathcal{M}\Big(\iota(\gamma), \delta\Big)\Big)
\end{align*}
since:
\begin{align*}
(\Psi_{\mathcal{M}} \circ \Theta_{\mathcal{M}})(\gamma, \delta)=\Psi_{\mathcal{M}}(\gamma, \gamma \cdot \delta)=\Big(\gamma, \mathcal{M}\Big(\iota(\gamma), \gamma \cdot \delta\Big)\Big)=(\gamma, \gamma^{-1} \cdot \gamma \cdot \delta)=(\gamma, \delta)
\end{align*}
These two properties give a sufficient condition for $\Theta_{\mathcal{M}}$ to be proper (see for example \cite[Appendix A]{LEE (2013)}).
\end{enumerate}
Thanks to Theorem \ref{thm 1.12}, this means that $(\mathcal{G}, \mathcal{G}/\mathcal{G}, \pi_{\mathcal{G}, \mathcal{G}/\mathcal{G}}, G)$ is a generalized principal bundle. Furthermore the equivalence relation $\sim_{\mathcal{G}}$ on $\mathcal{G}$ defining $\mathcal{G}/\mathcal{G}$ is given by:
\begin{align*}
\gamma_1 \sim_{\mathcal{G}} \gamma_2  \, \Longleftrightarrow \, \exists \, x \in M, \, \exists \, \delta \in \mathcal{G}_x \, | \, \gamma_1, \gamma_2 \in \mathcal{G}_x, \, \mathcal{M}(\gamma_1, \delta)=\gamma_1 \cdot \delta=\gamma_2 \, \Longleftrightarrow \, \exists \, x \in M \, | \, \gamma_1, \gamma_2 \in \mathcal{G}_x
\end{align*}
In this sense $\mathcal{G}/\mathcal{G}=\{ [\gamma]_{\mathcal{G}} \equiv x \,\,|\,\, x\in M,\, \gamma \in \mathcal{G}_x \}$ and $\pi_{\mathcal{G}/\mathcal{G}, M}$ in \eqref{7} is essentially the identity which is bijective and a submersion from Remark \ref{oss 1.13}, hence a diffeomorphism. As a consequence $\mathcal{G}/\mathcal{G} \cong M$ at a differentiable manifold level and $(\mathcal{G}, M, \pi_{\mathcal{G}, M}, G) \equiv (\mathcal{G}, \mathcal{G}/\mathcal{G}, \pi_{\mathcal{G}, \mathcal{G}/\mathcal{G}}, G)$.

The conclusion is that we can always regard a Lie group fiber bundle $(\mathcal{G}, M, \pi_{\mathcal{G}, M}, G)$ as a generalized principal bundle. This result can be also found with a different proof in \cite{FISCHER (2022)}, which includes an independent formulation of many results that appear in \cite{CASTRILLON LOPEZ (2023)} (see also \cite{FISCHER (2021)} for more on this equivalent framework). Moreover, the same result can be seen as a by-product of a more general finding (see \cite[Chapter 5, Section 7]{MOERDIJK (2003)}) in the context of {\em principal bundles over Lie groupoids}, which recover generalized principal bundles as particular cases in the same vein of the PB-groupoids cited in the Introduction.

\medskip

In particular since we have already noted in Remark \ref{oss 1.6} that Lie group fiber bundles are in general not principal bundles, we deduce that the notion of generalized principal bundle {\em strictly} extends the one of principal bundle. From Examples \ref{ex 1.2.3} and \ref{ex 1.2.4}, we deduce also that vector bundles and adjoint bundles are classes of examples of generalized principal bundles.
\end{ex}

\medskip

Following our constructive approach to the topic, it is natural at this point to give fibered coordinates on $P$ taking into account the generalized principal bundle structure on $(P, \mathcal{S}, \pi_{P, \mathcal{S}}, G)$. 

Local coordinates techniques regarding generalized principal bundles are already present in \cite{CASTRILLON LOPEZ (2023)} and especially in \cite{CASTRILLON LOPEZ (2024)}, inside some examples and proofs. Nevertheless the treatment in both cases lacks of the study of the corresponding transformation laws which we will analyze thoroughly in the next paragraphs.

\medskip

In order to do that, consider two local trivializations $(V_\alpha, \psi_P^\alpha)$ and $(V_\beta, \psi_P^\beta)$ for $(P, \mathcal{S}, \pi_{P, \mathcal{S}}, G)$, as defined above in the proof of Theorem \ref{thm 1.12}, such that $V_\alpha \cap V_\beta=\pi^{-1}_{\mathcal{S}, M}(U_\alpha) \cap \pi^{-1}_{\mathcal{S}, M}(U_\beta) \neq \varnothing$, which implies $U_\alpha \cap U_\beta \neq \varnothing$, and compute the map:
\begin{align*}
\psi^\beta_P \circ (\psi^\alpha_P)^{-1} \colon V_\alpha \cap V_\beta \times G & \longrightarrow \pi^{-1}_{P,\mathcal{S}}(V_\alpha \cap V_\beta) \longrightarrow V_\alpha \cap V_\beta \times G
\\
\Big([p]_{\mathcal{G}}, g\Big) & \longmapsto (\psi^\alpha_P)^{-1}\Big([p]_{\mathcal{G}}, g\Big) \longmapsto \psi^\beta_P \Big((\psi^\alpha_P)^{-1}\Big([p]_{\mathcal{G}}, g\Big)\Big)
\end{align*}
which from the definitions of $(\psi^\alpha_P)^{-1}$ and $\psi_P^\beta$ provides:
\begin{align*}
\psi^\beta_P \Big((\psi^\alpha_P)^{-1}\Big([p]_{\mathcal{G}}, g\Big)\Big)&=\psi^\beta_P \Big(\Phi\Big(\sigma\Big([p]_{\mathcal{G}}\Big), (\psi^\alpha_{\mathcal{G}})^{-1}\Big(\pi_{\mathcal{S}, M}\Big([p]_{\mathcal{G}}\Big), g\Big)\Big)\Big)=\Big([p]_{\mathcal{G}}, h\Big)
\end{align*}
where $h \in G$ is the unique element (as seen in the proof of Theorem \ref{thm 1.12}) such that:
\begin{align*}
\Phi\Big(\sigma\Big([p]_{\mathcal{G}}\Big), (\psi^\alpha_{\mathcal{G}})^{-1}(x, g)\Big)&=\Phi\Big(\hat{\sigma}\Big(\Big[\Phi\Big(\sigma\Big([p]_{\mathcal{G}}\Big), (\psi^\alpha_{\mathcal{G}})^{-1}(x, g)\Big)\Big]_{\mathcal{G}}\Big), \gamma\Big)=\Phi\Big(\hat{\sigma}\Big([p]_{\mathcal{G}}\Big), \gamma\Big)
\end{align*}
with $\gamma=(\psi^\beta_{\mathcal G})^{-1}(x, h) \in \mathcal{G}_x$, once set $\pi_{\mathcal{S}, M}\Big([p]_{\mathcal{G}}\Big)=x$ and where $\hat{\sigma}$ is the fixed local section related to $\psi^\beta_P$. Note that:
\begin{align*}
\pi_{P, M}\Big(\sigma\Big([p]_{\mathcal{G}}\Big)\Big)&=\Big(\pi_{\mathcal{S}, M} \circ \pi_{P, \mathcal{S}}\Big)\Big(\sigma\Big([p]_{\mathcal{G}}\Big)\Big)=\pi_{\mathcal{S}, M}\Big([p]_{\mathcal{G}}\Big)=x
\\
\pi_{P, M}\Big(\hat{\sigma}\Big([p]_{\mathcal{G}}\Big)\Big)&=\Big(\pi_{\mathcal{S}, M} \circ \pi_{P, \mathcal{S}}\Big)\Big(\hat{\sigma}\Big([p]_{\mathcal{G}}\Big)\Big)=\pi_{\mathcal{S}, M}\Big([p]_{\mathcal{G}}\Big)=x
\end{align*}
and hence in a similar manner as in the proof of Theorem \ref{thm 1.12} there exists a unique $\varphi \in G$ such that $\sigma\Big([p]_{\mathcal{G}}\Big)=\Phi\Big(\hat{\sigma}\Big([p]_{\mathcal{G}}\Big), (\psi^\beta_{\mathcal G})^{-1}(x, \varphi)\Big)$. In this way we have that:
\begin{align*}
\Phi\Big(\sigma\Big([p]_{\mathcal{G}}\Big), (\psi^\alpha_{\mathcal{G}})^{-1}(x, g)\Big)&=\Phi\Big(\Phi\Big(\hat{\sigma}\Big([p]_{\mathcal{G}}\Big), (\psi^\beta_{\mathcal G})^{-1}(x, \varphi)\Big), (\psi^\alpha_{\mathcal{G}})^{-1}(x, g)\Big)
\\
&=\Phi\Big(\hat{\sigma}\Big([p]_{\mathcal{G}}\Big), (\psi^\beta_{\mathcal G})^{-1}(x, \varphi) \cdot (\psi^\alpha_{\mathcal{G}})^{-1}(x, g)\Big)
\\
&=\Phi\Big(\hat{\sigma}\Big([p]_{\mathcal{G}}\Big), (\psi^\beta_{\mathcal G})^{-1}(x, \varphi) \cdot \Big((\psi^\beta_{\mathcal G})^{-1} \circ \psi^\beta_{\mathcal G} \circ (\psi^\alpha_{\mathcal{G}})^{-1}\Big)(x, g)\Big)
\\
&=\Phi\Big(\hat{\sigma}\Big([p]_{\mathcal{G}}\Big), (\psi^\beta_{\mathcal G})^{-1}\Big((x, \varphi) \cdot \Big(x, G_{\beta\alpha}(x)[g]\Big)\Big)\Big)
\\
&=\Phi\Big(\hat{\sigma}\Big([p]_{\mathcal{G}}\Big), (\psi^\beta_{\mathcal G})^{-1}\Big(x, \varphi \cdot G_{\beta\alpha}(x)[g]\Big)\Big)
\end{align*}
where we used that $(\psi^\beta_{\mathcal G})^{-1}$ is fiberwise a Lie group isomorphism and where:
\begin{align*}
\psi^\beta_{\mathcal G} \circ (\psi^\alpha_{\mathcal{G}})^{-1} \colon U_\alpha \cap U_\beta \times G \longrightarrow U_\alpha \cap U_\beta \times G \colon (x, g) \longmapsto \Big(x, G_{\beta\alpha}(x)[g]\Big)
\end{align*}
with $G_{\beta\alpha}(x) \in \mathrm{Aut}(G)$ according to Remark \ref{oss 1.6}. This means that:
\begin{align*}
\psi^\beta_P \circ (\psi^\alpha_P)^{-1} \colon V_\alpha \cap V_\beta \times G \longrightarrow V_\alpha \cap V_\beta \times G \colon \Big([p]_{\mathcal{G}}, g\Big) \longmapsto \Big([p]_{\mathcal{G}}, h\Big)=\Big([p]_{\mathcal{G}}, \varphi\Big([p]_{\mathcal{G}}\Big) \cdot G_{\beta\alpha}(x)[g]\Big)
\end{align*}
where we have written $\varphi=\varphi\Big([p]_{\mathcal{G}}\Big)$, since $\varphi$ depends on $[p]_{\mathcal{G}}$.

Recall that on the arbitrary $V_\alpha=\pi^{-1}_{\mathcal{S}, M}(U_\alpha)$ we have coordinates given through the Quotient Manifold Theorem and the bijection $\varepsilon_\alpha \colon (\pi_{\mathcal{S}, M})^{-1}(U_\alpha) \rightarrow (U_\alpha \times F)/G$. Note also that from Remark \ref{oss 1.13}, we know that $(\mathcal{S}, M, \pi_{\mathcal{S}, M})$ is a fibered manifold, hence we can in particular fix fibered coordinates $(x^\mu, \sigma^\Theta)$ on $V_\alpha$ respecting the projection $\pi_{\mathcal{S}, M}$, which means that we want the change of coordinates on $\mathcal{S}$ to satisfy:
\begin{align}\label{10}	
		\left\{
		\begin{array}{l}
			x'^\mu=x'^\mu(x)
			\\
			\sigma'^\Theta=\Sigma^\Theta(x, \sigma)
		\end{array}
		\right.	
\end{align}
where we also fix $x'^\mu=x'^\mu(x)$ as in \eqref{1}, with $\Theta \in \{1, \dots, n-l=\mathrm{dim}(F)-\mathrm{dim}(G)\}$ (observe that this whole construction implies then that we have $\mathrm{dim}(F) \geq \mathrm{dim}(G)$), because the base space of $(\mathcal{S}, M, \pi_{\mathcal{S}, M})$ is $M$ and since from the Quotient Manifold Theorem we have:
\begin{align*}
\mathrm{dim}(\mathcal{S})&=\mathrm{dim}\Big((U_\alpha \times F)/G\Big)=\mathrm{dim}(U_\alpha)+\mathrm{dim}(F)-\mathrm{dim}(G)=\mathrm{dim}(M)+\Big(\mathrm{dim}(F)-\mathrm{dim}(G)\Big)
\end{align*}
The map $\psi^\beta_P \circ (\psi^\alpha_P)^{-1}$ written in these local coordinates reads as:
\begin{align*}
(x^\mu, \sigma^\Theta, g^I) \longmapsto \Big(x'^\mu(x), \Sigma^\Theta(x, \sigma), \pi^I\Big(\varphi(x, \sigma), G(x, g)\Big)\Big)
\end{align*}
once fixed coordinates on the Lie group $G$ also as in \eqref{1}. This gives us the transformation laws of {\em generalized principal bundle coordinates} $(x^\mu, \sigma^\Theta, g^I)$ on $P$, as we shall call them, which can be summarized as:
\begin{align}\label{11}	
                 \left\{
		\begin{array}{l}
			x'^\mu=x'^\mu(x)
			\\
			\sigma'^\Theta=\Sigma^\Theta(x, \sigma)
			\\
			g'^I=\pi^I\Big(\varphi(x, \sigma), G(x, g)\Big)
		\end{array}
		\right.
\end{align}
where $\pi^I$ are the local expressions of the product in the Lie group $G$, $\varphi(x, \sigma)=\varphi^A(x, \sigma)$ are the local expressions of the function $\varphi$ (which characterizes the generalized principal bundle) and $G(x, g)=G^B(x, g)$ are the local expressions of the transition function $G_{\beta\alpha}$ (as in \eqref{1}), with $\mu \in \{1, \dots, m\}$, $\Theta \in \{1, \dots, n-l\}$, $A, B, I \in \{1, \dots, l\}$, $\mathrm{dim}(M)=m$, $\mathrm{dim}(F)=n$, $\mathrm{dim}(G)=l$.

\begin{oss}\label{oss 1.16}
From the definition of $\psi^{\alpha}_P$, it also holds that:
\begin{equation*}
\xymatrixrowsep{0.56in}
\xymatrixcolsep{0.8in}
	\xymatrix{
	   {\pi^{-1}_{P,\mathcal{S}}(V_\alpha)} \ar[r]^{\psi^{\alpha}_P} \ar[d]_{\pi_{P, \mathcal{S}}} & {V_\alpha \times G} \ar[d]^{p_1} \\
	   {V_\alpha}  \ar[r]_{id_{V_\alpha}} \ar[d]_{\pi_{\mathcal{S}, M}} & {V_\alpha} \ar[d]^{\pi_{\mathcal{S}, M}} \\
	   {U_\alpha}  \ar[r]_{id_{U_\alpha}} & {U_\alpha} \\}
\end{equation*}
is a commutative diagram, meaning that in generalized principal bundle coordinates we have:
\begin{align*}
\pi_{P, \mathcal{S}} \colon (x^\mu, \sigma^\Theta, g^I) \longmapsto (x^\mu, \sigma^\Theta), \quad \pi_{\mathcal{S}, M} \colon (x^\mu, \sigma^\Theta) \longmapsto x^\mu
\end{align*}
where we used that $(x^\mu, \sigma^\Theta)$ are fibered coordinates on $\mathcal{S}$ because of \eqref{10}. This is a consequence of the nature of the transformation laws \eqref{11} which can be described saying that generalized principal bundle coordinates are {\em stratified fibered coordinates} on the {\em composite fibered manifold} \eqref{7}.
\end{oss}

Note that when we consider a Lie group fiber bundle as a generalized principal bundle through the identification $(\mathcal{G}, M, \pi_{\mathcal{G}, M}, G) \equiv (\mathcal{G}, \mathcal{G}/\mathcal{G}, \pi_{\mathcal{G}, \mathcal{G}/\mathcal{G}}, G)$ the relation $\sigma\Big([p]_{\mathcal{G}}\Big)=\Phi\Big(\hat{\sigma}\Big([p]_{\mathcal{G}}\Big), (\psi^\beta_{\mathcal G})^{-1}(x, \varphi)\Big)$ becomes just:
\begin{align*}
e_x=\mathcal{M}\Big(e_x, (\psi^\beta_{\mathcal G})^{-1}(x, \varphi)\Big)=e_x \cdot (\psi^\beta_{\mathcal G})^{-1}(x, \varphi)=(\psi^\beta_{\mathcal G})^{-1}(x, \varphi) \, \Longrightarrow \, \varphi=e \in G
\end{align*}
In fact, remembering that $\pi_{\mathcal{G}/\mathcal{G}, M}$ is essentially the identity, we can set all local sections $\sigma$ from Theorem \ref{thm 1.12} as the restriction of the unit section $1 \colon M \rightarrow \mathcal{G}$ on the relevant open set of $M$. Moreover, since $\mathcal{G}/\mathcal{G} \cong M$, the fibered coordinates $(x^\mu, \sigma^\Theta)$ reduce just to the coordinates $x^\mu$ on $M$, so that on the whole \eqref{11} just reads as:
\begin{align}\label{1.12}	
                 \left\{
		\begin{array}{l}
			x'^\mu=x'^\mu(x)
			\\
			g'^I=\pi^I\Big(e, G(x, g)\Big)=G^I(x, g)
		\end{array}
		\right.
\end{align}
This means that the generalized principal bundle coordinates on $(\mathcal{G}, M, \pi_{\mathcal{G}, M}, G)$ are the fibered coordinates $(x^\mu, g^I)$ satisfying \eqref{1}.

\medskip

At this point, referring to \cite[Chapter 16, Section 14]{DIEUDONNE' (1972)}, to \cite[Chapter 1, Section 5]{FATIBENE (2003)} and to \cite[Chapter 3, Section 10]{KOLAR (1993)} for the general theory of associated bundles to standard principal bundles, it is moreover possible to observe that the following holds true:
\begin{prop}\label{prop 1.3.4}
Any Lie group fiber bundle $(\mathcal{G}, M, \pi_{\mathcal{G}, M}, G)$, with connected typical fiber $G$, is an associated bundle to a certain standard principal bundle.
\end{prop}
\begin{proof}

\medskip

\noindent As a consequence of a result that can be found in \cite{HOCHSCHILD (1952)}, the automorphism group $(\mathrm{Aut}(G), \circ)$ of a connected Lie group $(G, \cdot)$ is itself a Lie group. Considering now a trivializing atlas of $(\mathcal{G}, M, \pi_{\mathcal{G}, M}, G)$ of the form $\{ (U_\alpha,\psi_{\mathcal{G}}^\alpha) \,\,|\,\, \alpha \in \mathcal{A} \}$, we have seen above that $\psi^\beta_{\mathcal G} \circ (\psi^\alpha_{\mathcal{G}})^{-1}(x, g)=\Big(x, G_{\beta\alpha}(x)[g]\Big)$, where:
\begin{align*}
G_{\beta\alpha} \colon U_\alpha \cap U_\beta \longrightarrow \mathrm{Aut}(G) \colon x \longmapsto G_{\beta\alpha}(x)
\end{align*}
are the transition functions. Thanks to the fiber bundle construction theorem, given the open covering $\{ U_\alpha \,\,|\,\, \alpha \in \mathcal{A} \}$ of $M$ and the maps:
\begin{align*}
\varphi_{\beta\alpha} \colon U_\alpha \cap U_\beta & \longrightarrow \mathrm{Aut}(G) \longhookrightarrow \mathrm{Diff}\Big(\mathrm{Aut}(G)\Big)
\\
x & \longmapsto G_{\beta\alpha}(x) \longmapsto G_{\beta\alpha}(x) \circ (-) \colon \mathrm{Aut}(G) \rightarrow \mathrm{Aut}(G) \colon a \mapsto G_{\beta\alpha}(x) \circ a
\end{align*}
satisfying the cocycle conditions, since the transition functions $G_{\beta\alpha}$ themselves must respect the cocycle conditions (see for example \cite[Chapter 3, Section 9]{KOLAR (1993)}), we obtain a (unique up to fiber bundle isomorphism) fiber bundle $(\hat{P}, M, \pi_{\hat{P}, M}, \mathrm{Aut}(G))$ with a trivializing atlas of the form $\{ (U_\alpha,\psi_{\hat{P}}^\alpha) \,\,|\,\, \alpha \in \mathcal{A} \}$, where:
\begin{align*}
\psi^\beta_{\hat{P}} \circ (\psi^\alpha_{\hat{P}})^{-1} \colon U_\alpha \cap U_\beta \times \mathrm{Aut}(G) & \longrightarrow U_\alpha \cap U_\beta \times \mathrm{Aut}(G)
\\
(x, a) & \longmapsto \Big(x, \varphi_{\beta\alpha}(x)[a]\Big)=\Big(x, G_{\beta\alpha}(x) \circ a\Big)
\end{align*}
In particular it is a standard principal bundle thanks to the nature of its transition functions $\varphi_{\beta\alpha}$, compare again with \cite{KOLAR (1993)}.

Finally we can consider the evaluation map:
\begin{align*}
\lambda \colon \mathrm{Aut}(G) \times G \longrightarrow G \colon (a, g) \longmapsto \lambda(a, g)=a(g)
\end{align*}
which gives a smooth left action of the Lie group $\mathrm{Aut}(G)$ on $G$. With respect to the left action $\lambda$, we have that the associated bundle $(\hat{P} \times_\lambda G, M, \pi_{\hat{P} \times_\lambda G, M}, G)$ to the principal bundle $(\hat{P}, M, \pi_{\hat{P}, M}, \mathrm{Aut}(G))$ has a trivializing atlas $\{ (U_\alpha, \psi_{\hat{P} \times_\lambda G}^\alpha) \,\,|\,\, \alpha \in \mathcal{A} \}$ satisfying:
\begin{align*}
\psi^\beta_{\hat{P} \times_\lambda G} \circ (\psi^\alpha_{\hat{P} \times_\lambda G})^{-1} \colon U_\alpha \cap U_\beta \times G & \longrightarrow U_\alpha \cap U_\beta \times G
\\
(x, g) & \longmapsto \Big(x, \lambda\Big(\varphi_{\beta\alpha}(x), g\Big)\Big)=\Big(x, \varphi_{\beta\alpha}(x)[g]\Big)=\Big(x, G_{\beta\alpha}(x)[g]\Big)
\end{align*}
For a proof of this property see for example again \cite{FATIBENE (2003)}.

Observing that $\psi^\beta_{\hat{P} \times_\lambda G} \circ (\psi^\alpha_{\hat{P} \times_\lambda G})^{-1}=\psi^\beta_{\mathcal G} \circ (\psi^\alpha_{\mathcal{G}})^{-1}$ and using once more the fiber bundle construction theorem, we conclude that the associated bundle $(\hat{P} \times_\lambda G, M, \pi_{\hat{P} \times_\lambda G, M}, G)$ coincides with the Lie group fiber bundle $(\mathcal{G}, M, \pi_{\mathcal{G}, M}, G)$ up to fiber bundle isomorphism. We have thus obtained the thesis.

\end{proof}

Note that the Proposition above does not contradict Remark \ref{oss 1.2.5}, since adjoint bundles are just a particular {\em example} of associated bundles to principal bundles.

\medskip

Even though we will not go into its details, an alternative proof of Proposition \ref{prop 1.3.4} could use the fact that a Lie group fiber bundle $(\mathcal{G}, M, \pi_{\mathcal{G}, M}, G)$, with connected typical fiber $G$, actually has an $\mathrm{Aut}(G)$-structure and that:
\begin{thm}\label{thm 1.3.5}
For any fiber bundle $(N, M, \pi_{N, M}, Q)$ endowed with a $G$-structure there exists a standard principal bundle $(P, M, \pi_{P, M}, G)$ such that $(N, M, \pi_{N, M}, Q)$ can be seen as an associated bundle to $(P, M, \pi_{P, M}, G)$.
\end{thm}
For the definition of {\em $G$-structure} on fiber bundles, where $G$ is a Lie group, and a proof of the preceding theorem see \cite[Chapter 3, Section 10]{KOLAR (1993)}.

Proposition \ref{prop 1.3.4} has been proved also in \cite{BLAZQUEZ-SANZ (2022)} thanks to the completely different approach that addresses Lie group fiber bundles as particular cases of Lie groupoids. Our proof, from this point of view, has the advantage of being built up from less elaborate techniques.

\medskip

The following table summarizes and compares standard and generalized principal bundles:
\begin{align*}
{\renewcommand{\arraystretch}{1.5} \renewcommand{\tabcolsep}{0.2cm}
\begin{tabular}{|c|c|} 
\hline
{\bfseries Principal bundles} & {\bfseries Generalized principal bundles} \\
\hline 
\hline
Smooth manifold $P$ & Fiber bundle $(P, M, \pi_{P, M}, F)$ \\ 
\hline
Lie group $G$ & Lie group fiber bundle $(\mathcal{G}, M, \pi_{\mathcal{G}, M}, G)$ \\ 
\hline
Right action $m \colon P \times G  \rightarrow P$ & Right fibered action $\Phi \colon P \times_M \mathcal{G} \rightarrow P$ \\
\hline
$m$ is free and proper & $\Phi$ is free and proper \\
\hline
$\Downarrow$ & $\Downarrow$ \\
$(P, S, \pi_{P, S}, G)$ is a fiber bundle, $S=P/G$ & $(P, \mathcal{S}, \pi_{P, \mathcal{S}}, G)$ is a fiber bundle, $\mathcal{S}=P/\mathcal{G}$ \\
\hline
\end{tabular}}
\end{align*}

\medskip

We conclude this section recalling that on a standard principal bundle $(P, M, \pi_{P, M}, G)$ one of the equivalent definitions of principal connection involves the assignment of a {\em connection form} which has a special behaviour on {\em fundamental fields}, i.e.\,\,particular vector fields defined as:
\begin{align*}
\xi^* \colon P \longrightarrow V_{\pi_{P, M}}(P) \colon p \longmapsto T_{e}(\mathcal{L}_p)(\xi)
\end{align*}
where $\xi$ belonging to $\mathfrak{g}$, the Lie algebra of $G$, is fixed, $V_{\pi_{P, M}}(P)$ is the vertical vector bundle on $P$ related to $\pi_{P, M}$, $\mathcal{L}_p$ is the map given by:
\begin{align*}
\mathcal{L}_p \colon G \longrightarrow P_{\pi_{P, M}(p)} \subseteq P \colon g \longmapsto m(p, g)
\end{align*}
for a fixed $p \in P$ and $m$ is the smooth right action of $G$ on $P$. It holds that $\mathcal{L}_p$ is a diffeomorphism.

\medskip

In \cite{CASTRILLON LOPEZ (2023)} an extension of this kind of vector fields to a generalized principal bundle $(P, \mathcal{S}, \pi_{P, \mathcal{S}}, G)$ obtained from a fiber bundle $(P, M, \pi_{P, M}, F)$ has been introduced. This was carried forward considering the maps:
\begin{align*}
L_{\gamma_0} \colon \mathcal{G}_x &\longrightarrow \mathcal{G}_x \colon \delta \longmapsto \mathcal{M}(\gamma_0, \delta)=\gamma_0 \cdot \delta, \quad R_{\gamma_0} \colon \mathcal{G}_x \longrightarrow \mathcal{G}_x \colon \delta \longmapsto \mathcal{M}(\delta, \gamma_0)=\delta \cdot \gamma_0
\\
\Phi_{p_0} \colon \mathcal{G}_x &\longrightarrow P_x \colon \gamma \longmapsto \Phi(p_0, \gamma), \quad \Phi_{\gamma_0} \colon P_x \longrightarrow P_x \colon p \longmapsto \Phi(p, \gamma_0)
\end{align*}
for $x \in M$, $p_0\in P_x$ and $\gamma_0 \in \mathcal{G}_x$ fixed, which are all well-posed from the definition of $\mathcal{M}$ and $\Phi$. Note that by Lemma \ref{lem 1.11} $(P_x, \mathcal{S}_x, \pi_{P_x, \mathcal{S}_x}, \mathcal{G}_x)$ is a principal bundle with action $\Phi_x$ and that then thanks to the standard principal bundle case we have that:
\begin{align*}
\Phi_{x, \, p_0}=\Phi_{p_0} \colon \mathcal{G}_x \longrightarrow {(P_{x})}_{\pi_{P_x, \mathcal{S}_x}(p_0)}={(P_{x})}_{[p_0]_{\mathcal{G}_x}} \subseteq P_x \colon \gamma \longmapsto \Phi(p_0, \gamma)
\end{align*}
is a diffeomorphism. Because $[p_0]_{\mathcal{G}} \equiv (x, [p_0]_{\mathcal{G}_x})$ we can write:
\begin{align*}
T_{e_x}(\Phi_{p_0}) \colon T_{e_x}(\mathcal{G}_x)=\mathfrak{a}_x \longrightarrow T_{p_0}(P_{[p_0]_{\mathcal{G}}}) \equiv V_{\pi_{P, \mathcal{S}}, \, p_0}(P)=\mathrm{Ker}\Big(T_{p_0}(\pi_{P, \mathcal{S}})\Big) \colon \xi \longmapsto T_{e_x}(\Phi_{p_0})(\xi)
\end{align*}
It is now possible to give the following definition:
\begin{defn}\label{defn 1.17}
Let $(P, \mathcal{S}, \pi_{P, \mathcal{S}}, G)$ be a generalized principal bundle. A {\em generalized fundamental field on} $P$ is a vector field of the form:
\begin{align*}
\xi^* \colon P_x \longrightarrow V_{\pi_{P, \mathcal{S}}}(P) \colon p \longmapsto \xi^*_p=T_{e_x}(\Phi_{p})(\xi) \in V_{\pi_{P, \mathcal{S}}, \, p}(P)
\end{align*}
for a fixed $x \in M$ and $\xi \in \mathfrak{a}_x$, the Lie algebra of $\mathcal{G}_x$.
\end{defn}

Note that $\xi^*_p \in V_{\pi_{P, \mathcal{S}}, \, p}(P)=\mathrm{Ker}\Big(T_{p}(\pi_{P, \mathcal{S}})\Big) \subseteq \mathrm{Ker}\Big(T_{p}(\pi_{P, M})\Big)=V_{\pi_{P, M}, \, p}(P)$, thanks to \eqref{7}.

\section{Local expressions for Lie group fiber bundle connections} \label{sec 2.1}

In order to extend the notion of principal connection to the generalized principal bundles $(P, \mathcal{S}, \pi_{P, \mathcal{S}}, G)$, Castrillón López and Rodríguez Abella needed to take a look at special kinds of (fiber bundle) connections on the underlying Lie group fiber bundle $(\mathcal{G}, M, \pi_{\mathcal{G}, M}, G)$, the {\em Lie group fiber bundle connections}.

\medskip

Recall that on a fiber bundle $(P, M, \pi_{P, M}, F)$ there exists a canonical way to define vertical tangent vectors at a point $p \in P$: one needs just to consider tangent vector to the fiber to which $p$ belongs. The same cannot be done for \lq\lq horizontal\rq\rq\,\,tangent vectors, which should be the \lq\lq liftings\rq\rq\,\,on $P$ of tangent vectors to $\pi_{P, M}(p)$ in the base space $M$, since there is no intrinsic way to compare tangent vectors at different points of $P$.

In order to overcome this problem we fix a {\em connection} on $(P, M, \pi_{P, M}, F)$, which really means we are fixing a choice for the definition of these horizontal tangent vectors for the fiber bundle.

\medskip

More precisely and working directly with our Lie group fiber bundle $(\mathcal{G}, M, \pi_{\mathcal{G}, M}, G)$, remember that (here we refer to \cite[Chapter 3]{FATIBENE (2003)}, to \cite[Chapter 1]{GIACHETTA (2009)} and to \cite[Chapter 2]{KOBAYASHI (1963)}) assigning a {\em connection} on $(\mathcal{G}, M, \pi_{\mathcal{G}, M}, G)$ corresponds to define for all points $\gamma \in \mathcal{G}$ a vector subspace $H_{\gamma}(\mathcal{G})$ of $T_{\gamma}(\mathcal{G})$ such that:
\begin{enumerate}
\item $\forall \, \gamma \in \mathcal{G}$, $H_{\gamma}(\mathcal{G}) \oplus V_{\gamma}(\mathcal{G})=T_{\gamma}(\mathcal{G})$, where $V_{\gamma}(\mathcal{G})=V_{\pi_{\mathcal{G}, M}, \, \gamma}(\mathcal{G})=\mathrm{Ker}\Big(T_{\gamma}(\pi_{\mathcal{G}, M})\Big)$.
\item $H(\mathcal{G})=\coprod_{\gamma \in \mathcal{G}} H_{\gamma}(\mathcal{G})$ is a sub-bundle of $T(\mathcal{G})$.
\end{enumerate}
As a consequence, it holds true that the map ${T_{\gamma}(\pi_{\mathcal{G}, M})}_{|H_{\gamma}(\mathcal{G})} \colon H_{\gamma}(\mathcal{G}) \rightarrow T_{x}(M)$ is an isomorphism of vector spaces, where $\pi_{\mathcal{G}, M}(\gamma)=x$. The inverse of this map, which is of the form $\eta_{\gamma} \colon T_{x}(M) \rightarrow H_{\gamma}(\mathcal{G})$ is called {\em horizontal lift} at $\gamma \in \mathcal{G}$. Actually, it holds also that an {\em equivalent} way to assign a connection on $(\mathcal{G}, M, \pi_{\mathcal{G}, M}, G)$ is to give for all $x \in M$, for all $\gamma \in \mathcal{G}$, such that $\pi_{\mathcal{G}, M}(\gamma)=x$, a linear map $\eta_{\gamma} \colon T_{x}(M) \rightarrow T_{\gamma}(\mathcal{G})$ with $T_{\gamma}(\pi_{\mathcal{G}, M}) \circ \eta_{\gamma}=id_{T_{x}(M)}$, depending smoothly on $\gamma$. In fact, one can then define $H_{\gamma}(\mathcal{G})=\eta_{\gamma}\Big(T_{x}(M)\Big)$.

Moreover, once one fixes a connection on $(\mathcal{G}, M, \pi_{\mathcal{G}, M}, G)$, we obtain two linear maps for all $\gamma \in \mathcal{G}$:
\begin{align*}
h_{\gamma} \colon T_{\gamma}(\mathcal{G}) \longrightarrow H_{\gamma}(\mathcal{G}) \colon v \longmapsto h_{\gamma}(v), \quad \nu_{\gamma} \colon T_{\gamma}(\mathcal{G}) \longrightarrow V_{\gamma}(\mathcal{G}) \colon v \longmapsto \nu_{\gamma}(v)
\end{align*}
the horizontal and vertical projections, respectively, obeying for all $v \in T_{\gamma}(\mathcal{G})$, $h_{\gamma}(v)+\nu_{\gamma}(v)=v$ and for all $\gamma \in \mathcal{G}$, ${\nu_{\gamma}|}_{V_{\gamma}(\mathcal{G})}=id_{V_{\gamma}(\mathcal{G})}$. Note that the relationship between $\eta_{\gamma}$ and $\nu_{\gamma}$ is given by:
\begin{align}\label{16}
{\nu_{\gamma}}=id_{T_{\gamma}(\mathcal{G})}-\Big(\eta_{\gamma} \circ T_{\gamma}(\pi_{\mathcal{G}, M})\Big)
\end{align}

The collection of the vertical projections $\nu_{\gamma}$ defines then a vector bundle morphism $\nu$ over the identity ${id}_\mathcal{G} \colon \mathcal{G} \rightarrow \mathcal{G}$:
\begin{equation*}
\xymatrixrowsep{0.56in}
\xymatrixcolsep{0.8in}
	\xymatrix{
	   {T(\mathcal{G})} \ar[r]^{\nu} \ar[d]_{\pi_{T(\mathcal{G}), \mathcal{G}}} & {V(\mathcal{G})} \ar[d]^{\pi_{V(\mathcal{G}), \mathcal{G}}} \\
	   {\mathcal{G}}  \ar[r]_{{id}_\mathcal{G}} & {\mathcal{G}}  \\}
\end{equation*}
such that ${\nu|}_{V(\mathcal{G})}=id_{V(\mathcal{G})}$. Again, an {\em equivalent} way to assign a connection on $(\mathcal{G}, M, \pi_{\mathcal{G}, M}, G)$ is to fix a vector bundle morphism $\nu$ having the properties above.

\medskip

\needspace{3\baselineskip}
Once recalled this framework, it is now possible to give the following definition (that we write down as stated in \cite{CASTRILLON LOPEZ (2023)}):
\begin{defn}\label{defn 2.1}
A {\em Lie group fiber bundle connection} on the Lie group fiber bundle $(\mathcal{G}, M, \pi_{\mathcal{G}, M}, G)$ is a connection $\nu \colon T(\mathcal{G}) \longrightarrow V(\mathcal{G})$ such that:
\begin{enumerate}
\item $\forall \, x \in M$, $\mathrm{Ker}(\nu_{e_x})=T_{x}(1)\Big(T_{x}(M)\Big)$, where we recall that $e_x$ is the identity element in $\mathcal{G}_{x}$ and $1$ is the unit section from Remark \ref{oss 1.7}, so that $T_{x}(1) \colon T_{x}(M) \longrightarrow T_{e_x}(\mathcal{G})$.
\item $\forall \, (\gamma, \delta) \in \mathcal{G} \times_M \mathcal{G}$ and $\forall \, (v, w) \in T_{(\gamma, \delta)}(\mathcal{G} \times_M \mathcal{G}) \equiv T_{\gamma}(\mathcal{G}) \times_{T_{x}(M)} T_{\delta}(\mathcal{G})$, that is $T_{\gamma}(\pi_{\mathcal{G}, M})(v)=T_{\delta}(\pi_{\mathcal{G}, M})(w)$, with $\pi_{\mathcal{G}, M}(\gamma)=\pi_{\mathcal{G}, M}(\delta)=x$:
\begin{align}\label{17}
\nu_{\gamma \cdot \delta}\Big(T_{(\gamma, \delta)}(\mathcal{M})(v, w)\Big)=T_{\gamma}(R_{\delta})\Big(\nu_{\gamma}(v)\Big)+T_{\delta}(L_{\gamma})\Big(\nu_{\delta}(w)\Big)
\end{align}
where $\mathcal{M} \colon \mathcal{G} \times_M \mathcal{G} \longrightarrow \mathcal{G}$ is the multiplication map as in Section \ref{sec 1.2} and $R_{\delta}$ and $L_{\gamma}$ are defined as in Section \ref{sec 1.3}.
\end{enumerate}
\end{defn}
It is then clear that a Lie group fiber bundle connection is essentially a connection $\nu \colon T(\mathcal{G}) \rightarrow V(\mathcal{G})$ endowed with a natural notion of compatibility with the algebraic structure of $\mathcal{G}$. 

We can here note that an equivalent notion to the one above has actually appeared independently before in \cite{BLAZQUEZ-SANZ (2022)} under the name of {\em group connection}.

\medskip

Our contribution to this topic is, firstly, to present a straightforward way (alternative to the one given in \cite{CASTRILLON LOPEZ (2023)}) to characterize Lie group fiber bundle connections in terms of horizontal lifts and, secondly and more crucially, to use the obtained characterization in order to answer the following question: what are the local conditions on a connection ensuring that it is a Lie group fiber bundle connection?

\medskip

Considering now the horizontal lifts $\eta_{\gamma} \colon T_{x}(M) \rightarrow H_{\gamma}(\mathcal{G})$ associated with a Lie group fiber bundle connection $\nu \colon T(\mathcal{G}) \rightarrow V(\mathcal{G})$, where $\gamma \in \mathcal{G}$ and $\pi_{\mathcal{G}, M}(\gamma)=x$, if we fix in \eqref{17}:
\begin{align*}
		\left\{
		\begin{array}{l}
			v=\eta_{\gamma}(a)
			\\
			w=\eta_{\delta}(a)
		\end{array}
		\right.	
\end{align*}
where $a \in T_{x}(M)$, for which it holds that:
\begin{align*}
T_{\gamma}(\pi_{\mathcal{G}, M})(v)=T_{\gamma}(\pi_{\mathcal{G}, M})\Big(\eta_{\gamma}(a)\Big)=id_{T_{x}(M)}(a)=T_{\delta}(\pi_{\mathcal{G}, M})\Big(\eta_{\delta}(a)\Big)=T_{\delta}(\pi_{\mathcal{G}, M})(w)
\end{align*}
we obtain:
\begin{align*}
\nu_{\gamma \cdot \delta}\Big(T_{(\gamma, \delta)}(\mathcal{M})\Big(\eta_{\gamma}(a), \eta_{\delta}(a)\Big)\Big)=T_{\gamma}(R_{\delta})\Big(\nu_{\gamma}\Big(\eta_{\gamma}(a)\Big)\Big)+T_{\delta}(L_{\gamma})\Big(\nu_{\delta}\Big(\eta_{\delta}(a)\Big)\Big)
\end{align*}
and then:
\begin{align*}
\nu_{\gamma \cdot \delta}\Big(T_{(\gamma, \delta)}(\mathcal{M})\Big(\eta_{\gamma}(a), \eta_{\delta}(a)\Big)\Big)=0
\end{align*}
since clearly $\nu_{\gamma} \circ \eta_{\gamma}=0$, for any $\gamma \in \mathcal{G}$. Now using \eqref{16} we obtain that:
\begin{align*}
T_{(\gamma, \delta)}(\mathcal{M})\Big(\eta_{\gamma}(a), \eta_{\delta}(a)\Big)=\Big(\eta_{\gamma \cdot \delta} \circ T_{\gamma \cdot \delta}(\pi_{\mathcal{G}, M})\Big)\Big(T_{(\gamma, \delta)}(\mathcal{M})\Big(\eta_{\gamma}(a), \eta_{\delta}(a)\Big)\Big)
\end{align*}
We can also observe that since $\pi_{\mathcal{G}, M} \circ \mathcal{M}=\pi_{\mathcal{G} \times_M \mathcal{G}, M}$, then:
\begin{align*}
T_{(\gamma, \delta)}(\mathcal{M})\Big(\eta_{\gamma}(a), \eta_{\delta}(a)\Big)&=\eta_{\gamma \cdot \delta}\Big(T_{(\gamma, \delta)}(\pi_{\mathcal{G} \times_M \mathcal{G}, M})\Big(\eta_{\gamma}(a), \eta_{\delta}(a)\Big)\Big)=\eta_{\gamma \cdot \delta}\Big(T_{\gamma}(\pi_{\mathcal{G}, M})\Big(\eta_{\gamma}(a)\Big)\Big)=\eta_{\gamma \cdot \delta}(a)
\end{align*}
and this is true for all $a \in T_{x}(M)$, hence:
\begin{align}\label{19}
\eta_{\gamma \cdot \delta}=T_{(\gamma, \delta)}(\mathcal{M})\Big(\eta_{\gamma}(\cdot), \eta_{\delta}(\cdot)\Big), \, \forall \, (\gamma, \delta) \in \mathcal{G} \times_M \mathcal{G}
\end{align}
Moreover, the condition $\mathrm{Ker}(\nu_{e_x})=T_{x}(1)\Big(T_{x}(M)\Big)$ implies that:
\begin{align*}
H_{e_x}(\mathcal{G})=T_{x}(1)\Big(T_{x}(M)\Big)
\end{align*}
and this means that for all $a \in T_{x}(M)$, $\eta_{e_x}(a)=T_{x}(1)(b)$, for some $b \in T_{x}(M)$, but since it must hold that $T_{e_x}(\pi_{\mathcal{G}, M}) \circ \eta_{e_x}=id_{T_{x}(M)}$ we have:
\begin{align*}
a=\Big(T_{e_x}(\pi_{\mathcal{G}, M}) \circ \eta_{e_x}\Big)(a)=\Big(T_{e_x}(\pi_{\mathcal{G}, M}) \circ T_{x}(1)\Big)(b)=T_{x}\Big(\pi_{\mathcal{G}, M} \circ 1\Big)(b)=T_{x}(id_{M})(b)=b
\end{align*}
and since $1$ is the unit section. Then the following holds:
\begin{align}\label{20}
\eta_{e_x}=T_{x}(1), \, \forall \, x \in M
\end{align}

\medskip

Actually conditions \eqref{19} and \eqref{20} for the horizontal lifts characterize Lie group fiber bundle connections since they are {\em equivalent} to the request 1 and 2 in Definition \ref{defn 2.1}. In fact, if \eqref{20} holds then $H_{e_x}(\mathcal{G})=\eta_{e_x}\Big(T_{x}(M)\Big)=T_{x}(1)\Big(T_{x}(M)\Big)$, that is $\mathrm{Ker}(\nu_{e_x})=T_{x}(1)\Big(T_{x}(M)\Big)$, which is request 1. Moreover, if \eqref{19} holds, then for some fixed $(v, w) \in T_{\gamma}(\mathcal{G}) \times_{T_{x}(M)} T_{\delta}(\mathcal{G})$, with $T_{\gamma}(\pi_{\mathcal{G}, M})(v)=T_{\delta}(\pi_{\mathcal{G}, M})(w)=a \in T_{x}(M)$, we have:
\begin{align*}
T_{(\gamma, \delta)}(\mathcal{M})\Big(\eta_{\gamma}(a), w\Big)-\eta_{\gamma \cdot \delta}(a)=T_{(\gamma, \delta)}(\mathcal{M})\Big(0, w-\eta_{\delta}(a)\Big)=T_{(\gamma, \delta)}(\mathcal{M})\Big(0, \nu_{\delta}(w)\Big)
\end{align*}
where the last step is due to the fact that $T_{\delta}(\pi_{\mathcal{G}, M})\Big(w-\eta_{\delta}(a)\Big)=T_{\delta}(\pi_{\mathcal{G}, M})(w)-T_{\delta}(\pi_{\mathcal{G}, M})\Big(\eta_{\delta}(a)\Big)=a-a=0$ and hence $\nu_{\delta}(w)=w-\eta_{\delta}(a)$.

Considering now a smooth curve $\Big(\alpha_1(t), \alpha_2(t)\Big) \colon I \subseteq \mathbb{R} \rightarrow \mathcal{G} \times_M \mathcal{G}$ with $\Big(\alpha_1(0), \alpha_2(0)\Big)=(\gamma, \delta)$ and $\Big(\alpha'_1(0), \alpha'_2(0)\Big)=\Big(0, \nu_{\delta}(w)\Big)$, so that we can directly fix $\alpha_1(t)=\gamma$, for each $t \in I$, we then have:
\begin{align*}
T_{(\gamma, \delta)}(\mathcal{M})\Big(0, \nu_{\delta}(w)\Big)&=\left.\frac{d}{dt}\right|_{t=0}\Big(\mathcal{M} \circ \Big(\alpha_1(t), \alpha_2(t)\Big)\Big)=\left.\frac{d}{dt}\right|_{t=0}L_{\gamma}\Big(\alpha_2(t)\Big)=T_{\delta}(L_{\gamma})\Big(\nu_{\delta}(w)\Big)
\end{align*}
This implies that:
\begin{align*}
T_{(\gamma, \delta)}(\mathcal{M})\Big(\eta_{\gamma}(a), w\Big)&=\eta_{\gamma \cdot \delta}(a)+T_{\delta}(L_{\gamma})\Big(\nu_{\delta}(w)\Big)
\end{align*}
Since, with similar steps as above, we also have that $\eta_{\gamma}(a)+\nu_{\gamma}(v)=v$ and $T_{(\gamma, \delta)}(\mathcal{M})\Big(\nu_{\gamma}(v), 0\Big)=T_{\gamma}(R_{\delta})\Big(\nu_{\gamma}(v)\Big)$, we can then write:
\begin{align*}
T_{(\gamma, \delta)}(\mathcal{M})\Big(\eta_{\gamma}(a)+\nu_{\gamma}(v), w\Big)&=T_{(\gamma, \delta)}(\mathcal{M})\Big(\eta_{\gamma}(a), w\Big)+T_{(\gamma, \delta)}(\mathcal{M})\Big(\nu_{\gamma}(v), 0\Big)
\\
&=\eta_{\gamma \cdot \delta}(a)+T_{\gamma}(R_{\delta})\Big(\nu_{\gamma}(v)\Big)+T_{\delta}(L_{\gamma})\Big(\nu_{\delta}(w)\Big)
\end{align*}
and also:
\begin{align*}
\nu_{\gamma \cdot \delta}\Big(T_{(\gamma, \delta)}(\mathcal{M})(v, w)\Big)=T_{\gamma}(R_{\delta})\Big(\nu_{\gamma}(v)\Big)+T_{\delta}(L_{\gamma})\Big(\nu_{\delta}(w)\Big)
\end{align*}
This last step is due to the fact that $\nu_{\gamma \cdot \delta} \circ \eta_{\gamma \cdot \delta}=0$ and that $R_{\delta}$ and $L_{\gamma}$ are maps along the fiber $\mathcal{G}_x$ (which means that $T_{\gamma}(R_{\delta})\Big(\nu_{\gamma}(v)\Big)+T_{\delta}(L_{\gamma})\Big(\nu_{\delta}(w)\Big)$ is already a vertical tangent vector). Therefore we have found \eqref{17}, which is request 2 in Definition \ref{defn 2.1}.

\medskip

\needspace{3\baselineskip}
Summing up, we have proved the following statement: 
\begin{prop}[Horizontal lift characterization of Lie group fiber bundle connections]\label{prop 6.3}
A Lie group fiber bundle connection on the Lie group fiber bundle $(\mathcal{G}, M, \pi_{\mathcal{G}, M}, G)$ is a connection $\eta$, given by the collection of the horizontal lifts $\eta_{\gamma} \colon T_{x}(M) \longrightarrow T_{\gamma}(\mathcal{G})$, $\forall \, x \in M$, $\forall \, \gamma \in \mathcal{G}$, with $\pi_{\mathcal{G}, M}(\gamma)=x$, such that:
\begin{enumerate}
\item $\forall \, x \in M$, $\eta_{e_x}=T_{x}(1)$, where $1$ is the unit section.
\item $\forall \, (\gamma, \delta) \in \mathcal{G} \times_M \mathcal{G}$, $\eta_{\gamma \cdot \delta}=T_{(\gamma, \delta)}(\mathcal{M})\Big(\eta_{\gamma}(\cdot), \eta_{\delta}(\cdot)\Big)$, where $\mathcal{M}$ is the multiplication map.
\end{enumerate}
\end{prop}
Further geometric interpretations of the notion of Lie group fiber bundle connection can be found already in \cite{CASTRILLON LOPEZ (2023)}.

\bigskip

Working subsequently in fibered coordinates $(x^\mu, g^I)$ satisfying \eqref{1} on the Lie group fiber bundle $(\mathcal{G}, M, \pi_{\mathcal{G}, M}, G)$, we know from standard fiber bundle theory that a connection on $(\mathcal{G}, M, \pi_{\mathcal{G}, M}, G)$ given in terms of the smooth assignment of horizontal lifts $\eta_{\gamma}$, $\gamma \in \mathcal{G}$, can be locally written as:
\begin{align*}
\eta=dx^\mu \otimes \Big(\partial_{\mu}-\eta^I_\mu(x, g)\partial_{I}\Big)
\end{align*}
meaning that if $\pi_{\mathcal{G}, M}(\gamma)=x$, $a=a^\mu\partial_{\mu} \in T_{x}(M)$ and in local coordinates $\gamma$ is given by $(x^\mu, g^I)$, then:
\begin{align}\label{22}
\eta_{\gamma}(a^\mu\partial_{\mu})=a^\mu\Big(\partial_{\mu}-\eta^I_\mu(x, g)\partial_{I}\Big) \in H_{\gamma}(\mathcal{G})
\end{align}
where we set $\partial_{\mu}=\frac{\partial}{\partial x^\mu}$, $\partial_{I}=\frac{\partial}{\partial g^I}$, and $\eta^I_\mu(x, g)$ are the {\em connection coefficients}.

These coefficients define the connection and, if we are considering a Lie group fiber bundle connection, are locally free once also applied the behaviour with respect to conditions \eqref{19} and \eqref{20}, with the exception of the non-local restrictions given by the transformation laws under change of charts.

For what concerns \eqref{20}, remembering that in fibered coordinates $1 \colon x^\mu \mapsto (x^\mu, e^I)$, it is clear that:
\begin{align*}
T_{x}(1) \colon T_{x}(M) \longrightarrow T_{e_x}(\mathcal{G}) \colon a^\mu\partial_{\mu} \longmapsto a^\mu\partial_{\mu}
\end{align*}
Then working on the coordinate basis of $T_{x}(M)$ gives:
\begin{align*}
\partial_{\mu}-\eta^I_\mu(x, e)\partial_{I}=\eta_{e_x}(\partial_{\mu})=T_{x}(1)(\partial_{\mu})=\partial_{\mu} \, \Longleftrightarrow \, \eta^I_\mu(x, e)=0
\end{align*}

For what concerns \eqref{19}, we note that in fibered coordinates the multiplication map $\mathcal{M}$ takes the form:
\begin{align*}
\mathcal{M} \colon \mathcal{G} \times_M \mathcal{G} \longrightarrow \mathcal{G} \colon (x^\mu, g^I, h^J) \longmapsto \Big(x^\mu, \pi^I(g, h)\Big)
\end{align*}
where $\pi \colon G \times G \rightarrow G$ is the multiplication in the Lie group $G$. This holds true since, once taken a local trivialization $\psi \colon (\pi_{\mathcal{G}, M})^{-1}(U) \rightarrow U \times G$ for $(\mathcal{G}, M, \pi_{\mathcal{G}, M}, G)$ around $\gamma, \delta \in \mathcal{G}_{x}$, we have that:
\begin{enumerate}
\item If the local coordinates of $\gamma$ and $\delta$ through this trivialization are respectively $(x^\mu, g^I)$ and $(x^\mu, h^J)$, then the local coordinates of $(\gamma, \delta) \in \mathcal{G} \times_M \mathcal{G}$ are $(x^\mu, g^I, h^J)$, because:
\begin{align*}
\hat{\psi} \colon (\pi_{\mathcal{G} \times_M \mathcal{G}, M})^{-1}(U) \longrightarrow U \times G \times G \colon (\gamma, \delta) \longmapsto (x, g, h)
\end{align*}
is the local trivialization of $(\mathcal{G} \times_M \mathcal{G}, M, \pi_{\mathcal{G} \times_M \mathcal{G}, M}, G \times G)$ associated with $\psi$, where $\psi(\gamma)=(x, g)$ and $\psi(\delta)=(x, h)$.
\item It holds that $\psi(\gamma \cdot \delta)=\psi|_{\mathcal{G}_{x}}(\gamma \cdot \delta)=\pi\Big(\psi|_{\mathcal{G}_{x}}(\gamma), \psi|_{\mathcal{G}_{x}}(\delta)\Big)=\pi\Big(\psi(\gamma), \psi(\delta)\Big) \equiv \Big(x, \pi(g, h)\Big)$ by Definition \ref{defn 1.5}.
\end{enumerate}
From this we obtain that:
\begin{align*}
T_{(\gamma, \delta)}(\mathcal{M}) \colon T_{\gamma}(\mathcal{G}) \times_{T_{x}(M)} T_{\delta}(\mathcal{G}) \equiv T_{(\gamma, \delta)}(\mathcal{G} \times_M \mathcal{G})  & \longrightarrow T_{\gamma \cdot \delta}(\mathcal{G})
\\
(a^\mu\partial_{\mu}+b^I\partial_{I}, a^\mu\partial_{\mu}+c^I\partial_{I}) \equiv a^\mu\partial_{\mu}+b^I\partial_{I}+c^I\hat{\partial}_{I} & \longmapsto a^\mu\partial_{\mu}+\Big(b^A\partial^1_A\pi^I(g, h)+c^A\partial^2_A\pi^I(g, h)\Big)\partial_{I}
\end{align*}
where $\partial_{I}=\frac{\partial}{\partial g^I}$, $\hat{\partial}_{I}=\frac{\partial}{\partial h^I}$ and $\partial^1_A\pi^I(g, h)$, $\partial^2_A\pi^I(g, h)$ are the partial derivatives with respect to $g^A$ and $h^A$, respectively.

By combining this result with \eqref{19} and \eqref{22} on the coordinate basis of $T_{x}(M)$ we have:
\begin{align*}
&& \eta_{\gamma \cdot \delta}(\partial_{\mu})&=T_{(\gamma, \delta)}(\mathcal{M})\Big(\eta_{\gamma}(\partial_{\mu}), \eta_{\delta}(\partial_{\mu})\Big)
\\
\, \Longleftrightarrow \, &&\partial_{\mu}-\eta^I_\mu\Big(x, \pi(g, h)\Big)\partial_{I}&=T_{(\gamma, \delta)}(\mathcal{M})\Big(\partial_{\mu}-\eta^I_\mu(x, g)\partial_{I}, \partial_{\mu}-\eta^I_\mu(x, h)\partial_{I}
\Big) 
\\
&& &=T_{(\gamma, \delta)}(\mathcal{M})\Big(\partial_{\mu}-\eta^I_\mu(x, g)\partial_{I}-\eta^I_\mu(x, h)\hat{\partial}_{I}
\Big) 
\\
&& &=\partial_{\mu}+\Big(-\eta^A_\mu(x, g)\partial^1_A\pi^I(g, h)-\eta^A_\mu(x, h)\partial^2_A\pi^I(g, h)\Big)\partial_{I}
\\
\, \Longleftrightarrow \, &&\eta^I_\mu\Big(x, \pi(g, h)\Big)&=\eta^A_\mu(x, g)\partial^1_A\pi^I(g, h)+\eta^A_\mu(x, h)\partial^2_A\pi^I(g, h)
\end{align*}

We can summarize the result of the computations above with the following statement:
\begin{thm}[Local coordinate characterization of Lie group fiber bundle connections]\label{thm 6.4}
The connection coefficients $\eta^I_\mu(x, g)$ of a Lie group fiber bundle connection, that we shall call {\em Lie group fiber bundle connection coefficients}, are characterized, in any fibered coordinates $(x^\mu, g^I)$ satisfying \eqref{1}, by the conditions:
\begin{enumerate}
\item $\eta^I_\mu(x, e)=0$, where $e$ is the identity element of the Lie group $G$.
\item $\forall \, g, h \in G$, $\eta^I_\mu\Big(x, \pi(g, h)\Big)=\eta^A_\mu(x, g)\partial^1_A\pi^I(g, h)+\eta^A_\mu(x, h)\partial^2_A\pi^I(g, h)$, where $\pi^I(g, h)$ are the local expressions of the product in the Lie group $G$.
\end{enumerate}
\end{thm}

\begin{oss}\label{oss 2.3}
The fact that the conditions above are well-posed is a direct consequence of the global nature of request \eqref{19} and \eqref{20}. It is, nevertheless, very instructive to explicitly check the invariance of these local conditions under change of coordinates \eqref{1}, see Appendix \ref{app A}.
\end{oss}

\begin{ex}\label{ex 2.2.2}
Since we have seen in Example \ref{ex 1.2.3} that vector bundles are examples of Lie group fiber bundles, we can ask ourselves: what is a Lie group fiber bundle connection on a vector bundle $(E, M, \pi_{E, M}, \mathbb{R}^l)$? 

On this fiber bundle we have fibered coordinates $(x^\mu, v^I)$ satisfying:
\begin{align} \label{2.9}
		\left\{
		\begin{array}{l}
			x'^\mu=x'^\mu(x)
			\\
			v'^I=G^I_J(x)v^J
		\end{array}
		\right.	
\end{align}
where $G^I_J(x) \in \mathrm{GL}(l, \mathbb{R})$, since the transition functions take values in $\mathrm{GL}(l, \mathbb{R})$ (compare with \cite[Chapter 1, Section 3]{FATIBENE (2003)}). Note that conditions \eqref{2} and \eqref{3} are here fulfilled, since the typical fiber is the Lie group $(\mathbb{R}^l, +)$ with identity element $0$ and thus the local expressions of the product $\pi$ are:
\begin{align*}
\pi^I(v, w)=v^I+w^I
\end{align*}
In particular then $\partial^1_A\pi^I(v, w)=\frac{\partial}{\partial v^A}\pi^I(v, w)=\frac{\partial}{\partial v^A}[v^I+w^I]=\delta^I_A$ and $\partial^2_A\pi^I(v, w)=\frac{\partial}{\partial w^A}\pi^I(v, w)=\frac{\partial}{\partial w^A}[v^I+w^I]=\delta^I_A$.

From the discussion above and working in fibered coordinates $(x^\mu, v^I)$, a Lie group fiber bundle connection on $(E, M, \pi_{E, M}, \mathbb{R}^l)$ can be locally written as:
\begin{align*}
\eta=dx^\mu \otimes \Big(\partial_{\mu}-\eta^I_\mu(x, v)\partial_{I}\Big)
\end{align*}
where we set $\partial_{\mu}=\frac{\partial}{\partial x^\mu}$ and $\partial_{I}=\frac{\partial}{\partial v^I}$, with the conditions:
\begin{align*}
\eta^I_\mu(x, 0)&=0
\end{align*}
and:
\begin{align*}
&&\eta^I_\mu\Big(x, \pi(v, w)\Big)&=\eta^A_\mu(x, v)\partial^1_A\pi^I(v, w)+\eta^A_\mu(x, w)\partial^2_A\pi^I(v, w)
\\
\, \Longleftrightarrow \, &&\eta^I_\mu(x, v+w)&=\eta^A_\mu(x, v)\delta^I_A+\eta^A_\mu(x, w)\delta^I_A=\eta^I_\mu(x, v)+\eta^I_\mu(x, w)
\end{align*}
Taking into account the smoothness of $\eta$, this implies that $\eta^I_\mu(x, v)=\eta^I_{J\mu}(x)v^J$ or:
\begin{align*}
\eta=dx^\mu \otimes \Big(\partial_{\mu}-\eta^I_{J\mu}(x)v^J\partial_{I}\Big)
\end{align*}
for some real-valued smooth local functions $\eta^I_{J\mu}(x)$.

The conclusion is that the Lie group fiber bundle connections on the vector bundle $(E, M, \pi_{E, M}, \mathbb{R}^l)$ are exactly the {\em linear connections} on it, as observed without using local coordinates also in \cite{CASTRILLON LOPEZ (2023)}.
\end{ex}

\medskip

We have proved in Proposition \ref{prop 1.3.4} that any Lie group fiber bundle $(\mathcal{G}, M, \pi_{\mathcal{G}, M}, G)$ (with connected typical fiber $G$) is an associated bundle to a certain standard principal bundle. It would be then finally natural to study if:
\begin{conj}\label{conj 2.2.3}
Any Lie group fiber bundle connection $\eta$ on a Lie group fiber bundle $(\mathcal{G}, M, \pi_{\mathcal{G}, M}, G)$ (with connected typical fiber $G$) can be seen as an associated connection to a principal connection on the principal bundle $(\hat{P}, M, \pi_{\hat{P}, M}, \mathrm{Aut}(G))$ constructed in the proof of Proposition \ref{prop 1.3.4}.
\end{conj}
For the general theory of associated connections see for example \cite[Chapter 3, Section 6]{FATIBENE (2003)} and \cite[Chapter 3, Section 11]{KOLAR (1993)}. To our knowledge, this particular problem has been addressed only in \cite{BLAZQUEZ-SANZ (2022)}, through Lie groupoid arguments, with a positive answer. 

It remains to be seen if the starting hypotheses inside \cite{BLAZQUEZ-SANZ (2022)} and the present work are exactly the same, or if there is some relevant difference. Nevertheless, even though we will not give deeper insights here, it can be expected that results regarding connections on fiber bundles with $G$-structures (since we have seen that $(\mathcal{G}, M, \pi_{\mathcal{G}, M}, G)$ has an $\mathrm{Aut}(G)$-structure) and other classical theorems on the recognition of associated connections (see for example again \cite{KOLAR (1993)}) could settle the matter without bringing up Lie groupoid theory. We suspect Conjecture \ref{conj 2.2.3} to be true under our hypotheses since:
\begin{enumerate}
\item From Example \ref{ex 2.2.2} we know in particular that a Lie group fiber bundle connection on a tangent bundle $(T(M), M, \pi_{T(M), M}, \mathbb{R}^m)$, which has connected typical fiber, is a linear connection.
\item Linear connections on $(T(M), M, \pi_{T(M), M}, \mathbb{R}^m)$ are associated connections, once we see the tangent bundle as an associated bundle to the {\em principal frame bundle} $(L(M), M, \pi_{L(M), M}, \mathrm{GL}(m, \mathbb{R}))$ in the standard way (see on this topic for example \cite[Chapter 1, Section 5]{FATIBENE (2003)}).
\item It can be easily checked that, in the case of the tangent bundle $T(M)$, the principal bundle $(\hat{P}, M, \pi_{\hat{P}, M}, \mathrm{Aut}(G))$ constructed in the proof of Proposition \ref{prop 1.3.4} is the frame bundle $L(M)$ itself.
\end{enumerate}
\needspace{3\baselineskip}
For the moment, we can observe that the following result provides the converse to Conjecture \ref{conj 2.2.3}:
\begin{lem}\label{lem 5.7}
Given an associated connection $\eta$ on the associated bundle $(\hat{P} \times_\lambda G, M, \pi_{\hat{P} \times_\lambda G, M}, G)$ obtained from a fixed principal connection $\psi$ on the principal bundle $(\hat{P}, M, \pi_{\hat{P}, M}, \mathrm{Aut}(G))$ constructed in the proof of Proposition \ref{prop 1.3.4}, it holds that $\eta$ is a Lie group fiber bundle connection.
\end{lem}
\begin{proof}

\medskip

\noindent Let us indicate by $m \colon \hat{P} \times \mathrm{Aut}(G) \rightarrow \hat{P}$ the relevant smooth right action for the principal bundle $(\hat{P}, M, \pi_{\hat{P}, M}, \mathrm{Aut}(G))$. A point in $(\hat{P} \times_\lambda G, M, \pi_{\hat{P} \times_\lambda G, M}, G)$ (which we know is a Lie group fiber bundle thanks to Proposition \ref{prop 1.3.4}) will be then denoted by:
\begin{align*}
[p, g]=\Big\{ \Big(m(p, a^{-1}), \lambda(a, g)\Big)=\Big(m(p, a^{-1}), a(g)\Big) \,\,\Big|\,\, a \in \mathrm{Aut}(G) \Big\} 
\end{align*}
for $p \in \hat{P}$ and $g \in G$ (see for example \cite[Chapter 1, Section 5]{FATIBENE (2003)}). It holds also that $\pi_{\hat{P} \times_\lambda G, M}\Big([p, g]\Big)=\pi_{\hat{P}, M}(p)$. Note that for the fiber $(\hat{P} \times_\lambda G)_x=(\pi_{\hat{P} \times_\lambda G, M})^{-1}(x)=\{ [p, g] \,\,|\,\, \pi_{\hat{P}, M}(p)=x \} $, given $x \in M$, the identity element is:
\begin{align*}
[p, e]=\Big\{ \Big(m(p, a^{-1}), a(e)\Big)=\Big(m(p, a^{-1}), e\Big) \,\,\Big|\,\, a \in \mathrm{Aut}(G) \Big\}=\{ (q, e) \,\,|\,\, q \in \hat{P}_x \}
\end{align*}
for any $p \in \hat{P}_x$ (recall that $e$ is the identity element of the Lie group $G$), and the product is given by:
\begin{align}\label{200}
[p, g] \cdot [q, h]=[p, g] \cdot [p, a_{p, q}(h)]=\Big[p, \pi\Big(g, a_{p, q}(h)\Big)\Big]
\end{align}
as there exists a unique $a_{p, q} \in \mathrm{Aut}(G)$ such that $m(q, a_{p, q}^{-1})=p$, since $p, q \in \hat{P}_x$ (recall that $\pi$ is the multiplication in the Lie group $G$). This means that, in this case, the unit section takes the form:
\begin{align*}
1 \colon M \longrightarrow \hat{P} \times_\lambda G \colon x \longmapsto [p, e]
\end{align*}
for any $p \in \hat{P}_x$, and the multiplication map is:
\begin{align*}
\mathcal{M} \colon (\hat{P} \times_\lambda G) \times_M (\hat{P} \times_\lambda G) \longrightarrow \hat{P} \times_\lambda G \colon \Big([p, g], [q, h]\Big) \longmapsto \mathcal{M}\Big([p, g], [q, h]\Big)=\Big[p, \pi\Big(g, a_{p, q}(h)\Big)\Big]
\end{align*}

Let us now fix a principal connection $\psi$ on the principal bundle $(\hat{P}, M, \pi_{\hat{P}, M}, \mathrm{Aut}(G))$, given by the collection of the horizontal lifts $\psi_{p} \colon T_{x}(M) \rightarrow T_{p}(\hat{P})$, $\forall \, x \in M$, $\forall \, p \in \hat{P}$, with $\pi_{\hat{P}, M}(p)=x$. The corresponding associated connection $\eta$ on $(\hat{P} \times_\lambda G, M, \pi_{\hat{P} \times_\lambda G, M}, G)$ is given by the collection of the horizontal lifts $\eta_{[p, g]} \colon T_{x}(M) \rightarrow T_{[p, g]}(\hat{P} \times_\lambda G)$, $\forall \, x \in M$, $\forall \, [p, g] \in \hat{P} \times_\lambda G$, with $\pi_{\hat{P} \times_\lambda G, M}\Big([p, g]\Big)=x$, defined by:
\begin{align*}
\eta_{[p, g]}=T_{p}(\theta_{g})\Big(\psi_{p}(\cdot)\Big), \, \forall \, [p, g] \in \hat{P} \times_\lambda G
\end{align*}
which is well-posed and where $\theta_{g} \colon \hat{P} \rightarrow \hat{P} \times_\lambda G \colon p \mapsto [p, g]$ (see \cite[Chapter 3, Section 6]{FATIBENE (2003)}).

In order to prove that $\eta$ is a Lie group fiber bundle connection we just need to show (according to Proposition \ref{prop 6.3}) that:
\begin{enumerate}
\item $\forall \, x \in M$, $\eta_{[p, e]}=T_{x}(1)$, with $\pi_{\hat{P}, M}(p)=x$, which is obtained observing that $\theta_{e}=1 \circ \pi_{\hat{P}, M}$, since $\theta_{e} \colon \hat{P} \rightarrow \hat{P} \times_\lambda G \colon p \mapsto [p, e]$. Hence fixing $v \in T_{x}(M)$:
\begin{align*}
\eta_{[p, e]}(v)=T_{p}(\theta_{e})\Big(\psi_{p}(v)\Big)=T_{p}(1 \circ \pi_{\hat{P}, M})\Big(\psi_{p}(v)\Big)=T_{x}(1)\Big(T_{p}(\pi_{\hat{P}, M})\Big(\psi_{p}(v)\Big)\Big)=T_{x}(1)(v)
\end{align*}
where the last step is due to the definition of horizontal lift.
\item $\forall \, \Big([p, g], [q, h]\Big) \in (\hat{P} \times_\lambda G) \times_M (\hat{P} \times_\lambda G)$, $\eta_{[p, g] \cdot [q, h]}=T_{([p, g], [q, h])}(\mathcal{M})\Big(\eta_{[p, g]}(\cdot), \eta_{[q, h]}(\cdot)\Big)$, which is obtained recalling that we can always write $[q, h]=[p, a_{p, q}(h)]$ for a unique $a_{p, q} \in \mathrm{Aut}(G)$. Let us then set $k=a_{p, q}(h)$, so that $[q, h]=[p, k]$. Hence fixing $v \in T_{x}(M)$:
\begin{align*}
T_{([p, g], [p, k])}(\mathcal{M})\Big(\eta_{[p, g]}(v), \eta_{[p, k]}(v)\Big)&=T_{([p, g], [p, k])}(\mathcal{M})\Big(T_{p}(\theta_{g})\Big(\psi_{p}(v)\Big), T_{p}(\theta_{k})\Big(\psi_{p}(v)\Big)\Big)
\\
&=T_{([p, g], [p, k])}(\mathcal{M})\Big(T_{p}(\theta_{g, k})\Big(\psi_{p}(v)\Big)\Big)=T_{p}(\theta_{\pi(g, k)})\Big(\psi_{p}(v)\Big)
\\
&=\eta_{[p, \pi(g, k)]}(v)=\eta_{ [p, g] \cdot [p, k]}(v)
\end{align*}
since, once set $\theta_{g, k}=(\theta_{g}, \theta_{k}) \colon \hat{P} \rightarrow (\hat{P} \times_\lambda G) \times_M (\hat{P} \times_\lambda G) \colon p \mapsto \Big([p, g], [p, k]\Big)$, we have that $\mathcal{M} \circ \theta_{g, k}=\theta_{\pi(g, k)}$ by \eqref{200}.
\end{enumerate}

\end{proof}

\section{Local expressions for generalized principal connections} \label{sec 2.3}

So far it has been given a new perspective to the necessary background in order to treat {\em generalized principal connections}, which will be now studied following the same line of approach. 

As above, we are considering a generalized principal bundle $(P, \mathcal{S}, \pi_{P, \mathcal{S}}, G)$ obtained from a free and proper right fibered action $\Phi \colon P \times_M \mathcal{G} \rightarrow P$ of the Lie group fiber bundle $(\mathcal{G}, M, \pi_{\mathcal{G}, M}, G)$ on the fiber bundle $(P, M, \pi_{P, M}, F)$.

\medskip

Similarly to what done in Section \ref{sec 2.1} in the case of the Lie group fiber bundle $(\mathcal{G}, M, \pi_{\mathcal{G}, M}, G)$, a connection on $(P, \mathcal{S}, \pi_{P, \mathcal{S}}, G)$ (taking into account just its fiber bundle structure) can be seen as a vector bundle morphism $\omega$ over the identity ${id}_P \colon P \rightarrow P$:
\begin{equation*}
\xymatrixrowsep{0.56in}
\xymatrixcolsep{0.8in}
	\xymatrix{
	   {T(P)} \ar[r]^{\omega} \ar[d]_{\pi_{T(P), P}} & {V_{\pi_{P, \mathcal{S}}}(P)} \ar[d]^{\pi_{V_{\pi_{P, \mathcal{S}}}(P), P}} \\
	   {P}  \ar[r]_{{id}_P} & {P}  \\}
\end{equation*}
such that ${\omega|}_{V_{\pi_{P, \mathcal{S}}}(P)}=id_{V_{\pi_{P, \mathcal{S}}}(P)}$.

\medskip

As already announced, generalized principal connections on $(P, \mathcal{S}, \pi_{P, \mathcal{S}}, G)$ are established on the previous choice of a ({\em a priori} not necessarily Lie group fiber bundle) connection on the Lie group fiber bundle $(\mathcal{G}, M, \pi_{\mathcal{G}, M}, G)$ and this happens in the following particular manner (that we write down as stated in \cite{CASTRILLON LOPEZ (2023)}):
\begin{defn}\label{defn 2.6}
Let $\nu \colon T(\mathcal{G}) \longrightarrow V(\mathcal{G})$ be a fixed connection on the relevant Lie group fiber bundle $(\mathcal{G}, M, \pi_{\mathcal{G}, M}, G)$. A {\em generalized principal connection} on the generalized principal bundle $(P, \mathcal{S}, \pi_{P, \mathcal{S}}, G)$ associated to $\nu$ is a connection $\omega \colon T(P) \longrightarrow V_{\pi_{P, \mathcal{S}}}(P)$ such that:
\begin{align}\label{26}
\omega_{\Phi(p, \gamma)}\Big(T_{(p, \gamma)}(\Phi)(z, v)\Big)=T_{p}(\Phi _{\gamma})\Big(\omega_p(z)\Big)+\Big(T_{\gamma}(L_{\gamma^{-1}})\Big(\nu_{\gamma}(v)\Big)\Big)^*_{\Phi(p, \gamma)}
\end{align}
$\forall \, (p, \gamma) \in P \times_M \mathcal{G}$ and  $\forall \, (z, v) \in T_{(p, \gamma)}(P \times_M \mathcal{G}) \equiv T_{p}(P) \times_{T_{x}(M)} T_{\gamma}(\mathcal{G})$, that is $T_{p}(\pi_{P, M})(z)=T_{\gamma}(\pi_{\mathcal{G}, M})(v)$, with $\pi_{P, M}(p)=\pi_{\mathcal{G}, M}(\gamma)=x$, where $\Phi _{\gamma}$ and $L_{\gamma^{-1}}$ are defined as in Section \ref{sec 1.3}.
\end{defn}
\needspace{3\baselineskip}
Furthermore, it is important to observe that an existence theorem for generalized principal connections holds true:
\begin{thm}\label{thm 2.8}
If $M$ is a paracompact manifold, then there exist a connection $\nu \colon T(\mathcal{G}) \longrightarrow V(\mathcal{G})$ on the Lie group fiber bundle $(\mathcal{G}, M, \pi_{\mathcal{G}, M}, G)$ and a generalized principal connection $\omega \colon T(P) \longrightarrow V_{\pi_{P, \mathcal{S}}}(P)$,  associated to $\nu$, on the generalized principal bundle $(P, \mathcal{S}, \pi_{P, \mathcal{S}}, G)$.
\end{thm}
Note that the proof (compare with \cite{CASTRILLON LOPEZ (2023)}) is a local construction relying on a partition of unity argument, as in the classical case of standard principal connections on standard principal bundles (for which we refer for example to \cite[Chapter 2, Section 2]{KOBAYASHI (1963)}).

\bigskip

Altogether, the peculiarity in the new equivariance formula \eqref{26} is given by the additional term containing the connection $\nu$ on the Lie group fiber bundle $(\mathcal{G}, M, \pi_{\mathcal{G}, M}, G)$, as it can be deduced from Appendix \ref{app B}. This connection must actually be a Lie group fiber bundle connection, as we will see below, even though it is not an {\em a priori} requirement.

\medskip

We will give three contributions to this topic:
\begin{enumerate}
\item Firstly, we present a straightforward way (alternative to the one given in \cite{CASTRILLON LOPEZ (2023)}) to characterize generalized principal connections in terms of horizontal lifts.
\item Secondly, we use the obtained characterization in order to answer a question with the same flavour as the one stated in Section \ref{sec 2.1}: what are the local conditions on a connection ensuring that it is a generalized principal connection?
\item Thirdly, we exhibit a new proof (derived from the aforementioned local conditions) of the fact that generalized principal connections are associated {\em only} to Lie group fiber bundle connections (the original proof of this result can be found in \cite{CASTRILLON LOPEZ (2023)}).
\end{enumerate}

Along the lines of what done in Section \ref{sec 2.1}, consider now the horizontal lifts $\tau_{p} \colon T_{[p]_{\mathcal{G}}}(\mathcal{S}) \rightarrow H_{\pi_{P, \mathcal{S}}, p}(P)$ associated with a generalized principal connection $\omega \colon T(P) \rightarrow V_{\pi_{P, \mathcal{S}}}(P)$, where $p \in P$ and $\pi_{P, \mathcal{S}}(p)=[p]_{\mathcal{G}}$, for which it holds that:
\begin{align}
T_{p}(\pi_{P, \mathcal{S}}) \circ \tau_{p}&=id_{T_{[p]_{\mathcal{G}}}(\mathcal{S})} \label{29}
\\
\omega_{p}&=id_{T_{p}(P)}-\Big(\tau_{p} \circ T_{p}(\pi_{P, \mathcal{S}})\Big) \label{30}
\end{align}
If we fix in \eqref{26}:
\begin{align*}
		\left\{
		\begin{array}{l}
			z=\tau_{p}(a)
			\\
			v=\eta_{\gamma}(b)
		\end{array}
		\right.	
\end{align*}
where $\eta_{\gamma} \colon T_{x}(M) \rightarrow H_{\gamma}(\mathcal{G})$ are the horizontal lifts associated with the connection $\nu \colon T(\mathcal{G}) \rightarrow V(\mathcal{G})$, $\pi_{P, M}(p)=\pi_{\mathcal{G}, M}(\gamma)=x$ and $a \in T_{[p]_{\mathcal{G}}}(\mathcal{S})$ and $b \in T_{x}(M)$ are such that $T_{[p]_{\mathcal{G}}}(\pi_{\mathcal{S}, M})(a)=b$, for which, thanks to \eqref{7} and \eqref{29}, the following holds true:
\begin{align*}
T_{p}(\pi_{P, M})(z)&=T_{[p]_{\mathcal{G}}}(\pi_{\mathcal{S}, M})\Big(T_{p}(\pi_{P, \mathcal{S}})\Big(\tau_{p}(a)\Big)\Big)=T_{[p]_{\mathcal{G}}}(\pi_{\mathcal{S}, M})(a)=b
\\
&=T_{\gamma}(\pi_{\mathcal{G}, M})\Big(\eta_{\gamma}(b)\Big)=T_{\gamma}(\pi_{\mathcal{G}, M})(v)
\end{align*}
we obtain:
\begin{align*}
\omega_{\Phi(p, \gamma)}\Big(T_{(p, \gamma)}(\Phi)\Big(\tau_{p}(a), \eta_{\gamma}(b)\Big)\Big)=T_{p}(\Phi _{\gamma})\Big(\omega_p\Big(\tau_{p}(a)\Big)\Big)+\Big(T_{\gamma}(L_{\gamma^{-1}})\Big(\nu_{\gamma}\Big(\eta_{\gamma}(b)\Big)\Big)\Big)^*_{\Phi(p, \gamma)}
\end{align*}
and then:
\begin{align*}
\omega_{\Phi(p, \gamma)}\Big(T_{(p, \gamma)}(\Phi)\Big(\tau_{p}(a), \eta_{\gamma}(b)\Big)\Big)=0
\end{align*}
since clearly $\omega_p \circ \tau_{p}=0$ and $\nu_{\gamma} \circ \eta_{\gamma}=0$, for any $p \in P$ and $\gamma \in \mathcal{G}$. Now using \eqref{30} we have:
\begin{align*}
T_{(p, \gamma)}(\Phi)\Big(\tau_{p}(a), \eta_{\gamma}(b)\Big)=\Big(\tau_{\Phi(p, \gamma)} \circ T_{\Phi(p, \gamma)}(\pi_{P, \mathcal{S}})\Big)\Big(T_{(p, \gamma)}(\Phi)\Big(\tau_{p}(a), \eta_{\gamma}(b)\Big)\Big)
\end{align*}
We can also observe that, by definition, for all $(p, \gamma) \in P \times_M \mathcal{G}$:
\begin{align*}
\Big(\pi_{P, \mathcal{S}} \circ \Phi\Big)(p, \gamma)=[\Phi(p, \gamma)]_{\mathcal{G}}=[p]_{\mathcal{G}}=\Big(\pi_{P, \mathcal{S}} \circ p_1\Big)(p, \gamma)
\end{align*}
where $p_1$ is the projection on the first factor, so that:
\begin{align*}
\Big(T_{\Phi(p, \gamma)}(\pi_{P, \mathcal{S}}) \circ T_{(p, \gamma)}(\Phi)\Big)\Big(\tau_{p}(a), \eta_{\gamma}(b)\Big)&=\Big(T_{p}(\pi_{P, \mathcal{S}}) \circ T_{(p, \gamma)}(p_1)\Big)\Big(\tau_{p}(a), \eta_{\gamma}(b)\Big)
\\
&=T_{p}(\pi_{P, \mathcal{S}})\Big(\tau_{p}(a)\Big)=a
\end{align*}
Then:
\begin{align*}
T_{(p, \gamma)}(\Phi)\Big(\tau_{p}(a), \eta_{\gamma}\Big(T_{[p]_{\mathcal{G}}}(\pi_{\mathcal{S}, M})(a)\Big)\Big)=T_{(p, \gamma)}(\Phi)\Big(\tau_{p}(a), \eta_{\gamma}(b)\Big)=\tau_{\Phi(p, \gamma)}(a)
\end{align*}
and this is true for all $a \in T_{[p]_{\mathcal{G}}}(\mathcal{S})$, hence:
\begin{align}\label{31}
\tau_{\Phi(p, \gamma)}=T_{(p, \gamma)}(\Phi)\Big(\tau_{p}(\cdot), \eta_{\gamma}\Big(T_{[p]_{\mathcal{G}}}(\pi_{\mathcal{S}, M})(\cdot)\Big)\Big), \, \forall \, (p, \gamma) \in P \times_M \mathcal{G}
\end{align}

\medskip

Actually condition \eqref{31} for the horizontal lifts characterizes generalized principal connections. In fact, if the condition holds true for some fixed $(z, v) \in T_{p}(P) \times_{T_{x}(M)} T_{\gamma}(\mathcal{G})$, with $T_{p}(\pi_{P, M})(z)=T_{\gamma}(\pi_{\mathcal{G}, M})(v)=b \in T_{x}(M)$ and $T_{p}(\pi_{P, \mathcal{S}})(z)=a \in T_{[p]_{\mathcal{G}}}(\mathcal{S})$, we have:
\begin{align*}
T_{(p, \gamma)}(\Phi)\Big(\tau_{p}(a), v\Big)-\tau_{\Phi(p, \gamma)}(a)&=T_{(p, \gamma)}(\Phi)\Big(\tau_{p}(a), v\Big)-T_{(p, \gamma)}(\Phi)\Big(\tau_{p}(a), \eta_{\gamma}\Big(T_{[p]_{\mathcal{G}}}(\pi_{\mathcal{S}, M})(a)\Big)\Big)
\\
&=T_{(p, \gamma)}(\Phi)\Big(0, v-\eta_{\gamma}\Big(T_{[p]_{\mathcal{G}}}(\pi_{\mathcal{S}, M})(a)\Big)\Big)=T_{(p, \gamma)}(\Phi)\Big(0, \nu_{\gamma}(v)\Big)
\end{align*}
where the last step is due to the fact that:
\begin{align*}
T_{\gamma}(\pi_{\mathcal{G}, M})\Big(v-\eta_{\gamma}\Big(T_{[p]_{\mathcal{G}}}(\pi_{\mathcal{S}, M})(a)\Big)\Big)&=T_{\gamma}(\pi_{\mathcal{G}, M})(v)-T_{\gamma}(\pi_{\mathcal{G}, M})\Big(\eta_{\gamma}\Big(T_{[p]_{\mathcal{G}}}(\pi_{\mathcal{S}, M})(a)\Big)\Big)
\\
&=b-T_{[p]_{\mathcal{G}}}(\pi_{\mathcal{S}, M})(a)=b-T_{[p]_{\mathcal{G}}}(\pi_{\mathcal{S}, M})\Big(T_{p}(\pi_{P, \mathcal{S}})(z)\Big)
\\
&=b-T_{p}(\pi_{P, M})(z)=0
\end{align*}
and hence $\nu_{\gamma}(v)=v-\eta_{\gamma}\Big(T_{[p]_{\mathcal{G}}}(\pi_{\mathcal{S}, M})(a)\Big)$. 

Considering now a smooth curve $\Big(\alpha_1(t), \alpha_2(t)\Big) \colon I \subseteq \mathbb{R} \rightarrow P \times_M \mathcal{G}$ with $\Big(\alpha_1(0), \alpha_2(0)\Big)=(p, \gamma)$ and $\Big(\alpha'_1(0), \alpha'_2(0)\Big)=\Big(0, \nu_{\gamma}(v)\Big)$, so that we can directly fix $\alpha_1(t)=p$, for each $t \in I$, we then have:
\begin{align*}
T_{(p, \gamma)}(\Phi)\Big(0, \nu_{\gamma}(v)\Big)&=\left.\frac{d}{dt}\right|_{t=0}\Big(\Phi \circ \Big(\alpha_1(t), \alpha_2(t)\Big)\Big)=\left.\frac{d}{dt}\right|_{t=0}\Phi_{p}\Big(\alpha_2(t)\Big)=T_{\gamma}(\Phi_{p})\Big(\nu_{\gamma}(v)\Big)
\end{align*}
Note also that from the definition of $\Phi$ we have:
\begin{align*}
\Big(T_{\gamma}(L_{\gamma^{-1}})\Big(\nu_{\gamma}(v)\Big)\Big)^*_{\Phi(p, \gamma)}&=T_{e_x}\Big(\Phi_{\Phi(p, \gamma)}\Big)\Big(T_{\gamma}(L_{\gamma^{-1}})\Big(\nu_{\gamma}(v)\Big)\Big)=T_{\gamma}\Big(\Phi_{\Phi(p, \gamma)} \circ L_{\gamma^{-1}}\Big)\Big(\nu_{\gamma}(v)\Big)
\\
&=T_{\gamma}(\Phi_{p})\Big(\nu_{\gamma}(v)\Big)
\end{align*}
All these observations imply that:
\begin{align*}
T_{(p, \gamma)}(\Phi)\Big(\tau_{p}(a), v\Big)=\tau_{\Phi(p, \gamma)}(a)+\Big(T_{\gamma}(L_{\gamma^{-1}})\Big(\nu_{\gamma}(v)\Big)\Big)^*_{\Phi(p, \gamma)}
\end{align*}
and because, with similar steps as above, we also have that $\tau_{p}(a)+\omega_{p}(z)=z$ and $T_{(p, \gamma)}(\Phi)\Big(\omega_{p}(z), 0\Big)=T_{p}(\Phi_{\gamma})\Big(\omega_{p}(z)\Big)$, then:
\begin{align*}
T_{(p, \gamma)}(\Phi)\Big(\tau_{p}(a)+\omega_{p}(z), v\Big)&=T_{(p, \gamma)}(\Phi)\Big(\tau_{p}(a), v\Big)+T_{(p, \gamma)}(\Phi)\Big(\omega_{p}(z), 0\Big)
\\
&=\tau_{\Phi(p, \gamma)}(a)+T_{p}(\Phi_{\gamma})\Big(\omega_{p}(z)\Big)+\Big(T_{\gamma}(L_{\gamma^{-1}})\Big(\nu_{\gamma}(v)\Big)\Big)^*_{\Phi(p, \gamma)}
\end{align*}
and therefore:
\begin{align*}
\omega_{\Phi(p, \gamma)}\Big(T_{(p, \gamma)}(\Phi)(z, v)\Big)=T_{p}(\Phi_{\gamma})\Big(\omega_{p}(z)\Big)+\Big(T_{\gamma}(L_{\gamma^{-1}})\Big(\nu_{\gamma}(v)\Big)\Big)^*_{\Phi(p, \gamma)}
\end{align*}
This last step is due to the fact that $\omega_{\Phi(p, \gamma)} \circ \tau_{\Phi(p, \gamma)}=0$ and that since Definition \ref{defn 2.6} is well-posed, as it is checked in Appendix \ref{app B} in \eqref{27}, then $T_{p}(\Phi_{\gamma})\Big(\omega_{p}(z)\Big)+\Big(T_{\gamma}(L_{\gamma^{-1}})\Big(\nu_{\gamma}(v)\Big)\Big)^*_{\Phi(p, \gamma)}$ is already a vertical tangent vector. In this way we have found \eqref{26} and, with the computations carried out, it has been established another way to write it:
\begin{align*}
\omega_{\Phi(p, \gamma)}\Big(T_{(p, \gamma)}(\Phi)(z, v)\Big)&=T_{p}(\Phi_{\gamma})\Big(\omega_{p}(z)\Big)+T_{\gamma}(\Phi_{p})\Big(\nu_{\gamma}(v)\Big)
\end{align*}

\medskip

Summing up, we have proved the following statement: 
\begin{prop}[Horizontal lift characterization of generalized principal connections]\label{prop 7.4}
Let $\eta$ be a connection on the Lie group fiber bundle $(\mathcal{G}, M, \pi_{\mathcal{G}, M}, G)$. A generalized principal connection on the generalized principal bundle $(P, \mathcal{S}, \pi_{P, \mathcal{S}}, G)$ associated to $\eta$ is a connection $\tau$, given by the collection of the horizontal lifts $\tau_{p} \colon T_{[p]_{\mathcal{G}}}(\mathcal{S}) \longrightarrow T_{p}(P)$, $\forall \, [p]_{\mathcal{G}} \in \mathcal{S}$, $\forall \, p \in P$, with $\pi_{P, \mathcal{S}}(p)=[p]_{\mathcal{G}}$, such that:
\begin{align*}
\tau_{\Phi(p, \gamma)}=T_{(p, \gamma)}(\Phi)\Big(\tau_{p}(\cdot), \eta_{\gamma}\Big(T_{[p]_{\mathcal{G}}}(\pi_{\mathcal{S}, M})(\cdot)\Big)\Big), \, \forall \, (p, \gamma) \in P \times_M \mathcal{G}
\end{align*}
where $\Phi$ is the relevant right fibered action.
\end{prop}
Further geometric interpretations of the notion of generalized principal connection can be found again in \cite{CASTRILLON LOPEZ (2023)}.

\bigskip

Working subsequently in fibered generalized principal bundle coordinates $(x^\mu, \sigma^\Theta, g^I)$ satisfying \eqref{11} on the generalized principal bundle $(P, \mathcal{S}, \pi_{P, \mathcal{S}}, G)$, remembering that $(x^\mu, \sigma^\Theta)$ are fibered coordinates on $(\mathcal{S}, M, \pi_{\mathcal{S}, M})$, we know that, as previously seen in Section \ref{sec 2.1}, a connection on $(P, \mathcal{S}, \pi_{P, \mathcal{S}}, G)$ given in terms of the smooth assignment of horizontal lifts $\tau_{p}$, $p \in P$, can be locally written as:
\begin{align*}
\tau=dx^\mu \otimes \Big(\partial_{\mu}-A^I_\mu(x, \sigma, g)\partial_{I}\Big)+d\sigma^\Theta \otimes \Big(\partial_{\Theta}-A^I_\Theta(x, \sigma, g)\partial_{I}\Big)
\end{align*}
meaning that if $\pi_{P, \mathcal{S}}(p)=[p]_{\mathcal{G}}$, $a=a^\mu\partial_{\mu}+a^\Theta\partial_{\Theta} \in T_{[p]_{\mathcal{G}}}(\mathcal{S})$ and in local coordinates $p$ is given by $(x^\mu, \sigma^\Theta, g^I)$, then:
\begin{align}\label{33}
\tau_{p}(a^\mu\partial_{\mu}+a^\Theta\partial_{\Theta})=a^\mu\Big(\partial_{\mu}-A^I_\mu(x, \sigma, g)\partial_{I}\Big)+a^\Theta\Big(\partial_{\Theta}-A^I_\Theta(x, \sigma, g)\partial_{I}\Big) \in H_{\pi_{P, \mathcal{S}}, p}(P)
\end{align}
where we set $\partial_{\mu}=\frac{\partial}{\partial x^\mu}$, $\partial_{\Theta}=\frac{\partial}{\partial \sigma^\Theta}$, $\partial_{I}=\frac{\partial}{\partial g^I}$, and $\{A^I_\mu(x, \sigma, g), A^I_\Theta(x, \sigma, g)\}$ are the connection coefficients. 

Again, these coefficients define the connection and, if we are considering a generalized principal connection, are locally free once also applied the behaviour with respect to condition \eqref{31}, with the exception of the non-local restrictions given by the transformation laws under change of charts.

In Remark \ref{oss 1.16} we have already noted that, in fibered coordinates as above, $\pi_{\mathcal{S}, M} \colon (x^\mu, \sigma^\Theta) \mapsto x^\mu$, so that for all $[p]_{\mathcal{G}} \in \mathcal{S}$:
\begin{align*}
T_{[p]_{\mathcal{G}}}(\pi_{\mathcal{S}, M}) \colon T_{[p]_{\mathcal{G}}}(\mathcal{S}) \longrightarrow T_{x}(M) \colon a^\mu\partial_{\mu}+a^\Theta\partial_{\Theta} \longmapsto a^\mu\partial_{\mu}
\end{align*}
where $\pi_{\mathcal{S}, M}\Big([p]_{\mathcal{G}}\Big)=\pi_{P, M}(p)=x$.

Moreover, thanks to Proposition \ref{prop 1.14}, we have that the fibered action $\Phi$ is locally given by the right multiplication on the Lie group $G$:
\begin{align*}
\Phi \colon P \times_M \mathcal{G} \longrightarrow P \colon (x^\mu, \sigma^\Theta, g^I, h^J) \longmapsto \Big(x^\mu, \sigma^\Theta, \pi^I(g, h)\Big)
\end{align*}
where $\pi \colon G \times G \rightarrow G$ is the multiplication in the Lie group $G$, $(x^\mu, \sigma^\Theta, g^I)$ are generalized principal bundle coordinates around $p \in P_x$ and $(x^\mu, h^J)$ are coordinates around $\gamma \in \mathcal{G}_x$ satisfying \eqref{1} (similarly to what done in Section \ref{sec 2.1} for the multiplication map $\mathcal{M}$). Note that the coordinates $(x^\mu, \sigma^\Theta, g^I, h^J)$ transform collectively as:
\begin{align}\label{34}	
                 \left\{
		\begin{array}{l}
			x'^\mu=x'^\mu(x)
			\\
			\sigma'^\Theta=\Sigma^\Theta(x, \sigma)
			\\
			g'^I=\pi^I\Big(\varphi(x, \sigma), G(x, g)\Big)
			\\
			h'^J=G^J(x, h)
		\end{array}
		\right.
\end{align}
see \cite[Chapter 1, Section 4]{SAUNDERS (1989)}.

From this we obtain that:
\begin{align*}
T_{(p, \gamma)}(\Phi) \colon T_{(p, \gamma)}(P \times_M \mathcal{G})  & \longrightarrow T_{\Phi(p, \gamma)}(P)
\\
a^\mu\partial_{\mu}+b^\Theta\partial_{\Theta}+c^I\partial_{I}+d^I\hat{\partial}_{I} & \longmapsto a^\mu\partial_{\mu}+b^\Theta\partial_{\Theta}+\Big(c^A\partial^1_A\pi^I(g, h)+d^A\partial^2_A\pi^I(g, h)\Big)\partial_{I}
\end{align*}
where $\hat{\partial}_{I}=\frac{\partial}{\partial h^I}$ and $\partial^1_A\pi^I(g, h)$, $\partial^2_A\pi^I(g, h)$ are the partial derivatives with respect to $g^A$ and $h^A$, respectively. It holds also the identification:
\begin{align*}
T_{(p, \gamma)}(P \times_M \mathcal{G}) &\equiv T_{p}(P) \times_{T_{x}(M)} T_{\gamma}(\mathcal{G}) 
\\
a^\mu\partial_{\mu}+b^\Theta\partial_{\Theta}+c^I\partial_{I}+d^I\hat{\partial}_{I}  &\equiv (a^\mu\partial_{\mu}+b^\Theta\partial_{\Theta}+c^I\partial_{I}, a^\mu\partial_{\mu}+d^I\hat{\partial}_{I})
\end{align*}
Combining these results with \eqref{22}, \eqref{31} and \eqref{33} on the coordinate basis of $T_{[p]_{\mathcal{G}}}(\mathcal{S})$ we have:
\begin{align*}
&&\tau_{\Phi(p, \gamma)}(\partial_{\mu})&=T_{(p, \gamma)}(\Phi)\Big(\tau_{p}(\partial_{\mu}), \eta_{\gamma}(\partial_{\mu})\Big)
\\
\, \Longleftrightarrow \, &&\partial_{\mu}-A^I_\mu\Big(x, \sigma, \pi(g, h)\Big)\partial_{I}&=T_{(p, \gamma)}(\Phi)\Big(\partial_{\mu}-A^I_\mu(x, \sigma, g)\partial_{I},  \partial_{\mu}-\eta^I_\mu(x, h)\hat{\partial}_{I}\Big)
\\
&& &=T_{(p, \gamma)}(\Phi)\Big(\partial_{\mu}-A^I_\mu(x, \sigma, g)\partial_{I}-\eta^I_\mu(x, h)\hat{\partial}_{I}\Big)
\\
&& &=\partial_{\mu}+\Big(-A^A_\mu(x, \sigma, g)\partial^1_A\pi^I(g, h)-\eta^A_\mu(x, h)\partial^2_A\pi^I(g, h)\Big)\partial_{I}
\\
\, \Longleftrightarrow \, &&A^I_\mu\Big(x, \sigma, \pi(g, h)\Big)&=A^A_\mu(x, \sigma, g)\partial^1_A\pi^I(g, h)+\eta^A_\mu(x, h)\partial^2_A\pi^I(g, h)
\end{align*}
and:
\begin{align*}
&&\tau_{\Phi(p, \gamma)}(\partial_{\Theta})&=T_{(p, \gamma)}(\Phi)\Big(\tau_{p}(\partial_{\Theta}), \eta_{\gamma}(0)\Big)
\\
\, \Longleftrightarrow \, &&\partial_{\Theta}-A^I_\Theta\Big(x, \sigma, \pi(g, h)\Big)\partial_{I}&=T_{(p, \gamma)}(\Phi)\Big(\partial_{\Theta}-A^I_\Theta(x, \sigma, g)\partial_{I}, 0\Big)
\\
&& &=T_{(p, \gamma)}(\Phi)\Big(\partial_{\Theta}-A^I_\Theta(x, \sigma, g)\partial_{I}\Big)
\\
&& &=\partial_{\Theta}+\Big(-A^A_\Theta(x, \sigma, g)\partial^1_A\pi^I(g, h)\Big)\partial_{I}
\\
\, \Longleftrightarrow \, &&A^I_\Theta\Big(x, \sigma, \pi(g, h)\Big)&=A^A_\Theta(x, \sigma, g)\partial^1_A\pi^I(g, h)
\end{align*}

We can summarize the result of the computations above with the following statement:
\begin{thm}[Local coordinate characterization of generalized principal connections]\label{thm 7.6}
The connection coefficients $\{A^I_\mu(x, \sigma, g), A^I_\Theta(x, \sigma, g)\}$ of a generalized principal connection, that we shall call {\em generalized principal connection coefficients}, are characterized, in any generalized principal bundle coordinates $(x^\mu, \sigma^\Theta, g^I)$, by the conditions for all $g, h \in G$:
\begin{enumerate}
\item $A^I_\mu\Big(x, \sigma, \pi(g, h)\Big)=A^A_\mu(x, \sigma, g)\partial^1_A\pi^I(g, h)+\eta^A_\mu(x, h)\partial^2_A\pi^I(g, h)$.
\item $A^I_\Theta\Big(x, \sigma, \pi(g, h)\Big)=A^A_\Theta(x, \sigma, g)\partial^1_A\pi^I(g, h)$.
\end{enumerate}
Here $\eta^A_\mu(x, h)$ are the connection coefficients of the fixed connection on the Lie group fiber bundle $(\mathcal{G}, M, \pi_{\mathcal{G}, M}, G)$ and $\pi^I(g, h)$ are the local expressions of the product in the Lie group $G$.
\end{thm}

Note that in the expressions above, thanks to \eqref{34}, the coordinates $(x, \sigma, g)$ transform as \eqref{11} and the coordinates $(x, h)$ as \eqref{1}.

\medskip

From this local characterization we finally derive that the fixed connection on the Lie group fiber bundle $(\mathcal{G}, M, \pi_{\mathcal{G}, M}, G)$ must actually be a Lie group fiber bundle connection. In fact, if in the first of the above conditions we fix $g=e$, we get:
\begin{align}\label{35}
A^I_\mu\Big(x, \sigma, \pi(e, h)\Big)&=A^A_\mu(x, \sigma, e)\partial^1_A\pi^I(e, h)+\eta^A_\mu(x, h)\partial^2_A\pi^I(e, h) \, \Longleftrightarrow \, \nonumber
\\
A^I_\mu(x, \sigma, h)&=A^A_\mu(x, \sigma, e)\partial^1_A\pi^I(e, h)+\eta^A_\mu(x, h)\delta^I_A=A^A_\mu(x, \sigma, e)\partial^1_A\pi^I(e, h)+\eta^I_\mu(x, h)
\end{align}
where we have used that $\pi$ is the multiplication map in $G$, that $e$ is the identity element of $G$ and that:
\begin{align*}
\partial^2_A\pi^I(e, h)=\frac{\partial}{\partial h^A}\pi^I(e, h)=\frac{\partial}{\partial h^A}h^I=\delta^I_A
\end{align*}
Fixing now also $h=e$ and observing that for any $g \in G$ similarly $\partial^1_A\pi^I(g, e)=\delta^I_A$, so that also $\partial^1_A\pi^I(e, e)=\delta^I_A$, we obtain:
\begin{align*}
A^I_\mu(x, \sigma, e)=A^A_\mu(x, \sigma, e)\partial^1_A\pi^I(e, e)+\eta^I_\mu(x, e)=A^I_\mu(x, \sigma, e)+\eta^I_\mu(x, e) \, \Longleftrightarrow \, \eta^I_\mu(x, e)=0
\end{align*}
which is the first condition in order to have a Lie group fiber bundle connection, as seen in Theorem \ref{thm 6.4}.

Moreover, substituting in \eqref{35} $h$ with $\pi(g, h)$ and subsequently using Theorem \ref{thm 7.6} and \eqref{35} itself, we get:
\begin{align*}
&&A^I_\mu\Big(x, \sigma, \pi(g, h)\Big)&=A^A_\mu(x, \sigma, e)\partial^1_A\pi^I\Big(e, \pi(g, h)\Big)+\eta^I_\mu\Big(x, \pi(g, h)\Big)
\\
\, \Longleftrightarrow \, &&\eta^I_\mu\Big(x, \pi(g, h)\Big)&=A^A_\mu(x, \sigma, g)\partial^1_A\pi^I(g, h)+\eta^A_\mu(x, h)\partial^2_A\pi^I(g, h)-A^A_\mu(x, \sigma, e)\partial^1_A\pi^I\Big(e, \pi(g, h)\Big)
\\
&& &=\Big(A^B_\mu(x, \sigma, e)\partial^1_B\pi^A(e, g)+\eta^A_\mu(x, g)\Big)\partial^1_A\pi^I(g, h)+\eta^A_\mu(x, h)\partial^2_A\pi^I(g, h)
\\
&& &\,\,\,\,\,\,\,-A^A_\mu(x, \sigma, e)\partial^1_A\pi^I\Big(e, \pi(g, h)\Big)
\\
&& &=\eta^A_\mu(x, g)\partial^1_A\pi^I(g, h)+\eta^A_\mu(x, h)\partial^2_A\pi^I(g, h)
\\
&& &\,\,\,\,\,\,\,+A^A_\mu(x, \sigma, e)\Big(\partial^1_A\pi^B(e, g)\partial^1_B\pi^I(g, h)-\partial^1_A\pi^I\Big(e, \pi(g, h)\Big)\Big)
\end{align*}
Note that thanks to the associativity of the product in $G$:
\begin{align*}
\pi^I\Big(k, \pi(g, h)\Big)=\pi^I\Big(\pi(k, g), h\Big) \, &\Longrightarrow \, \frac{\partial}{\partial k^A}\pi^I\Big(k, \pi(g, h)\Big)=\frac{\partial}{\partial k^A}\pi^I\Big(\pi(k, g), h\Big)
\\
&\Longrightarrow \, \partial^1_A\pi^I\Big(k, \pi(g, h)\Big)=\partial^1_B\pi^I\Big(\pi(k, g), h\Big)\partial^1_A\pi^B(k, g)
\end{align*}
which evaluated in $k=e$ gives:
\begin{align*}
\partial^1_A\pi^I\Big(e, \pi(g, h)\Big)&=\partial^1_B\pi^I\Big(\pi(e, g), h\Big)\partial^1_A\pi^B(e, g)=\partial^1_B\pi^I(g, h)\partial^1_A\pi^B(e, g)
\end{align*}
Inserting this result in the previous computation leads us to the second condition in order to have a Lie group fiber bundle connection, as seen in Theorem \ref{thm 6.4}. We have therefore finally proved that the fixed connection $\eta$ on the Lie group fiber bundle $(\mathcal{G}, M, \pi_{\mathcal{G}, M}, G)$ must be compatible with the algebraic structure of $\mathcal{G}$.

\medskip

If we fix $g=e$ in the second of the conditions characterizing generalized principal connections, we get:
\begin{align}\label{36}
A^I_\Theta(x, \sigma, h)=A^A_\Theta(x, \sigma, e)\partial^1_A\pi^I(e, h)
\end{align}
Actually, conditions \eqref{35} and \eqref{36} together with the assumption that $\eta$ is a Lie group fiber bundle connection are equivalent to the local characterizations of generalized principal connections, as seen in Theorem \ref{thm 7.6}, as it is easy to check using again that $\partial^1_A\pi^I\Big(e, \pi(g, h)\Big)=\partial^1_B\pi^I(g, h)\partial^1_A\pi^B(e, g)$.

Summarizing, we have a generalized principal connection environment if and only if:
\begin{align*}
A^I_\mu(x, \sigma, g)&=A^A_\mu(x, \sigma, e)\partial^1_A\pi^I(e, g)+\eta^I_\mu(x, g)
\\
A^I_\Theta(x, \sigma, g)&=A^A_\Theta(x, \sigma, e)\partial^1_A\pi^I(e, g)
\\
\eta^I_\mu\Big(x, \pi(g, h)\Big)&=\eta^A_\mu(x, g)\partial^1_A\pi^I(g, h)+\eta^A_\mu(x, h)\partial^2_A\pi^I(g, h)
\end{align*}
where we have omitted $\eta^I_\mu(x, e)=0$ since it has been shown that it is implied by the first of the above expressions. As a consequence, analogously to standard principal connections, we know all the generalized principal connections coefficients once we know both their behaviour in $(x, \sigma, e)$ and the Lie group fiber bundle connection $\eta$, which is the true novelty with respect to the usual theory, as already noted before. 

\medskip

As a side note, we point out that it would be worthwhile to understand if there is a relationship between generalized principal connections and {\em fibered connections} on fibered manifolds, as defined in \cite{FERRARIS (1984)}. In fact, both notions rely on the interplay of two distinct connections on some different fiber bundles (or fibered manifolds).

\begin{oss}\label{oss 2.9}
Again, the fact that the conditions above are well-posed is a direct consequence of the global nature of request \eqref{31}. It is, nevertheless, again instructive to explicitly check the invariance of the local conditions \eqref{35} and \eqref{36} under change of coordinates \eqref{1} and \eqref{11} (the invariance of the third conditions is checked in Appendix \ref{app A}), see Appendix \ref{app C}. 
\end{oss}

\begin{ex}\label{ex 2.4.2}
Since we have seen through Examples \ref{ex 1.2.3} and \ref{ex 1.3.2} that vector bundles are examples of generalized principal bundles arising from a free and proper fibered action of the vector bundles on themselves, we might wonder: what is a generalized principal connection on a vector bundle $(E, M, \pi_{E, M}, \mathbb{R}^l)$ associated to a fixed Lie group fiber bundle connection on $(E, M, \pi_{E, M}, \mathbb{R}^l)$ itself?

In Example \ref{ex 2.2.2} we have already proved that any Lie group fiber bundle connection on the vector bundle $(E, M, \pi_{E, M}, \mathbb{R}^l)$ is a linear connection that can be locally written in fibered coordinates $(x^\mu, v^I)$ satisfying \eqref{2.9} as:
\begin{align*}
\eta=dx^\mu \otimes \Big(\partial_{\mu}-\eta^I_{J\mu}(x)v^J\partial_{I}\Big)
\end{align*}
and that $\partial^1_A\pi^I(v, w)=\delta^I_A$ and $\partial^2_A\pi^I(v, w)=\delta^I_A$, since $\pi^I(v, w)=v^I+w^I$.

From the discussion above and still working in fibered coordinates $(x^\mu, v^I)$ (which are the generalized principal bundle coordinates of this case, thanks to the reasoning in Section \ref{sec 1.3} leading to \eqref{1.12}), a generalized principal connection on $(E, M, \pi_{E, M}, \mathbb{R}^l)$ can be locally written as:
\begin{align*}
\tau=dx^\mu \otimes \Big(\partial_{\mu}-A^I_\mu(x, v)\partial_{I}\Big)
\end{align*}
where we set $\partial_{\mu}=\frac{\partial}{\partial x^\mu}$ and $\partial_{I}=\frac{\partial}{\partial v^I}$, with the condition:
\begin{align*}
&&A^I_\mu\Big(x, \pi(v, w)\Big)&=A^A_\mu(x, v)\partial^1_A\pi^I(v, w)+\eta^A_\mu(x, w)\partial^2_A\pi^I(v, w)
\\
\, \Longleftrightarrow \, &&A^I_\mu(x, v+w)&=A^A_\mu(x, v)\delta^I_A+\eta^A_\mu(x, w)\delta^I_A=A^I_\mu(x, v)+\eta^I_\mu(x, w)
\end{align*}
where $\eta^I_\mu(x, w)=\eta^I_{J\mu}(x)w^J$ are the connection coefficients of the fixed Lie group fiber bundle connection. Fixing $v=0 \in \mathbb{R}^l$ and once set $\sigma^I_\mu(x)=A^I_\mu(x, 0)$, we get that:
\begin{align*}
 A^I_\mu(x, w)=\sigma^I_\mu(x)+\eta^I_{J\mu}(x)w^J \, \Longleftrightarrow \, \tau=dx^\mu \otimes \Big(\partial_{\mu}-\Big(\sigma^I_\mu(x)+\eta^I_{J\mu}(x)v^J\Big)\partial_{I}\Big)
\end{align*}

The conclusion is that the generalized principal connections on a vector bundle $(E, M, \pi_{E, M}, \mathbb{R}^l)$ are exactly the {\em affine connections} on it, where the respective underlying linear connections are the fixed Lie group fiber bundle connections on $(E, M, \pi_{E, M}, \mathbb{R}^l)$ itself.

This is an analogous result to one that can be found again in \cite{CASTRILLON LOPEZ (2023)}, nonetheless we proved it here in a completely different way.

\medskip

For the general theory of linear and affine connections we refer for example to \cite[Chapter 3]{KOBAYASHI (1963)}, to the whole article \cite{MODUGNO (1991)} and for a local coordinate approach to \cite[Chapter 1, Section 3]{GIACHETTA (2009)}.
\end{ex}

\section{Standard principal bundles as generalized principal bundles}\label{sec 2.5}

In this section we will use the local coordinate approach outlined above in order to prove that the construction reviewed in this article is eventually an instance of generalization of the usual notion of principal bundle and of principal connection, another key result present in \cite{CASTRILLON LOPEZ (2023)}. This confirms a natural intuition since the new geometric structures have been assembled having in mind the standard theory, as it is clear looking at Appendix \ref{app B}. Especially it will be shown that generalized principal connections reduce to usual principal connections on standard principal bundles.

\medskip

The main advantage of the preliminary work done to achieve the relevant local conditions is that the proofs will be direct and short, while Castrillón López and Rodríguez Abella had to carry out a complete check starting from definitions, in particular for what regards the aforementioned result concerning connections. Of course, the two proof strategies highlight different aspects of the geometric structure under study and they shall be useful in complementary issues.

\medskip

Let $(P, M, \pi_{P, M}, G)$ be now a standard principal bundle with smooth right action $m \colon P \times G \rightarrow P$. Remember that $M \cong S=P/G$, as seen in Section \ref{sec 1.1}.

We can consider the {\em trivial Lie group fiber bundle} $(M \times G, M, \pi_{M \times G, M}, G)$, where the (global) trivialization is given by $id_{M \times G}$ and where:
\begin{align*}
\pi_{M \times G, M} \colon M \times G \longrightarrow M \colon (x, g) \longmapsto x
\end{align*}
so that $(M \times G)_x=\{x\} \times G$, which has $(x, e)$ as identity element. In particular this means that $G^I(x, g)=g^I$. 

\medskip

Now we can define a right fibered action (as it can be easily checked) of the trivial Lie group fiber bundle $(M \times G, M, \pi_{M \times G, M}, G)$ on the standard principal bundle $(P, M, \pi_{P, M}, G)$:
\begin{align*}
\Phi_{m} \colon P \times_M (M \times G) \longrightarrow P \colon \Big(p, (x, g)\Big) \longmapsto \Phi_{m}\Big(p, (x, g)\Big)=m(p, g)
\end{align*}
Since $m$ must be free and proper, it holds as well that $\Phi_{m}$ is free and proper as a fibered action.

\medskip

Thanks to Theorem \ref{thm 1.12}, this means that $(P, P/(M \times G), \pi_{P, P/(M \times G)}, G)$ is a generalized principal bundle. Furthermore the equivalence relation $\sim_{M \times G}$ on $P$ defining $P/(M \times G)$ is given by:
\begin{align*}
p_1 \sim_{M \times G} p_2  \, &\Longleftrightarrow \, \exists \, x \in M, \, \exists \, (x, g) \in \{x\} \times G \, | \, p_1, p_2 \in P_x, \, \Phi_{m}\Big(p_1, (x, g)\Big)=p_2 
\\
&\Longleftrightarrow \, \exists \, x \in M, \, \exists \, g \in G \, | \, p_1, p_2 \in P_x, \, m(p_1, g)=p_2
\\
&\Longleftrightarrow \, \exists \, g \in G \, | \, m(p_1, g)=p_2 \, \Longleftrightarrow \, p_1 \sim_{G} p_2
\end{align*}
where $\sim_{G}$ is the equivalence relation induced by the right action $m$ on $P$. In this sense at a set level $P/(M \times G)=P/G=S$ and since $S \cong M$ then $\pi_{P/(M \times G), M}$ in \eqref{7} is essentially the identity which is bijective and a submersion from Remark \ref{oss 1.13}, hence a diffeomorphism. As a consequence $P/(M \times G)=S \cong M$ also at a differentiable manifold level and $(P, M, \pi_{P, M}, G) \equiv (P, P/(M \times G), \pi_{P, P/(M \times G)}, G)$.

The conclusion is that we can always regard a standard principal bundle $(P, M, \pi_{P, M}, G)$ as a generalized principal bundle. 

\medskip

Since $P/(M \times G) \cong M$, the fibered coordinates $(x^\mu, \sigma^\Theta)$ introduced in Section \ref{sec 1.3} reduce just to the coordinates $x^\mu$ on $M$ and hence the transformation laws \eqref{11} of generalized principal bundle coordinates $(x^\mu, g^I)$ on $P$ are simply:
\begin{align*}
                 \left\{
		\begin{array}{l}
			x'^\mu=x'^\mu(x)
			\\
			g'^I=\pi^I\Big(\varphi(x), g\Big)
		\end{array}
		\right.	
\end{align*}
We have used here that $G(x, g)=g$. These results can be rephrased saying that generalized principal bundle coordinates on a standard principal bundle are actually standard principal bundle coordinates given by \eqref{100}.

\medskip

We fix now on the trivial Lie group fiber bundle $(M \times G, M, \pi_{M \times G, M}, G)$ the {\em trivial Lie group fiber bundle connection} $\nu_0 \colon T(M \times G) \rightarrow V(M \times G)$ given by the linear maps:
\begin{align*}
(\nu_{0})_{(x, g)} \colon T_{(x, g)}(M \times G) \longrightarrow V_{(x, g)}(M \times G) \colon (a, v) \longmapsto (\nu_{0})_{(x, g)}(a, v)=(0, v)
\end{align*}
for all $(x, g) \in M \times G$ and for all $(a, v) \in T_{(x, g)}(M \times G) \equiv T_x(M) \times T_{g}(G)$. 

\medskip

At the level of the Lie group fiber bundle connection coefficients, this means that we are fixing always $(\eta_{0})^I_\mu(x, g)=0$. In fact, from \eqref{16} we know that for the horizontal lifts $(\eta_{0})_{(x, g)}$ associated with the trivial Lie group fiber bundle connection it holds that for all $(x, g) \in M \times G$ and for all $(a, v) \in T_{(x, g)}(M \times G) \equiv T_x(M) \times T_{g}(G)$:
\begin{align*}
&&\Big((\eta_{0})_{(x, g)} \circ T_{(x, g)}(\pi_{M \times G, M})\Big)(a, v)&=\Big(id_{T_{(x, g)}(M \times G)}-(\nu_{0})_{(x, g)}\Big)(a, v) 
\\
\, \Longleftrightarrow \, &&(\eta_{0})_{(x, g)}(a)&=(a, 0)
\end{align*}
From \eqref{22}, if $a=a^\mu\partial_{\mu} \in T_{x}(M)$ and in local coordinates $(x, g)$ is given by $(x^\mu, g^I)$, we then have:
\begin{align*}
a^\mu\partial_{\mu}=(\eta_{0})_{(x, g)}(a^\mu\partial_{\mu})&=a^\mu\Big(\partial_{\mu}-(\eta_{0})^I_\mu(x, g)\partial_{I}\Big) \, \Longleftrightarrow \, (\eta_{0})^I_\mu(x, g)=0
\end{align*}
As a check, observe that this request trivially satisfies the local conditions in order to have a Lie group fiber bundle connection, depicted in Theorem \ref{thm 6.4}.

\medskip

Finally, in the following theorem we see in which sense generalized principal connections reduce to usual principal connections on standard principal bundles:
\begin{thm}\label{thm 2.5.1}
Let $(P, M, \pi_{P, M}, G)$ be a standard principal bundle. Then all standard principal connections on $(P, M, \pi_{P, M}, G)$ are generalized principal connections associated to $\nu_0$ on $(P, M, \pi_{P, M}, G)$, considered as a generalized principal bundle, and viceversa.
\end{thm}
\begin{proof}

\medskip

\noindent Let a connection on $(P, M, \pi_{P, M}, G)$ be locally written as:
\begin{align*}
\tau=dx^\mu \otimes \Big(\partial_{\mu}-A^I_\mu(x, \sigma, g)\partial_{I}\Big)+d\sigma^\Theta \otimes \Big(\partial_{\Theta}-A^I_\Theta(x, \sigma, g)\partial_{I}\Big)
\end{align*}
From Section \ref{sec 2.3} we know that the connection coefficients $\{A^I_\mu(x, \sigma, g), A^I_\Theta(x, \sigma, g)\}$ define a generalized principal connection associated to $\nu_0$ on $(P, M, \pi_{P, M}, G)$, considered as a generalized principal bundle, if and only if they satisfy:
\begin{align*}
A^I_\mu\Big(x, \sigma, \pi(g, h)\Big)&=A^A_\mu(x, \sigma, g)\partial^1_A\pi^I(g, h)+(\eta_{0})^A_\mu(x, h)\partial^2_A\pi^I(g, h)
\\
A^I_\Theta\Big(x, \sigma, \pi(g, h)\Big)&=A^A_\Theta(x, \sigma, g)\partial^1_A\pi^I(g, h)
\end{align*}
where $(\eta_{0})^A_\mu(x, h)$ are the connection coefficients of the trivial Lie group fiber bundle connection on the trivial Lie group fiber bundle $(M \times G, M, \pi_{M \times G, M}, G)$. In particular, since there are only coordinates $x^\mu$ on $M$, we have just the connection coefficients $A^I_\mu(x, g)$, so that:
\begin{align}\label{38}
\tau=dx^\mu \otimes \Big(\partial_{\mu}-A^I_\mu(x, g)\partial_{I}\Big)
\end{align}
and, having $(\eta_{0})^A_\mu(x, h)=0$, the coefficients $A^I_\mu(x, g)$ define a generalized principal connection associated to $\nu_0$ on $(P, M, \pi_{P, M}, G)$ if and only if:
\begin{align}\label{39}
A^I_\mu\Big(x, \pi(g, h)\Big)=A^A_\mu(x, g)\partial^1_A\pi^I(g, h)
\end{align}
We now remark again that, from the transformation laws \eqref{100}, the generalized principal bundle coordinates $(x^\mu, g^I)$ are here principal bundle coordinates. Since, in these coordinates, conditions \eqref{38} and \eqref{39} define standard principal connections (see for example \cite[Chapter 18, Section 5]{FATIBENE (2024)}), the thesis follows.

\end{proof}

\section{Conclusion and perspectives}\label{sec 2.6}

At this point we are already able to see how impressive the geometric framework of generalized principal bundles and generalized principal connection is: not only we recover in it the classical notions of principal bundles and principal connection, but we have seen in Example \ref{ex 2.4.2} that it is possible to deal {\em at the same time} with vector bundles and their affine connections. Actually this framework is for sure even richer, remembering that any Lie group fiber bundle is a generalized principal bundle, thanks to Example \ref{ex 1.3.2}.

\medskip

As stated in the Introduction, the next step should be to understand how generalized principal connections might fit in the fiber bundle treatment of classical field theories, a step that is now easy to make thanks to the results presented here. In this direction, it is key to note that it is possible to construct a fiber bundle of which the sections are in a one to one correspondence with generalized principal connections on a fixed generalized principal bundle $(P, \mathcal{S}, \pi_{P, \mathcal{S}}, G)$, along the lines of the gauge-natural formalism and the usual principal connection bundle $\mathrm{Con}(P)$ (compare with \cite{ECK (1981)}, \cite{FATIBENE (2003)} and \cite{KOLAR (1993)}). 

The construction of such a {\em generalized principal connection bundle}, and its use in a variational framework, will be addressed in the forthcoming paper \cite{FATIBENE (2026)} by studying the transformation laws for the generalized principal connection coefficients and the notion of curvature for generalized principal connections. 

\medskip

In particular, by specializing the curvature tensor associated with an arbitrary connection $\Gamma^i_\mu(x, y)$ on any fiber bundle $(P, M, \pi_{P, M}, F)$:
\begin{align*}
R^i_{\mu\nu}(x, y)=\frac{\partial \Gamma_{\nu}^i}{\partial x^{\mu}}(x, y)-\frac{\partial \Gamma_{\mu}^i}{\partial x^{\nu}}(x, y)+\Gamma_{\nu}^k(x, y)\frac{\partial \Gamma_{\mu}^i}{\partial y^k}(x, y)-\Gamma_{\mu}^k(x, y)\frac{\partial \Gamma_{\nu}^i}{\partial y^k}(x, y)
\end{align*}
to the case of a generalized principal connection $\tau$ associated to a Lie group fiber bundle connection $\eta$, our constructive approach will disclose that the curvature of $\tau$ has a hybrid nature. 

As it can be expected, this curvature includes significant modifications to the usual picture. Specifically, it will contain terms closely resembling the standard principal connection curvature coefficients (i.e.\,\,the field strength in usual Yang-Mills theories) albeit crucially modified by $\eta$ as well as the curvature coefficients $(R_\eta)^J_{\mu\nu}(x, g)=\partial_{\mu}\eta^J_\nu(x, g)-\partial_{\nu}\eta^J_\mu(x, g)+\eta^A_\nu(x, g)\partial_A\eta^J_\mu(x, g)-\eta^A_\mu(x, g)\partial_A\eta^J_\nu(x, g)$ of $\eta$ itself.

\medskip

By extensively using the local coordinate characterization of Lie group fiber bundle connections provided by Theorem \ref{thm 6.4} and by taking partial derivatives of these local conditions, it will also be possible to explicitly compute the local form that the {\em Bianchi identities}, holding for an arbitrary connection, assume for generalized principal connections. 

These identities will be referred to as the {\em generalized homogeneous field equations} since they notably generalize both the homogeneous part of the Maxwell equations and the homogeneous field equations for a standard principal connection. Furthermore, they embody the aforementioned hybrid nature by incorporating terms that ultimately permit, as a special case, the reconstruction of the Bianchi identities for connections (with torsion) on a smooth manifold $M$, provided that the relevant generalized principal bundle is specified to be the tangent bundle of $M$.

\medskip

Since in the usual framework the corresponding considerations to the above are all necessary steps on the path towards Yang-Mills equations, we will also analyze in \cite{FATIBENE (2026)} the possibility of introducing a generalized Yang-Mills Lagrangian on the generalized principal connection bundle leading to {\em generalized Yang-Mills equations}, i.e.\,\,to a class of generalized gauge theories.

\medskip

In the meantime, it is noteworthy that the $\sigma^I_\mu(x)$ terms appearing in Example \ref{ex 2.4.2} turn out to be the local expressions of a {\em basic soldering form} $\sigma$ on the vector bundle $(E, M, \pi_{E, M}, \mathbb{R}^l)$, that is a section of the vector bundle $(T^{*}(M) \otimes_{M} E, M, \pi_{T^{*}(M) \otimes_{M} E, M}, \mathbb{R}^{m \cdot l})$, where $m=\mathrm{dim}(M)$:
\begin{align*}
\sigma \colon M \longrightarrow T^{*}(M) \otimes_{M} E \colon x^\mu \longmapsto \Big(x^\mu, \sigma^I_\mu(x)\Big)
\end{align*}

Moreover, in the vector bundle case, the generalized principal connection curvature essentially coincides with $(T_{\eta, \sigma})^A_{\mu\nu}(x)=\partial_{\mu}\sigma^A_\nu(x)-\partial_{\nu}\sigma^A_\mu(x)+\sigma^B_\nu(x)\eta^A_{B\mu}(x)-\sigma^B_\mu(x)\eta^A_{B\nu}(x)$, i.e.\,\,the {\em torsion} of the linear connection $\eta$ with respect to the basic soldering form $\sigma$ (see once more \cite{FATIBENE (2026)}).

\medskip

Following \cite{WINTERROTH (2026)}, when the generalized principal bundle $(P, \mathcal{S}, \pi_{P, \mathcal{S}}, G)$ is set to be a vector bundle $(E, M, \pi_{E, M}, \mathbb{R}^l)$ over a Lorentzian manifold $M$ with $\mathrm{dim}(M)=4$, it is possible to recover the {\em Vielbein} formulation of (vacuum) General Relativity. This approach provides a conceptual blueprint for how the generalized gauge theories will be introduced in \cite{FATIBENE (2026)} (in which the example we are going through will also be reviewed and extended).

As a matter of fact, if $\mathrm{det}\Big(\sigma^I_\mu(x)\Big)\neq0$, which implies that $l=4$, we get the {\em fake tangent bundle} environment used in the definition of the Vielbein formalism, the Palatini formalism and even the teleparallel equivalent of General Relativity (see for example \cite{BAEZ (2015)}). In this context, the basic soldering form $\sigma$ is called Vielbein or {\em coframe field}.

Since $M$ is a Lorentzian manifold and hence $(E, M, \pi_{E, M}, \mathbb{R}^4) \cong (T(M), M, \pi_{T(M), M}, \mathbb{R}^4)$ has an $SO(3, 1)$-reduction (i.e.\,\,from now on $G^I_J(x) \in SO(3, 1)$, see Example \ref{ex 2.2.2}), we can define, as it is well known, a metric $g$ with signature $(3, 1)$ on $M$ by setting locally $g_{\mu\nu}(x)=\sigma^A_\mu(x)\eta_{AB}\sigma^B_\nu(x)$, where $\eta_{AB}=\mathrm{diag}(-1, 1, 1, 1)$.

We can choose as our configuration bundle the fiber bundle $(T^{*}(M) \otimes_{M} E, M, \pi_{T^{*}(M) \otimes_{M} E, M}, \mathbb{R}^{16})$ of which the sections are in a one to one correspondence with basic soldering forms, as seen above. It has fibered coordinates $(x^\mu, \sigma^I_\mu)$ satisfying:
\begin{align*}
                \left\{
		\begin{array}{l}
		        x'^\mu=x'^\mu(x)
		        \\
			\sigma'^I_\mu=\bar{J}^\nu_\mu\Big(x'(x)\Big) \sigma^J_\nu G^I_J(x)
		\end{array}
		\right.
\end{align*}
where $G^I_J(x) \in SO(3, 1)$ (compare with Appendix \ref{app A} for the notation), which induce $(x^\mu, \sigma^I_\mu, \sigma^I_{\mu, \nu})$ as fibered coordinates on the jet bundle $J^1(T^{*}(M) \otimes_{M} E)$, remembering also that $\mathrm{det}(\sigma^I_\mu)\neq0$. We then consider the Lagrangian:
\begin{align*}
{L}=\Big(-\frac{1}{4}F^A_{\mu\nu}g^{\mu\alpha}g^{\nu\beta}F^B_{\alpha\beta}\eta_{AB}+R\Big)\sqrt{g}ds
\end{align*}
where:
\begin{enumerate}
\item $g_{\mu\nu}=g_{\mu\nu}(\sigma)=\sigma^A_\mu\eta_{AB}\sigma^B_\nu$, so that $g^{\mu\nu}=g^{\mu\nu}(\sigma)=\sigma_A^\mu\eta^{AB}\sigma_B^\nu$, and $g=\Big|\det(g_{\mu\nu})\Big|$, so that $\sqrt{g}=|\sigma|=\Big|\mathrm{det}(\sigma^I_\mu)\Big|$.
\item $F^A_{\mu\nu}=F^A_{\mu\nu}(\sigma)=\sigma^B_{\nu, \mu}-\sigma^B_{\mu, \nu}+\sigma^C_\nu\eta^B_{C\mu}(\sigma)-\sigma^C_\mu\eta^B_{C\nu}(\sigma)$ are defined by taking inspiration from the generalized principal connection curvature in the vector bundle case, where $\eta^B_{C\mu}(\sigma)=\sigma^\alpha_C\eta^\beta_{\alpha\mu}(\sigma)\sigma^B_\beta-\sigma^\alpha_C\sigma^B_{\alpha, \mu}$ are the well-known {\em spin connection} coefficients. $\eta^\beta_{\alpha\mu}(\sigma)$ are just the Levi-Civita connection coefficients written in terms of the Vielbein $\sigma$, through the metric $g$.
\item $R=R(\sigma)$ is the Ricci scalar written in terms of the Vielbein $\sigma$, through the metric $g$.
\item $ds=dx^1 \wedge dx^2 \wedge dx^3 \wedge dx^4$.
\end{enumerate}
We can readily observe that:
\begin{align*}
\sigma^B_{\nu, \mu}-\sigma^B_{\mu, \nu}+\sigma^C_\nu\eta^B_{C\mu}(\sigma)-\sigma^C_\mu\eta^B_{C\nu}(\sigma)=\eta^\beta_{\nu\mu}(\sigma)\sigma^B_\beta-\eta^\beta_{\mu\nu}(\sigma)\sigma^B_\beta=\Big(\eta^\beta_{\nu\mu}(\sigma)-\eta^\beta_{\mu\nu}(\sigma)\Big)\sigma^B_\beta=0
\end{align*}
since $\eta^\beta_{\alpha\mu}(\sigma)$ are the Levi-Civita connection coefficients, hence $F^A_{\mu\nu}=0$ and ${L}=R|\sigma|ds$. In other words, we have actually written the Hilbert-Einstein Lagrangian in the Vielbein formalism.

\medskip

This construction can be interpreted as a way to embed, in a preliminary manner, General Relativity into the present new geometric structure, highlighting also the well-known asymmetry between Yang-Mills theories and General Relativity given by the Vielbein $\sigma$ (see for example \cite{ZANELLI (2012)}).

\medskip

In this sense, the local coordinate approach developed here allows us to see the similarities and differences between these two kinds of field theories in an explicit, quantitative and computable way inside the common framework of generalized principal bundles, as it will become even more evident in \cite{FATIBENE (2026)}. Finally, we observe that the General Relativity example provides a prototype for generalized Yang-Mills theories insofar as we will precisely use local representations to explicitly write unification-minded Lagrangians.

\appendix
\section{Invariance of the local coordinate characterization of Lie group fiber bundle connections}\label{app A}

We check here the invariance, under change of coordinates \eqref{1}, of the local conditions appearing in Theorem \ref{thm 6.4}. Comparing with \cite[Chapter 3, Section 5]{FATIBENE (2003)}, we know that if the transformation laws for the fibered coordinates are \eqref{1}, then it holds that the connection coefficients transform as:
\begin{align}\label{23}
{\eta'}^I_\mu(x', g')={\eta'}^I_\mu\Big(x'(x), G(x, g)\Big)={\bar{J}}^{\nu}_{\mu}\Big(x'(x)\Big)\Big({J}^I_J(x, g)\eta^J_\nu(x, g)-{J}^I_\nu(x, g)\Big) 
\end{align}
where ${J}^{\nu}_{\mu}(x)=\frac{\partial x'^\nu}{\partial x^\mu}(x)$, ${J}^I_\nu(x, g)=\frac{\partial G^I}{\partial x^\nu}(x, g)$, ${J}^I_J(x, g)=\frac{\partial G^I}{\partial g^J}(x, g)$ and ${\bar{J}}^{\nu}_{\mu}(x')=\frac{\partial x^\nu}{\partial x'^\mu}(x')$ is the inverse of ${J}^{\nu}_{\mu}$. Hence:
\begin{align*}
{\eta'}^I_\mu(x', e')={\bar{J}}^{\nu}_{\mu}\Big(x'(x)\Big)\Big({J}^I_J(x, e)\eta^J_\nu(x, e)-{J}^I_\nu(x, e)\Big)=0
\end{align*}
because we are assuming that $\eta^J_\nu(x, e)=0$ and since thanks to \eqref{2} it holds that:
\begin{align*}
{J}^I_\nu(x, e)=\frac{\partial G^I}{\partial x^\nu}(x, e)=\frac{\partial e^I}{\partial x^\nu}=0
\end{align*}
Moreover since:
\begin{align*}
\frac{\partial}{\partial x^\mu}G^I\Big(x, \pi(g, h)\Big)&={J}^I_\mu\Big(x, \pi(g, h)\Big)
\\
\frac{\partial}{\partial g^J}G^I\Big(x, \pi(g, h)\Big)&={J}^I_A\Big(x, \pi(g, h)\Big)\partial^1_J\pi^A(g, h)
\\
\frac{\partial}{\partial h^J}G^I\Big(x, \pi(g, h)\Big)&={J}^I_A\Big(x, \pi(g, h)\Big)\partial^2_J\pi^A(g, h)
\\
\frac{\partial}{\partial x^\mu}\pi^I\Big(G(x, g), G(x, h)\Big)&=\partial^1_A\pi^I\Big(G(x, g), G(x, h)\Big){J}^A_\mu(x, g)+\partial^2_A\pi^I\Big(G(x, g), G(x, h)\Big){J}^A_\mu(x, h)
\\
\frac{\partial}{\partial g^J}\pi^I\Big(G(x, g), G(x, h)\Big)&=\partial^1_A\pi^I\Big(G(x, g), G(x, h)\Big){J}^A_J(x, g)
\\
\frac{\partial}{\partial h^J}\pi^I\Big(G(x, g), G(x, h)\Big)&=\partial^2_A\pi^I\Big(G(x, g), G(x, h)\Big){J}^A_J(x, h)
\end{align*}
then using \eqref{3} we have that:
\begin{align*}
{J}^I_\mu\Big(x, \pi(g, h)\Big)&=\partial^1_A\pi^I\Big(G(x, g), G(x, h)\Big){J}^A_\mu(x, g)+\partial^2_A\pi^I\Big(G(x, g), G(x, h)\Big){J}^A_\mu(x, h)
\\
{J}^I_A\Big(x, \pi(g, h)\Big)\partial^1_J\pi^A(g, h)&=\partial^1_A\pi^I\Big(G(x, g), G(x, h)\Big){J}^A_J(x, g)
\\
{J}^I_A\Big(x, \pi(g, h)\Big)\partial^2_J\pi^A(g, h)&=\partial^2_A\pi^I\Big(G(x, g), G(x, h)\Big){J}^A_J(x, h)
\end{align*}
Assuming:
\begin{align*}
\eta^I_\mu\Big(x, \pi(g, h)\Big)=\eta^A_\mu(x, g)\partial^1_A\pi^I(g, h)+\eta^A_\mu(x, h)\partial^2_A\pi^I(g, h)
\end{align*}
this implies that:
\begin{align*}
{\eta'}^I_\mu\Big(x', \pi(g, h)'\Big)&={\bar{J}}^{\nu}_{\mu}\Big(x'(x)\Big)\Big({J}^I_J\Big(x, \pi(g, h)\Big)\eta^J_\nu\Big(x, \pi(g, h)\Big)-{J}^I_\nu\Big(x, \pi(g, h)\Big)\Big) 
\\
&={\bar{J}}^{\nu}_{\mu}\Big(x'(x)\Big)\Big({J}^I_J\Big(x, \pi(g, h)\Big)\eta^A_\nu(x, g)\partial^1_A\pi^J(g, h)
\\
&\,\,\,\,\,\,\,+{J}^I_J\Big(x, \pi(g, h)\Big)\eta^A_\nu(x, h)\partial^2_A\pi^J(g, h)-{J}^I_\nu\Big(x, \pi(g, h)\Big)\Big)
\\
&={\bar{J}}^{\nu}_{\mu}\Big(x'(x)\Big)\Big(\partial^1_J\pi^I\Big(G(x, g), G(x, h)\Big){J}^J_A(x, g)\eta^A_\nu(x, g)
\\
&\,\,\,\,\,\,\,+\partial^2_J\pi^I\Big(G(x, g), G(x, h)\Big){J}^J_A(x, h)\eta^A_\nu(x, h)-\partial^1_J\pi^I\Big(G(x, g), G(x, h)\Big){J}^J_\nu(x, g)
\\
&\,\,\,\,\,\,\,-\partial^2_J\pi^I\Big(G(x, g), G(x, h)\Big){J}^J_\nu(x, h)\Big)
\\
&=\partial^1_J\pi^I\Big(G(x, g), G(x, h)\Big){\bar{J}}^{\nu}_{\mu}\Big(x'(x)\Big)\Big({J}^J_A(x, g)\eta^A_\nu(x, g)-{J}^J_\nu(x, g)\Big)
\\
&\,\,\,\,\,\,\,+\partial^2_J\pi^I\Big(G(x, g), G(x, h)\Big){\bar{J}}^{\nu}_{\mu}\Big(x'(x)\Big)\Big({J}^J_A(x, h)\eta^A_\nu(x, h)-{J}^J_\nu(x, h)\Big)
\\
&=\partial^1_J\pi^I\Big(G(x, g), G(x, h)\Big){\eta'}^J_\mu(x', g')+\partial^2_J\pi^I\Big(G(x, g), G(x, h)\Big){\eta'}^J_\mu(x', h')
\\
&={\eta'}^J_\mu(x', g')\partial^1_J\pi^I(g', h')+{\eta'}^J_\mu(x', h')\partial^2_J\pi^I(g', h')
\end{align*}
where we used several times \eqref{23}.

\section{Heuristic comparison between generalized principal connections and standard principal connections}\label{app B}

\medskip

In \cite{CASTRILLON LOPEZ (2023)} the following technical results are presented:
\begin{lem}\label{lem 1.18}
Let $(\mathfrak{a}, M, \pi_{\mathfrak{a}, M}, \mathfrak{g})$ be the Lie algebra fiber bundle attached to a Lie group fiber bundle $(\mathcal G, M, \pi_{\mathcal G, M}, G)$ acting on a fiber bundle $(P, M, \pi_{P, M}, F)$ and generating the generalized principal bundle $(P, \mathcal{S}, \pi_{P, \mathcal{S}}, G)$, where $\mathcal{S}=P/\mathcal{G}$. Then:
\begin{enumerate}
\item The following map is a vertical vector bundle isomorphism over $P$:
\begin{align*}
f \colon P \times_M \mathfrak{a} \longrightarrow V_{\pi_{P, \mathcal{S}}}(P) \colon (p, \xi) \longmapsto \xi^*_p
\end{align*}
where $\xi^*$ are the generalized fundamental fields, as defined in Section \ref{sec 1.3}.
\item Given $(\gamma, \xi) \in \mathcal{G} \times_M \mathfrak{a}$, we have that $T(\Phi_{\gamma})(\xi^*)=\Big(Ad_{\gamma^{-1}}(\xi)\Big)^*$, in the sense that $\forall \, p \in P_x$:
\begin{align}\label{13}
T(\Phi_{\gamma})(\xi^*)(p)=T_p(\Phi_{\gamma})(\xi^*_p)=\Big(Ad_{\gamma^{-1}}(\xi)\Big)^*_{\Phi(p, \gamma)}
\end{align}
where $c_{\gamma}(\delta)=\gamma \cdot \delta \cdot \gamma^{-1}$, $\forall \, (\gamma, \delta) \in \mathcal{G} \times_M \mathcal{G}$, and $Ad_{\gamma}(\xi)=T_{e_x}(c_{\gamma})(\xi) \in \mathfrak{a}_x$, with $x=\pi_{\mathcal{G}, M}(\gamma)=\pi_{\mathfrak{a}, M}(\xi)$.
\end{enumerate}
\end{lem}
\begin{oss}\label{oss B.2}
We can note here that the first point of the statement is a direct consequence of the fact that:
\begin{align*}
f_p \colon (P \times_M \mathfrak{a})_p=\mathfrak{a}_{\pi_{P, M}(p)}=\mathfrak{a}_{x} \longrightarrow V_{{\pi_{P, \mathcal{S}}}, p}(P) \colon \xi \longmapsto \xi^*_p
\end{align*}
is a linear isomorphism, which can be proved noting that $f_p=T_{e_x}(\Phi_{p})$ and that we have already observed in Section \ref{sec 1.3} that $\Phi_{p}$ is a diffeomorphism. This is an alternative proof to the one appeared in \cite{CASTRILLON LOPEZ (2023)}.
\end{oss}
\begin{lem}\label{lem 2.4}
There is a one to one correspondence between (fiber bundle) connections on a generalized principal bundle $(P, \mathcal{S}, \pi_{P, \mathcal{S}}, G)$, seen as vector bundle morphisms $\omega \colon T(P) \longrightarrow V_{\pi_{P, \mathcal{S}}}(P)$, and $1$-forms $\hat{\omega} \in \Omega^1(P, P \times_M \mathfrak{a})$ with values in $P \times_M \mathfrak{a}$, where $(\mathfrak{a}, M, \pi_{\mathfrak{a}, M}, \mathfrak{g})$ is the Lie algebra fiber bundle attached to the relevant Lie group fiber bundle $(\mathcal G, M, \pi_{\mathcal G, M}, G)$, having the property that for any $(p, \xi)\in P \times_M \mathfrak{a}$:
\begin{align}\label{24}
\hat{\omega}_p(\xi^*_p)=\xi
\end{align}
with $\hat{\omega}_p \colon T_p(P) \longrightarrow (P \times_M \mathfrak{a})_p=\mathfrak{a}_{x}$ and $\xi^*$ is the generalized fundamental field on $P$ associated with $\xi$, once set $\pi_{P, M}(p)=\pi_{\mathfrak{a}, M}(\xi)=x \in M$. Moreover, $\hat{\omega}$ satisfies:
\begin{align}\label{25}
\hat{\omega}_{\Phi(p, \gamma)}\Big(T_p(\Phi_{\gamma})(\xi^*_p)\Big)=Ad_{\gamma^{-1}}\Big(\hat{\omega}_p(\xi_p^*)\Big)    
\end{align}
for all $(p, \gamma,\xi) \in P \times_M \mathcal{G} \times_M \mathfrak{a}$.
\end{lem}

The one to one correspondence in the previous lemma can be proved using that, by Lemma \ref{lem 1.18}, $f(p, \xi)=\xi^*_p$ is a vertical vector bundle isomorphism over $P$. The correspondence between vector bundle morphisms $\omega$ and $1$-forms $\hat{\omega}$ is given by the formula $\omega=f \circ \hat{\omega}$ and is characterized by property \eqref{24} if we are precisely taking into account connections on $(P, \mathcal{S}, \pi_{P, \mathcal{S}}, G)$. Relation \eqref{25} is instead implied by \eqref{13}. In particular we have that for all $p \in P$, for all $z \in T_p(P)$:
\begin{align*}
\omega_p(z)=\Big(f_p \circ \hat{\omega}_p\Big)(z)=\Big(\hat{\omega}_p(z)\Big)^*_p
\end{align*}
since $f_p(\xi)=\xi^*_p$, for all $\xi \in \mathfrak{a}_{\pi_{P, M}(p)}=\mathfrak{a}_{x}$.

It is worth noting, incidentally, that Definition \ref{defn 2.6} is well-posed since, thanks to \eqref{13}, it holds:
\begin{align}\label{27}
T_{p}(\Phi _{\gamma})\Big(\omega_p(z)\Big)=T_{p}(\Phi _{\gamma})\Big(\Big(\hat{\omega}_p(z)\Big)^*_p\Big)=\Big(Ad_{\gamma^{-1}}\Big(\hat{\omega}_p(z)\Big)\Big)^*_{\Phi(p, \gamma)} \in V_{\pi_{P, \mathcal{S}}, \Phi(p, \gamma)}(P)
\end{align}

\medskip

With the above in mind we can now consider that in order to assign a principal connection on a standard principal bundle $(P, M, \pi_{P, M}, G)$ one needs to define a {\em connection form}, that is a $1$-form $\omega \in \Omega^1(P, \mathfrak{g})$ with values in $\mathfrak{g}$, the Lie algebra of $G$, which means that one needs to define a vertical vector bundle morphism over $P$:
\begin{equation*}
\xymatrixrowsep{0.56in}
\xymatrixcolsep{0.8in}
	\xymatrix{
	   {T(P)} \ar[r]^{\omega} \ar[d]_{\pi_{T(P), P}} & {P \times \mathfrak{g}} \ar[d]^{\pi_{P \times \mathfrak{g}, P}} \\
	   {P}  \ar[r]_{{id}_P} & {P}  \\}
\end{equation*}
such that:
\begin{enumerate}
\item $\forall \, \xi \in \mathfrak{g}$, $\omega_p(\xi^*_p)=\xi$, where $\xi^*$ are the {\em standard} fundamental fields, as defined in Section \ref{sec 1.3}.
\item $\forall \, p \in P$, $\forall \, g \in G$, $\forall \, z \in T_p(P)$:
\begin{align}\label{28}
\omega_{m(p, g)}\Big(T_p(\mathcal{R}_g)(z)\Big)=Ad_{g^{-1}}\Big(\omega_p(z)\Big)
\end{align}
where $m$ is the smooth right action of $G$ on $P$, $Ad_g=T_{e}(c_g)$, $c_g=L_g\circ R_{g^{-1}}$, $L_g$ and $R_g$ are respectively the left and right multiplication in $G$ and $\mathcal{R}_g(p)=m(p, g)$.
\end{enumerate}
\needspace{3\baselineskip}
As a result, Definition \ref{defn 2.6} extends this classical notion of $Ad$-equivariance. In fact, thanks to Lemma \ref{lem 2.4} and \eqref{27}, the requirements in the new definition can be rephrased for a fixed connection $\nu$ on the Lie group fiber bundle $(\mathcal{G}, M, \pi_{\mathcal{G}, M}, G)$ in terms of $1$-forms $\hat{\omega} \in \Omega^1(P, P \times_M \mathfrak{a})$ as:
\begin{enumerate}
\item $\forall \, (p, \xi)\in P \times_M \mathfrak{a}$, $\hat{\omega}_p(\xi^*_p)=\xi$, where $\xi^*$ are the generalized fundamental fields, as defined in Section \ref{sec 1.3}.
\item $\forall \, (p, \gamma) \in P \times_M \mathcal{G}$ and  $\forall \, (z, v) \in T_{(p, \gamma)}(P \times_M \mathcal{G}) \equiv T_{p}(P) \times_{T_{x}(M)} T_{\gamma}(\mathcal{G})$, that is $T_{p}(\pi_{P, M})(z)=T_{\gamma}(\pi_{\mathcal{G}, M})(v)$, with $\pi_{P, M}(p)=\pi_{\mathcal{G}, M}(\gamma)=x$:
\begin{align*}
\hat{\omega}_{\Phi(p, \gamma)}\Big(T_{(p, \gamma)}(\Phi)(z, v)\Big)&=f^{-1}_{\Phi(p, \gamma)}\Big(\omega_{\Phi(p, \gamma)}\Big(T_{(p, \gamma)}(\Phi)(z, v)\Big)\Big)
\\
&=f^{-1}_{\Phi(p, \gamma)}\Big(T_{p}(\Phi _{\gamma})\Big(\omega_p(z)\Big)+\Big(T_{\gamma}(L_{\gamma^{-1}})\Big(\nu_{\gamma}(v)\Big)\Big)^*_{\Phi(p, \gamma)}\Big)
\\
&=f^{-1}_{\Phi(p, \gamma)}\Big(\Big(Ad_{\gamma^{-1}}\Big(\hat{\omega}_p(z)\Big)\Big)^*_{\Phi(p, \gamma)}+\Big(Ad_{\gamma^{-1}}\Big(\hat{\nu}_{\gamma}(v)\Big)\Big)^*_{\Phi(p, \gamma)}\Big)
\\
&=f^{-1}_{\Phi(p, \gamma)}\Big(\Big(Ad_{\gamma^{-1}}\Big(\hat{\omega}_p(z)+\hat{\nu}_{\gamma}(v)\Big)\Big)^*_{\Phi(p, \gamma)}\Big)
\\
&=Ad_{\gamma^{-1}}\Big(\hat{\omega}_p(z)+\hat{\nu}_{\gamma}(v)\Big)
\end{align*}
where, observing that $V_{\gamma}(\mathcal{G}) \equiv T_{\gamma}(\mathcal{G}_{\pi_{\mathcal{G}, M}(\gamma)})=T_{\gamma}(\mathcal{G}_{x})$, we have set:
\begin{align*}
\hat{\nu}_{\gamma}=T_{\gamma}(R_{\gamma^{-1}}) \circ \nu_{\gamma} \colon T_{\gamma}(\mathcal{G}) \longrightarrow T_{\gamma}(\mathcal{G}_{x}) \longrightarrow T_{e_x}(\mathcal{G}_{x})=\mathfrak{a}_x \colon v \longmapsto \nu_{\gamma}(v) \longmapsto T_{\gamma}(R_{\gamma^{-1}})\Big(\nu_{\gamma}(v)\Big)
\end{align*}
\end{enumerate}
Note now that considering a smooth curve $\Big(\alpha_1(t), \alpha_2(t)\Big) \colon I \subseteq \mathbb{R} \rightarrow P \times_M \mathcal{G}$ with $\Big(\alpha_1(0), \alpha_2(0)\Big)=(p, \gamma)$ and $\Big(\alpha'_1(0), \alpha'_2(0)\Big)=(z, 0)$, where $T_{p}(\pi_{P, M})(z)=0$, so that we can directly fix $\alpha_2(t)=\gamma$, for each $t \in I$, we then have:
\begin{align*}
T_{(p, \gamma)}(\Phi)(z, 0)&=\left.\frac{d}{dt}\right|_{t=0}\Big(\Phi \circ \Big(\alpha_1(t), \alpha_2(t)\Big)\Big)=\left.\frac{d}{dt}\right|_{t=0}\Phi_{\gamma}\Big(\alpha_1(t)\Big)=T_{p}(\Phi_{\gamma})(z)
\end{align*}
which implies:
\begin{align*}
\hat{\omega}_{\Phi(p, \gamma)}\Big(T_{p}(\Phi_{\gamma})(z)\Big)=\hat{\omega}_{\Phi(p, \gamma)}\Big(T_{(p, \gamma)}(\Phi)(z, 0)\Big)=Ad_{\gamma^{-1}}\Big(\hat{\omega}_p(z)+\hat{\nu}_{\gamma}(0)\Big)=Ad_{\gamma^{-1}}\Big(\hat{\omega}_p(z)\Big)
\end{align*}
This means that we are asking the connection to ensure that \eqref{25} holds for all $z \in \mathrm{Ker}\Big(T_{p}(\pi_{P, M})\Big)=V_{\pi_{P, M}, \, p}(P)$. We have then found eventually an algebraic expression similar to \eqref{28}. This heuristic reasoning becomes rigorous through Theorem \ref{thm 2.5.1} in Section \ref{sec 2.5}.

Summing up, a $1$-form $\hat{\omega} \in \Omega^1(P, P \times_M \mathfrak{a})$, that can be seen as a vector bundle morphism over the identity ${id}_P \colon P \rightarrow P$:
\begin{equation*}
\xymatrixrowsep{0.56in}
\xymatrixcolsep{0.8in}
	\xymatrix{
	   {T(P)} \ar[r]^{\hat{\omega}} \ar[d]_{\pi_{T(P), P}} & {P \times_M \mathfrak{a}} \ar[d]^{\pi_{P \times_M \mathfrak{a}, P}} \\
	   {P}  \ar[r]_{{id}_P} & {P}  \\}
\end{equation*}
satisfying the above conditions shall be thereby called {\em generalized connection form}.

\section{Invariance of the local coordinate characterization of generalized principal connections}\label{app C}

We check here the invariance, under change of coordinates \eqref{1} and \eqref{11}, of the local conditions \eqref{35} and \eqref{36}. We already know that under the transformation laws \eqref{1} the Lie group fiber bundle connection coefficients behave as:
\begin{align*}
{\eta'}^I_\mu(x', g')={\eta'}^I_\mu\Big(x'(x), G(x, g)\Big)={\bar{J}}^{\nu}_{\mu}\Big(x'(x)\Big)\Big({J}^I_J(x, g)\eta^J_\nu(x, g)-{J}^I_\nu(x, g)\Big) 
\end{align*}
Similarly, comparing with \cite[Chapter 3, Section 5]{FATIBENE (2003)} and knowing that the transformation laws for the generalized principal bundle coordinates are \eqref{11}, then:
\begin{align*}
A'^I_\mu(x', \sigma', g')&=A'^I_\mu\Big(x'(x), \Sigma(x, \sigma), \pi\Big(\varphi(x, \sigma), G(x, g)\Big)\Big)
\\
&={\bar{J}}^{\nu}_{\mu}\Big(x'(x)\Big)\Big(\frac{\partial}{\partial g^J}\Big[\pi^I\Big(\varphi(x, \sigma), G(x, g)\Big)\Big]A^J_\nu(x, \sigma, g)-\frac{\partial}{\partial x^\nu}\Big[\pi^I\Big(\varphi(x, \sigma), G(x, g)\Big)\Big]\Big)
\\
&\,\,\,\,\,\,-{\bar{J}}^{\nu}_{\mu}\Big(x'(x)\Big){J}^\Xi_\nu(x, \sigma){\bar{J}}^{\Lambda}_{\Xi}\Big(x'(x), \Sigma^\Theta(x, \sigma)\Big)\Big(\frac{\partial}{\partial g^J}\Big[\pi^I\Big(\varphi(x, \sigma), G(x, g)\Big)\Big]A^J_\Lambda(x, \sigma, g)
\\
&\,\,\,\,\,\,-\frac{\partial}{\partial \sigma^\Lambda}\Big[\pi^I\Big(\varphi(x, \sigma), G(x, g)\Big)\Big]\Big)
\\
A'^I_\Theta(x', \sigma', g')&=A'^I_\Theta\Big(x'(x), \Sigma(x, \sigma), \pi\Big(\varphi(x, \sigma), G(x, g)\Big)\Big)
\\
&={\bar{J}}^{\Lambda}_{\Theta}\Big(x'(x),\Sigma^\Theta(x, \sigma)\Big)\Big(\frac{\partial}{\partial g^J}\Big[\pi^I\Big(\varphi(x, \sigma), G(x, g)\Big)\Big]A^J_\Lambda(x, \sigma, g)
\\
&\,\,\,\,\,\,-\frac{\partial}{\partial \sigma^\Lambda}\Big[\pi^I\Big(\varphi(x, \sigma), G(x, g)\Big)\Big]\Big)
\end{align*}
where ${J}^\Lambda_\mu(x, \sigma)=\frac{\partial \Sigma^\Lambda}{\partial x^\mu}(x, \sigma)$, ${J}^\Lambda_\Xi(x, \sigma)=\frac{\partial \Sigma^\Lambda}{\partial \sigma^\Xi}(x, \sigma)$ and ${\bar{J}}^{\Lambda}_{\Xi}(x', \sigma')$ is the inverse of ${J}^\Lambda_\Xi$, having used that:
\begin{align*}
\left(
\begin{matrix}
J^\nu_\mu & 0 \\
{J}^\Lambda_\mu & {J}^\Lambda_\Xi
\end{matrix}
\right)^{-1}
=
\left(
\begin{matrix}
\bar{J}^\nu_\mu & 0 \\
-\bar{J}^\Lambda_\Xi{J}^\Xi_\nu\bar{J}^\nu_\mu & \bar{J}^\Lambda_\Xi
\end{matrix}
\right)
\end{align*}
Note that:
\begin{align*}
\left.\frac{\partial}{\partial g^J}\Big[\pi^I\Big(\varphi(x, \sigma), G(x, g)\Big)\Big]\right|_{g=e}&=\partial^2_A\pi^I\Big(\varphi(x, \sigma), e\Big){J}^A_J(x, e)
\\
\left.\frac{\partial}{\partial x^\nu}\Big[\pi^I\Big(\varphi(x, \sigma), G(x, g)\Big)\Big]\right|_{g=e}&=\partial^1_A\pi^I\Big(\varphi(x, \sigma), e\Big)\partial_{\nu}\varphi^A(x, \sigma)
\\
\left.\frac{\partial}{\partial \sigma^\Lambda}\Big[\pi^I\Big(\varphi(x, \sigma), G(x, g)\Big)\Big]\right|_{g=e}&=\partial^1_A\pi^I\Big(\varphi(x, \sigma), e\Big)\partial_{\Lambda}\varphi^A(x, \sigma)
\end{align*}
since ${J}^A_\nu(x, e)=0$, by Appendix \ref{app A}. In this way we get also that:
\begin{align*}
A'^I_\mu(x', \sigma', e')&=A'^I_\mu\Big(x'(x), \Sigma(x, \sigma), \pi\Big(\varphi(x, \sigma), e\Big)\Big)=A'^I_\mu\Big(x'(x), \Sigma(x, \sigma), \varphi(x, \sigma)\Big)
\\
&={\bar{J}}^{\nu}_{\mu}\Big(x'(x)\Big)\Big(\partial^2_A\pi^I\Big(\varphi(x, \sigma), e\Big){J}^A_J(x, e)A^J_\nu(x, \sigma, e)-\partial^1_A\pi^I\Big(\varphi(x, \sigma), e\Big)\partial_{\nu}\varphi^A(x, \sigma)\Big)
\\
&\,\,\,\,\,\,-{\bar{J}}^{\nu}_{\mu}\Big(x'(x)\Big){J}^\Xi_\nu(x, \sigma){\bar{J}}^{\Lambda}_{\Xi}\Big(x'(x), \Sigma^\Theta(x, \sigma)\Big)\Big(\partial^2_A\pi^I\Big(\varphi(x, \sigma), e\Big){J}^A_J(x, e)A^J_\Lambda(x, \sigma, e)
\\
&\,\,\,\,\,\,-\partial^1_A\pi^I\Big(\varphi(x, \sigma), e\Big)\partial_{\Lambda}\varphi^A(x, \sigma)\Big)
\\
A'^I_\Theta(x', \sigma', e')&=A'^I_\Theta\Big(x'(x), \Sigma(x, \sigma), \pi\Big(\varphi(x, \sigma), e\Big)\Big)=A'^I_\Theta\Big(x'(x), \Sigma(x, \sigma), \varphi(x, \sigma)\Big)
\\
&={\bar{J}}^{\Lambda}_{\Theta}\Big(x'(x),\Sigma^\Theta(x, \sigma)\Big)\Big(\partial^2_A\pi^I\Big(\varphi(x, \sigma), e\Big){J}^A_J(x, e)A^J_\Lambda(x, \sigma, e)
\\
&\,\,\,\,\,\,-\partial^1_A\pi^I\Big(\varphi(x, \sigma), e\Big)\partial_{\Lambda}\varphi^A(x, \sigma)\Big)
\end{align*}
Having the above results in mind and assuming:
\begin{align*}
A^I_\mu(x, \sigma, g)&=A^A_\mu(x, \sigma, e)\partial^1_A\pi^I(e, g)+\eta^I_\mu(x, g)
\\
A^I_\Theta(x, \sigma, g)&=A^A_\Theta(x, \sigma, e)\partial^1_A\pi^I(e, g)
\end{align*}
we want to prove that:
\begin{align*}
A'^I_\mu(x', \sigma', g')&=A'^A_\mu(x', \sigma', e')\partial^1_A\pi^I\Big(\varphi(x, \sigma), G(x, g)\Big)+\eta'^A_\mu\Big(x', G(x, g)\Big)\partial^2_A\pi^I\Big(\varphi(x, \sigma), G(x, g)\Big)
\\
A'^I_\Theta(x', \sigma', g')&=A'^A_\Theta(x', \sigma', e')\partial^1_A\pi^I\Big(\varphi(x, \sigma), G(x, g)\Big)
\end{align*}
Indeed, the conditions:
\begin{align*}
A'^I_\mu\Big(x', \sigma', \pi(g', h')\Big)&=A'^A_\mu(x', \sigma', g')\partial^1_A\pi^I(g', h')+\eta'^A_\mu(x', h')\partial^2_A\pi^I(g', h')
\\
A'^I_\Theta\Big(x', \sigma', \pi(g', h')\Big)&=A'^A_\Theta(x', \sigma', g')\partial^1_A\pi^I(g', h')
\end{align*}
would give for $g'=e'=\varphi(x, \sigma)$ exactly the above expressions, thanks to the transformation laws \eqref{34}. With this aim, it is possible to expand:
\begin{align*}
\frac{\partial}{\partial g^J}\Big[\pi^I\Big(\varphi(x, \sigma), G(x, g)\Big)\Big]A^J_\nu(x, \sigma, g)-\frac{\partial}{\partial x^\nu}\Big[\pi^I\Big(\varphi(x, \sigma), G(x, g)\Big)\Big]
\end{align*}
after lengthy computations as:
\begin{align*}
&\partial^1_B\pi^I\Big(\varphi(x, \sigma), G(x, g)\Big)\Big(\partial^2_A\pi^B\Big(\varphi(x, \sigma), e\Big){J}^A_J(x, e)A^J_\nu(x, \sigma, e)-\partial^1_A\pi^B\Big(\varphi(x, \sigma), e\Big)\partial_{\nu}\varphi^A(x, \sigma)\Big)
\\
&+\partial^2_A\pi^I\Big(\varphi(x, \sigma), G(x, g)\Big)\Big({J}^A_J(x, g)\eta^J_\nu(x, g)-{J}^A_\nu(x, g)\Big)
\end{align*}
This can be done using that by Appendix \ref{app A} it holds:
\begin{align*}
{J}^I_A\Big(x, \pi(g, h)\Big)\partial^1_J\pi^A(g, h)&=\partial^1_A\pi^I\Big(G(x, g), G(x, h)\Big){J}^A_J(x, g)
\end{align*}
and also that it has been already proved that $\partial^1_A\pi^I\Big(k, \pi(g, h)\Big)=\partial^1_B\pi^I\Big(\pi(k, g), h\Big)\partial^1_A\pi^B(k, g)$ which implies for $k=\varphi(x, \sigma)$, $g=e$ and $h=G(x, g)$ that:
\begin{align*}
\partial^1_A\pi^I\Big(\varphi(x, \sigma), G(x, g)\Big)=\partial^1_B\pi^I\Big(\varphi(x, \sigma), G(x, g)\Big)\partial^1_A\pi^B\Big(\varphi(x, \sigma), e\Big)
\end{align*}
Moreover it is necessary to use the associativity of the product in $G$ in a different way:
\begin{align*}
\pi^I\Big(k, \pi(g, h)\Big)=\pi^I\Big(\pi(k, g), h\Big) \, &\Longrightarrow\, \frac{\partial}{\partial g^J}\pi^I\Big(k, \pi(g, h)\Big)=\frac{\partial}{\partial g^J}\pi^I\Big(\pi(k, g), h\Big)
\\
\, &\Longrightarrow \, \partial^2_A\pi^I\Big(k, \pi(g, h)\Big)\partial^1_J\pi^A(g, h)=\partial^1_A\pi^I\Big(\pi(k, g), h\Big)\partial^2_J\pi^A(k, g)
\end{align*}
so that we get as well the following relation for $k=\varphi(x, \sigma)$, $g=e$ and $h=G(x, g)$:
\begin{align*}
\partial^2_A\pi^I\Big(\varphi(x, \sigma), G(x, g)\Big)\partial^1_J\pi^A\Big(e, G(x, g)\Big)=\partial^1_A\pi^I\Big(\varphi(x, \sigma), G(x, g)\Big)\partial^2_J\pi^A\Big(\varphi(x, \sigma), e\Big)
\end{align*}
In the same fashion, we also expand:
\begin{align*}
\frac{\partial}{\partial g^J}\Big[\pi^I\Big(\varphi(x, \sigma), G(x, g)\Big)\Big]A^J_\Lambda(x, \sigma, g)-\frac{\partial}{\partial \sigma^\Lambda}\Big[\pi^I\Big(\varphi(x, \sigma), G(x, g)\Big)\Big]
\end{align*}
as:
\begin{align*}
&\partial^1_B\pi^I\Big(\varphi(x, \sigma), G(x, g)\Big)\Big(\partial^2_A\pi^B\Big(\varphi(x, \sigma), e\Big){J}^A_J(x, e)A^J_\Lambda(x, \sigma, e)-\partial^1_A\pi^B\Big(\varphi(x, \sigma), e\Big)\partial_{\Lambda}\varphi^A(x, \sigma)\Big)
\end{align*}
Connecting all these expansions we finally prove the invariance under change of coordinates of both the conditions:
\begin{align*}
A'^I_\mu(x', \sigma', g')&={\bar{J}}^{\nu}_{\mu}\Big(x'(x)\Big)\Big(\partial^2_A\pi^B\Big(\varphi(x, \sigma), e\Big){J}^A_J(x, e)A^J_\nu(x, \sigma, e)
\\
&\,\,\,\,\,\,-\partial^1_A\pi^B\Big(\varphi(x, \sigma), e\Big)\partial_{\nu}\varphi^A(x, \sigma)\Big)\partial^1_B\pi^I\Big(\varphi(x, \sigma), G(x, g)\Big)
\\
&\,\,\,\,\,\,+{\bar{J}}^{\nu}_{\mu}\Big(x'(x)\Big)\Big({J}^A_J(x, g)\eta^J_\nu(x, g)-{J}^A_\nu(x, g)\Big)\partial^2_A\pi^I\Big(\varphi(x, \sigma), G(x, g)\Big)
\\
&\,\,\,\,\,\,-{\bar{J}}^{\nu}_{\mu}\Big(x'(x)\Big){J}^\Xi_\nu(x, \sigma){\bar{J}}^{\Lambda}_{\Xi}\Big(x'(x), \Sigma^\Theta(x, \sigma)\Big)\Big(\partial^2_A\pi^B\Big(\varphi(x, \sigma), e\Big){J}^A_J(x, e)A^J_\Lambda(x, \sigma, e)
\\
&\,\,\,\,\,\,-\partial^1_A\pi^B\Big(\varphi(x, \sigma), e\Big)\partial_{\Lambda}\varphi^A(x, \sigma)\Big)\partial^1_B\pi^I\Big(\varphi(x, \sigma), G(x, g)\Big)
\\
&=A'^A_\mu(x', \sigma', e')\partial^1_A\pi^I\Big(\varphi(x, \sigma), G(x, g)\Big)+{\eta'}^A_\mu\Big(x', G(x, g)\Big)\partial^2_A\pi^I\Big(\varphi(x, \sigma), G(x, g)\Big)
\\
A'^I_\Theta(x', \sigma', g')&={\bar{J}}^{\Lambda}_{\Theta}\Big(x'(x),\Sigma^\Theta(x, \sigma)\Big)\Big(\partial^2_A\pi^B\Big(\varphi(x, \sigma), e\Big){J}^A_J(x, e)A^J_\Lambda(x, \sigma, e)
\\
&\,\,\,\,\,\,-\partial^1_A\pi^B\Big(\varphi(x, \sigma), e\Big)\partial_{\Lambda}\varphi^A(x, \sigma)\Big)\partial^1_B\pi^I\Big(\varphi(x, \sigma), G(x, g)\Big)
\\
&=A'^A_\Theta(x', \sigma', e')\partial^1_A\pi^I\Big(\varphi(x, \sigma), G(x, g)\Big)
\end{align*}

\section*{Acknowledgements} 

We would like to acknowledge the contribution of the local research project {\em Metodi Geometrici in Fisica Matematica e Applicazioni (2025)}, Department of Mathematics, University of Torino, and of INFN ({\em Iniziativa Specifica QGSKY} and {\em Iniziativa Specifica Euclid}). This paper is also supported by INdAM-GNFM. We thank Francesco Cattafi and Simon-Raphael Fischer for useful comments. We extend our thanks to the anonymous reviewer for helpful remarks and suggestions.

\end{document}